\documentclass[12pt,twoside]{article}
\usepackage[utf8]{inputenc}
\usepackage{paralist}
\usepackage{amsmath, amsthm, amssymb}
\usepackage{authblk}
\usepackage{geometry}
\usepackage{graphicx}
\usepackage[font=normalsize,skip=0pt]{caption}
\usepackage{todonotes}
\usepackage{optidef}
\usepackage[normalem]{ulem} % \sout barre le texte
\newcommand{\cF}{{\mathcal F}}
\newcommand{\E}{{\mathbb E}}

\newcommand{\N}{{\mathbb N}}
\newcommand{\R}{{\mathbb R}}
\newtheorem{theorem}{Theorem}[section]
\newtheorem{prop}{Proposition}[section]
\newtheorem{lemma}{Lemma}[section]

\newtheorem{remark}{Remark}[section]

\title{A stochastic volatility model for the valuation of temperature derivatives}

\begin{document}
\author[1]{Aurélien Alfonsi}
\author[1,2]{Nerea Vadillo}
\affil[1]{CERMICS, Ecole des Ponts, Marne-la-Vall\'ee, France. MathRisk, Inria, Paris,
  France.  }
\affil[2]{AXA Climate, Paris, France. }
\affil[ ]{email: \texttt{aurelien.alfonsi@enpc.fr}, \texttt{nerea.vadillo@axa.com} }
\date{\today}

\maketitle
\begin{abstract}
    This paper develops a new stochastic volatility model for the average daily temperature that is a natural extension of the Ornstein-Uhlenbeck model proposed by Benth and Benth~\cite{benth2007volatility}. This model allows to be more conservative regarding extreme events while keeping tractability. We give a method based on Conditional Least Squares to estimate the parameters on daily data and estimate our model on eight major European cities. We then  show  how to calculate efficiently the average payoff of weather derivatives both by Monte-Carlo and Fourier transform techniques. This new model allows to better assess the risk related to temperature volatility. 
\end{abstract}

\section*{Introduction}

With the increased awareness on climate risk and its tremendous consequences on all economic sectors, the demand for financial tools that enable to hedge this weather-related perils has significantly increased in the last decades. The development of weather derivatives coincides with this trend. First launched in 1997 with over the counter (OTC) contracts, the Chicago Mercantile Exchange (CME) introduced standardised contracts for American cities in 1999. Until date, the open market remains relatively small and therefore lacks of liquidity. However, OTC market has developed between industrialists and large insurance and finance companies leading to the necessity to develop satisfactory risk valuation methodologies for those derivatives. These derivatives usually bring on a Heating Degree Day (HDD) index or a Cooling Degree Day (CDD) index or  a Cumulative Average Temperature (CAT) index. Since most of the deals are OTC, the market is incomplete and thus no arbitrage free pricing. Weather derivative sellers are more interested in understanding the distribution of the payoff, and how this distribution may change under stressed conditions. These information allow them to determine their pricing depending on their risk appetite as well as to assess their maximum losses.

The literature on the valuation of weather derivatives mostly relies on the development of temperature models. 
On the one hand, there are continuous time models. Brody et al.~\cite{brody2002dynamical} suggests to use an Ornstein-Uhlenbeck process driven by a fractional Brownian motion with periodic time-dependent parameters. Benth and Benth~\cite{benth2007volatility} consider the same dynamics with a classical Brownian motion with a further expansion in periodic functions and including a trend for the temperature. Benth et al.~\cite{benth2012critical} proposes a continuous time version of autoregressive models. More recently, Groll et al.~\cite{GLCMB} have developed a continuous time model with factors for the temperature forecasting curve. All these papers use their model to calculate the average payoff of standard weather derivatives on HDD, CDD or CAT. On the other hand, Tol~\cite{Tol} and Franses et al.~\cite{franses2001modeling} have proposed discrete-time GARCH models for the temperature. Cao and Wei~\cite{cao2004weather} and Campbell and Diebold~\cite{campbell2005weather} use this approach in view of pricing derivatives.  Benth and Benth~\cite{benth2012critical}  complete this latter study and show that a low order of autoregression is enough to fit well temperature data. Recently, Meng and Taylor~\cite{MeTa} have proposed an extension of this family of models that gives a joint modeling of the daily minimum and maximum temperature.

In this paper, we focus on the temperature related derivatives which have been, by far, the most studied tools in the literature. Following the studies of Benth {and} Benth \cite{benth2007putting} and many others on temperature dynamics,  we develop a time inhomogeneous affine stochastic volatility model inspired from the celebrated Heston model~\cite{heston1993closed} for equity. This model integrates additional flexibility to the temperature dynamics enabling a better time-continuous modelling of the temperature, while keeping tractability. In particular, it is more conservative regarding extreme events than the corresponding Gaussian model. Following the recent work of Bolyog {and} Pap~\cite{bolyog2019conditional}, we develop a Conditional Least Squares estimation method for its parameters, which is easy to implement. Besides, we propose two different pricing algorithms. The first one based on simulation enables to sample the payoff distribution and then to compute empirically quantities such as the average payoff or quantiles of this distribution. The second one takes advantage of the affine structure and is an adaptation of the Fast Fourier Transform method introduced by Carr and Madan~\cite{carr1999option} to calculate the average payoff of some weather derivatives. Besides, this second method, combined with control variates, enables to reduce considerably the computational time (up to $10^5$). 
We can then easily calculate the sensitivity of the average payoff with respect to the different parameters, including those on the temperature volatility. Thus, an important contribution of our model is to better assess the risk behind the volatility of the temperature. Finally, we compare the results of this approach to price weather derivatives with common business practices.

{This paper can also be leveraged by practitioners. First, our model contributes to a deeper understanding of daily temperature models. This topic is particularly sensitive in the context of global warming and increased extreme events. This paper reviews different approaches to trend modeling and volatility capture. Second, the study describes the step-by-step pricing of temperature derivatives from calibration to computational implementation. This is particularly useful for reproductive purposes and has already been implemented, with some adjustments, in the industrial framework. Finally, the last section provides a comparison with standard business pricing practices to bridge the gap between research and application.}

The paper is organised as follows. The first section presents our model for the average daily temperature. In particular, we explain the rationale behind this model and how it goes beyond the dynamics studied in the literature up to date. The second section focuses on the estimation of the different parameters of our model. We develop a conditional least squares approach and check the robustness of the fitting on simulated data. The third section concentrates on the pricing of weather derivatives. It develops both Monte Carlo and Fast Fourier Transform algorithms and studies the sensitivity to model parameters of the average payoff. It also compares these pricing methods with current business approaches.

\section{Temperature models}

%Interest of understanding temperature models with the increase of temperature derivatives

\subsection{A stochastic volatility model for temperature dynamics}

Temperature dynamics have largely been analysed in the statistics literature and they are of particular interest in the field of weather derivatives. These analysis apply to average daily temperature models that are defined as the average of maximum and minimum daily temperature, i.e. $T = \frac{T_{max} + T_{min}}{2}$. {This choice comes from the insurance contracts that typically use this average for the daily temperature.}

Different models have been suggested for the associated $(T_t)_{t \geq 0}$ process. In the present paper, we investigate the interest of applying a stochastic volatility model for the temperature process $(T_t)_{t \geq 0}$. Our model extends the existing models proposed in the literature while it enables to capture important fluctuations that were not illustrated by former models. Namely, we introduce Model~\eqref{stoch}:
\begin{equation}
    \tag{M}
  \begin{cases}
  T_t&=s(t)+\tilde{T}_t,\\
  d \tilde{T}_t&=- \kappa \tilde{T}_t dt +\sqrt{\zeta_t}(\rho dW_t +\sqrt{1-\rho^2} dZ_t),\\
  d \zeta_t&= - K (\zeta_t- \sigma^2(t))dt +  \eta \sqrt{\zeta_t} dW_t,
  \end{cases}\label{stoch} 
\end{equation}
\noindent where $(W_t)_{t \geq 0}$ and $(Z_t)_{t \geq 0}$ are independent Brownian motions, $\kappa, \eta,K>0$, $\rho \in [-1,1]$, $\sigma^2$ is a nonnegative function and the functions $s$ and $\sigma^2$ will be taken as in~\eqref{s} and~\eqref{sigma}. We will  {denote} $(\mathcal{F}_t)_{t\ge 0}$ the filtration generated by $(W,Z)$, so that the processes $T$ and $\zeta$ are adapted to it. This model integrates three components:

\begin{itemize}
    \item First, the function $s$ represents the trend and seasonality of the average temperature $(T_t)_{t \geq 0}$. This function is deterministic, bounded and continuously differentiable. { In this paper, we will consider the following parametric form
    $$s(t)=\alpha_0 +\beta_0 t+\alpha_1 \sin\left( \frac{2\pi}{365}t \right)+\beta_1 \cos\left( \frac{2\pi}{365}t \right),\ t\ge 0, $$
    see Subsection~\ref{subsubsec_OU} for a discussion on this.}
    \item Second, the detrended and deseasonalised temperature process $(\tilde{T}_t)_{t \geq 0}$ follows a mean-reverting process which enables to include memory effects into the model. The parameter $\kappa$ tunes the mean-reversion speed. 
    \item Third, the volatility of the  $(\tilde{T}_t)_{t \geq 0}$ process  {denoted}  $(\zeta_t)_{t \geq 0}$ follows a Cox-Ingersoll-Ross (CIR) process~\cite{cox1985interest} {with time-dependent parameters~\cite{HW90}}. The process $(\zeta_t)_{t \geq 0}$ includes a seasonal deterministic component $\sigma^2(\cdot)$ which is supposed  deterministic, bounded, continuously differentiable and nonnegative, so that the process $\zeta$ {is well defined and} remains nonnegative {if $\zeta_0\ge 0$}. {In this paper, we will the following parametric form for $\sigma^2$:
    $$ \sigma^2(t)=\gamma_0+ \sum_{k=1}^2 \gamma_k \sin\left( k\frac{2\pi}{365}t \right)+\delta_k \cos\left(k \frac{2\pi}{365}t \right).$$}The parameter $K$ tunes the mean-reversion of the volatility process and $\eta$ corresponds to the volatility of the volatility.
\end{itemize}
\noindent { Here and through the paper, we note $f(t)$ (time $t\ge 0$ in parenthesis) a deterministic function and $F_t$ (time in index) a stochastic process.}

This model is somehow an adaptation of the celebrated Heston model \cite{heston1993closed}. It enables to go beyond Ornstein–Uhlenbeck models  \cite{benth2007volatility} \cite{benth2007putting}   \cite{taylor2004comparison} \cite{franses2001modeling} and is an alternative to GARCH volatility models \cite{taylor2004comparison} \cite{franses2001modeling} \cite{campbell2005weather} \cite{benth2012critical}, while keeping some flexibility as we will show in the next sections.

In the following sections we study daily average temperature data series for 8 major European cities: Stockholm, Paris, Amsterdam, Berlin, Brussels, London, Rome and Madrid. The data series present daily data from January 1st 1980 to December 31st 2020. After removing 29 February of leap years, this gives a time series of 14,965 observations coming from weather stations. {Table~\ref{tab:weather_stations} in Appendix~\ref{appendix:weather_station} summarizes the characteristics of the weather stations.  These time series were extracted from Speedwell, the main historical and settlement data provider for weather derivatives. The company performs quality checks including physical consistency comparison, statistical consistency tests and comparison with neighboring sources and corrects the data if necessary. For instance, gap filling is mentioned for Rome Ciampino. The positive aspect of using these data is that it can be considered as cleaned and reliable and has been used to determine weather derivatives payoff in the past. The downside is that it is private with restricted access.} % and recalibration performed.

\subsection{Some background on temperature models}

\subsubsection{Ornstein–Uhlenbeck models}\label{subsubsec_OU}

Interest on temperature models for weather temperature derivative pricing have arisen in the last years. Even though there exists some alternative modeling \cite{gulpinar2017robust} \cite{elias2014comparison} \cite{schiller2012temperature}, there seems to be an overall agreement on the capacity of Ornstein–Uhlenbeck processes to model daily average temperatures. These models have largely been studied by Benth {and} Benth \cite{benth2007volatility} \cite{benth2007putting} and are written as follows:

\begin{equation}
   \begin{cases}
  T_t &=s(t)+\tilde{T}_t,\\
  \tilde{T}_t&=- \kappa \tilde{T}_t dt +\sigma(t) dW_t,
    \end{cases}\label{T} 
\end{equation}

\noindent where $(W_t)_{t \geq 0}$ is an independent Brownian motion, $s$ a seasonality deterministic function, $\kappa$ a nonnegative mean-reverting parameter and $\sigma^2$ is a deterministic nonnegative function. 

There is no clear agreement on the form of $s$. Different functions have been suggested from  {step} functions \cite{taylor2004comparison} to polynomial functions \cite{franses2001modeling}. However, there is a clear preference for Fourier decomposition (\cite{benth2007putting}, \cite{holleland2020decline}, \cite{benth2012critical}) which gives the following expression
\begin{equation}
   \begin{aligned}
  s(t)& = \alpha_0 + \beta_0 t + \sum_{k=1}^{K_s} \alpha_{k} \sin(\xi_k t) + \sum_{k=1}^{K_s}\beta_{k} \cos(\xi_k t),
    \end{aligned}\label{s} 
\end{equation}
where $\xi = \frac{2 \pi }{365}$, $\xi_k = k \xi$ and $K_s \in \N^*$. Here and through the paper, the time unit is the day.

In this model, $s$ integrates a trend component which enables to introduce climate change phenomenon into our model and a periodic component through trigonometric functions. There is a large discussion on the convenient $K_s \in \N^\star$ to choose as well as on the pertinence to introduce interaction terms \cite{xiaochun}. For simplicity purposes, we take $K_s=1$ which is the most commonly choice (see e.g. \cite{benth2007putting} \cite{benth2012critical}) and shows significance at 5\% confidence level on our experiments.

The different parameters of the Ornstein-Uhlenbeck process can be determined by Conditional Least Squares Estimation (CLSE) developed by Klimko {and} Nelson \cite{klimko1978conditional}, which boils down to minimise 
\begin{equation}
\sum_{i=0}^{N-1} \left( T_{(i+1)\Delta} - \mathbb{E}  [T_{(i + 1)\Delta} | T_{i\Delta} ] \right)^2.
\label{linear_reg_T_OU}
\end{equation} 
Since we have daily data, we mostly use $\Delta=1$ through the paper, unless specified. {The quantity 
\begin{equation}\label{def_residual}
  Res_{i\Delta}=  T_{(i+1)\Delta} - \mathbb{E}  [T_{(i + 1)\Delta} | T_{i\Delta} ]=T_{(i+1)\Delta}-s((i+1)\Delta)-e^{-\kappa \Delta}(T_{i\Delta}-s(i\Delta))
\end{equation}
is called the residual of the regression at time~$i\Delta$. Note that the same formula for the residual holds true for Model~\eqref{stoch}.} {We can equally note that the autoregressive factor $e^{-\kappa \Delta}$ derives from the integration of the Ornstein-Uhlenbeck dynamics which is explicit in Equation~\eqref{integral} in Appendix~\ref{appendix:T}. }{The}
CLS estimators have been studied by Overbeck {and} Ryden \cite{overbeck1997estimation} for the CIR process, Li {and} Ma \cite{LiMa} for the stable CIR process, and  Bolyog {and} Pap \cite{bolyog2019conditional} for Heston-like models. This approach has been used to estimate the parameters of Model~\eqref{stoch}, See Section~\ref{section:Parameter estimation}.

Similarly, $\sigma$ represents the deterministic volatility function. For flexibility and periodicity reasons, it can also be modeled thanks to a Fourier decomposition:

\begin{equation}
   \begin{aligned}
  \sigma^2(t)& = \gamma_0 + \sum_{k=1}^{K_{\sigma^2}} \gamma_{k} \sin(\xi_k t ) + \sum_{k=1}^{K_{\sigma^2}} \delta_{k} \cos(\xi_k t ).
    \end{aligned}\label{sigma} 
\end{equation}
{The coefficients $\gamma$'s and $\delta$'s are assumed to be such that $\sigma^2$ is indeed a nonnegative function, and we also exclude the trivial case $\sigma^2(t)=0$ for all $t\ge 0$. For $K_{\sigma^2}=1$,  a straightforward necessary and sufficient condition for having $\sigma^2$ nonnegative is  $\gamma_0\ge \sqrt{\gamma_1^2+\delta_1^2}$. For $K_{\sigma^2}\ge 2$, a sufficient condition is to assume $\gamma_0\ge \sum_{k=1}^{K_{\sigma^2}} \sqrt{\gamma_k^2+\delta_k^2}$. }
This decomposition is used for example by~\cite{benth2007volatility} or~\cite{MeTa} in a more evolved model. We also consider it in Model~\eqref{stoch} as the function to which the stochastic volatility mean reverts. The annual periodicity is a quite natural feature. Besides, this choice gives a bounded continuous function whose nonnegativity is easy to check, which is required for the definition of Model~\eqref{stoch}. In this paper, $K_{\sigma^2}$ will be taken equal to 2 but $K_{\sigma^2}=1$ will also be considered. 

\subsubsection{Limits of Ornstein–Uhlenbeck models}

Ornstein-Uhlenbeck processes, and their discrete form corresponding to Autoregressive Models, present some important limitations. 

First, they remove any potential long memory effects as today's temperature will only depend on the previous day's. This weakness has been challenged through literature. In particular, Brody {et} al. \cite{brody2002dynamical} suggest the introduction of Fractional Brownian motions. They suppose an Ornstein-Uhlenbeck model driven by a Fractional Brownian motion of Hurst parameter $H$. Such models generate temperature paths with  {$(H- \varepsilon)$-H\"older regularity, $\varepsilon>0$ being an arbitrary small real number}. Following the work of Gatheral and al. \cite{gatheral2018volatility}, we perform an analysis on our data to check whether we could observe such regularity and consider the below metric.

\begin{equation*}
m(q,\Delta) = \frac{1}{\lfloor N/\Delta \rfloor} \sum_{i=1}^{\lfloor N/\Delta \rfloor} \mid \tilde{T}_{(i+1)\Delta} - \tilde{T}_{i\Delta} \mid^q
\end{equation*}
For a $H$-H\"older dynamics, $m(q,\Delta)$ should behave as $\Delta^{qH}$. Figure \ref{fig:H} shows $\log(m(q,\Delta))$ is well approximated by an affine function of $\log(\Delta)$ for different values of $q$. Therefore, we can compute the coefficient of the regression of $\log(m(q,\Delta))$ by $\log(\Delta)$ and estimate $H$, see Table \ref{tab:H}. We observe that this parameter is very close to $0.5$, which corresponds to the regularity of a diffusion driven by a standard Brownian Motion. This justifies why we still consider in Model~\eqref{stoch} a diffusion model.

\begin{figure}[h!]
    \begin{minipage}[b]{0.45\linewidth}
        \centering
        \includegraphics[width=\textwidth]{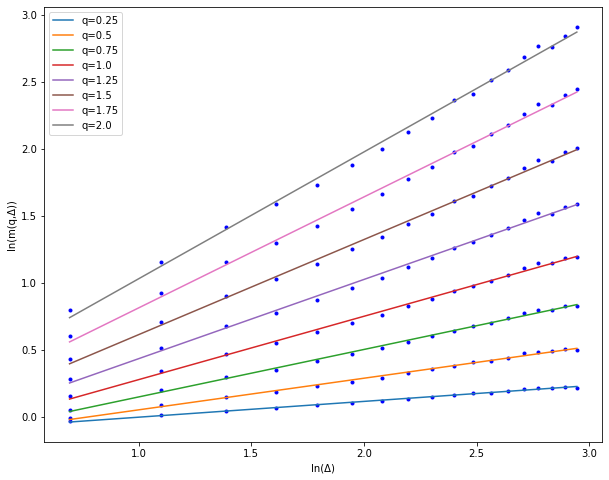}
    \end{minipage}
    \hspace{0.3cm}
    \begin{minipage}[b]{0.45\linewidth}
        \centering
        \includegraphics[width=\textwidth]{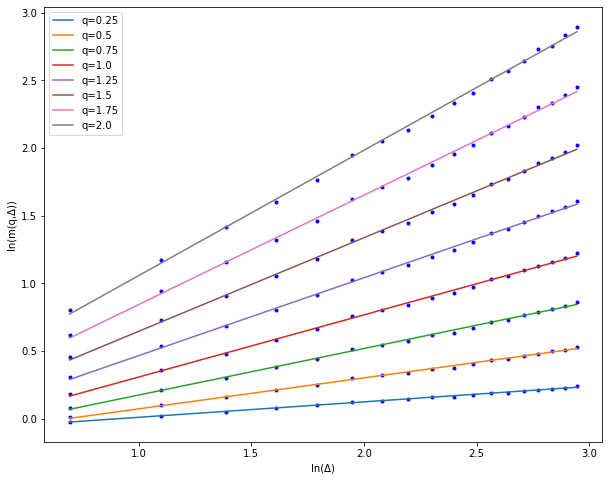}
    \end{minipage}
    \caption{Smoothing plots of the temperature dynamics for Stockholm and Paris (left and right)}
    \label{fig:H}
\end{figure}

\begin{table}[h!]
\centering
\begin{tabular}{l|llllllll}
City & Stockholm & Paris & Amsterdam & Berlin & Brussels & London & Rome & Madrid \\ \hline \hline
$\hat{H}$ & 0.498 & 0.476 & 0.497 & 0.477 & 0.456 & 0.510 & 0.525 & 0.539
\end{tabular}
\caption{Parameter estimations for the Hurst coefficient $H$}
\label{tab:H}
\end{table}

Second, we also contemplated autoregressive models with higher than one autoregressive order. Figure \ref{fig:pacf} shows the partial autocorrelation plots of $(\tilde{T}_t)_{t \geq 0}$ and the residuals   {$(Res_t)_{t \geq 0}$ (see Eq.~\eqref{def_residual})} for the city of Paris in the first $1,000$ observed days. While we can consider a second significant correlation term, its coefficient should be pretty small. Higher than two autoregressive terms are clearly non significant. Therefore, we studied the possibility of having two time delay Ornstein-Uhlenbeck dynamics. However the second order autoregressive coefficients were small and unstable. We hence decided to keep a model with a unique autoregressive terms which is coherent with Franses ad al. \cite{franses2001modeling}, Diebold and Campbell \cite{campbell2005weather} and Taylor and Buizza \cite{taylor2004comparison} findings.

\begin{figure}[ht]
    \begin{minipage}[b]{0.45\linewidth}
        \centering
        \includegraphics[width=\textwidth]{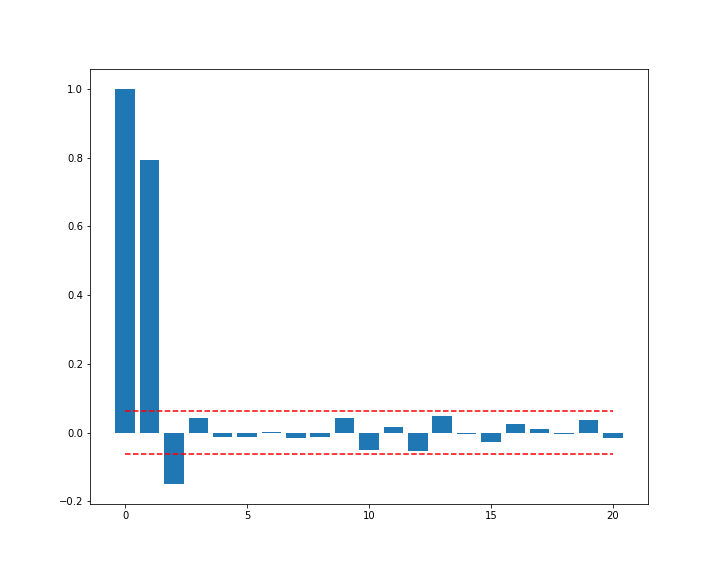}
    \end{minipage}
    \hspace{0.3cm}
    \begin{minipage}[b]{0.45\linewidth}
        \centering
        \includegraphics[width=\textwidth]{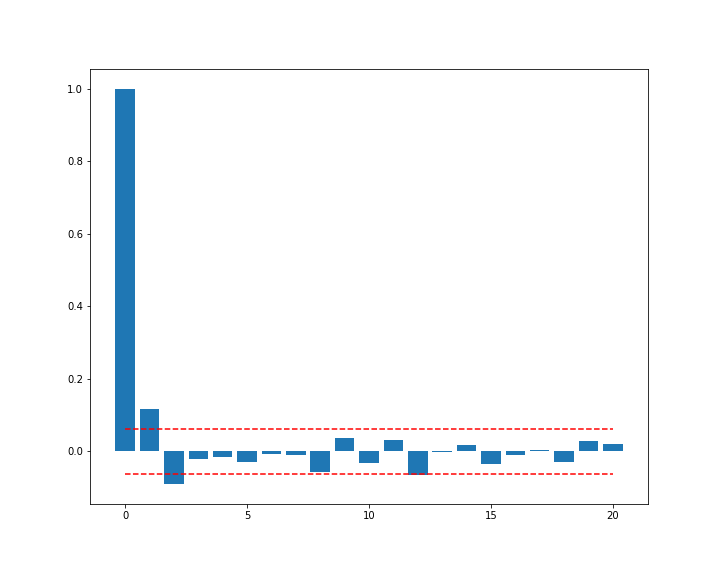}
    \end{minipage}
    \caption{Partial autocorrelation plots of $(\tilde{T}_t)_{t \in \N}$ and {of the residuals $Res_t$ (see Eq.~\eqref{def_residual})} for the city of Paris in the first $1,000$ observations. The dashed red line corresponds to the 95\% confidence interval from which we can consider the partial autocorrelation coefficient is significantly different from 0.} \label{fig:pacf}
\end{figure}

{Third, as can be observed in Figure~\ref{fig:s}, daily temperature present erratic noises with possible volatility autoregression or clustering.}

\begin{figure}[h!]
    \begin{minipage}[b]{0.45\linewidth}
        \centering
        \includegraphics[width=\textwidth]{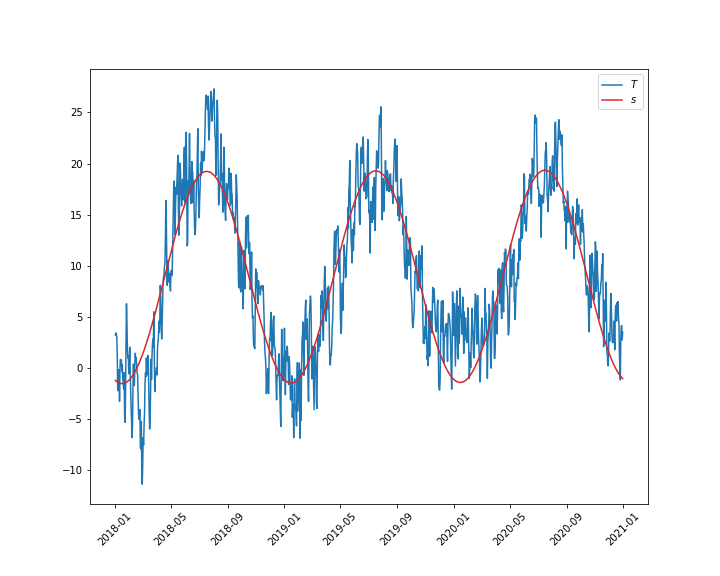}
    \end{minipage}
    \hspace{0.3cm}
    \begin{minipage}[b]{0.45\linewidth}
        \centering
        \includegraphics[width=\textwidth]{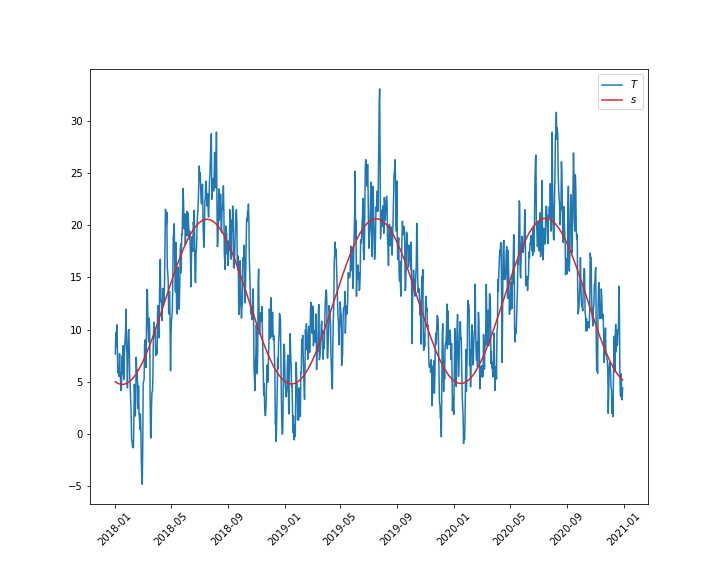}
    \end{minipage}
    \caption{Average temperature and trend on temperature from 2018 to 2020 in Stockholm and Paris (from left to right respectively)}
    \label{fig:s}
\end{figure}
To analyse this, we focus on the residuals~\eqref{def_residual} and compare the historical ones to those obtained on 9~simulations of the same 40 years period. {More precisely, for each simulation, we have plotted the curve $(Res^{sim}_{(i)},Res^{obs}_{(i)})_{1\le i\le 40\times 365}$, where $Res^{sim}_{(i)}$ (resp. $Res^{obs}_{(i)}$) is the ordered statistic of the residuals obtained with the simulated (resp. observed) temperature. In the center of the distribution, the points are on the line $y=x$ in red, which indicates a very good fit by the model. Instead, we see that for small (resp. large) values, the curves are slightly below (resp. above) the red line which indicates that the extreme events produced by the Ornstein-Uhlenbeck model are smaller than the ones observed.} The residuals given by Model~\eqref{T} present noticeable deviations from the historical ones as can be seen in the qqplots: Figure \ref{fig:qqplot_simu_OU} shows slight skewness but the main observation remains that the observed tails are heavier than the simulated ones, especially for the left tail. Different authors emphasise skewness deviation of residuals as well as  {volatility clustering} \cite{franses2001modeling} \cite{xiaochun}. This skewness and heavy tail issues are coped with different methods either thanks to the fitting of a skew-t distribution on residuals \cite{erhardt2013predicting} or through GARCH models that enable to capture dependencies on volatility and can lead to skewed residual qqplots. Here, anticipating on our estimation results, we have plotted in Figure~\ref{fig:qqplot_simu} the qqplot on the residuals between those observed and 9 simulations of Model~\eqref{stoch} with estimated parameters.  {We observe again a very good fit of the center distribution, but the curves are are slightly above (resp. below) the red line $y=x$ for small (resp. large) values, which means that Model~\eqref{stoch} produces larger extreme events than the ones observed. Thus,} compared with the Ornstein-Uhlenbeck model, we note that Model~\eqref{stoch}  produces heavier tails, that are slightly heavier to those observed on our 40 years data set. Thus, Model~\eqref{stoch} is more conservative on extreme events, which is an interesting feature when dealing about risk quantification.

\begin{figure}[h!]
    \begin{minipage}[b]{0.33\linewidth}
        \centering
        \includegraphics[width=1\textwidth]{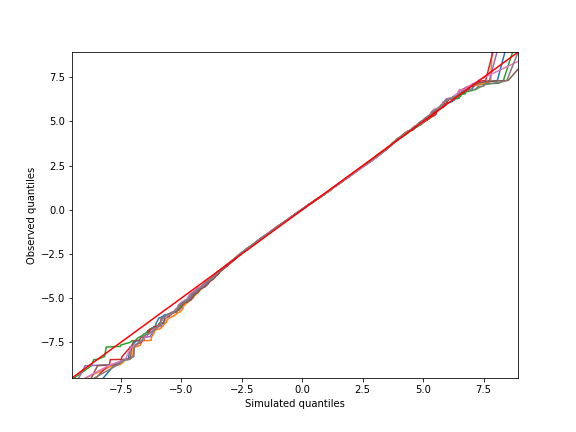}
    \end{minipage}
    \begin{minipage}[b]{0.33\linewidth}
        \centering
        \includegraphics[width=1\textwidth]{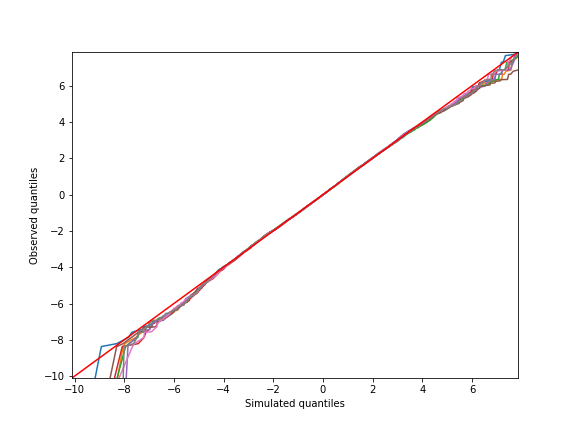}
    \end{minipage}
    \begin{minipage}[b]{0.33\linewidth}
        \centering
        \includegraphics[width=1\textwidth]{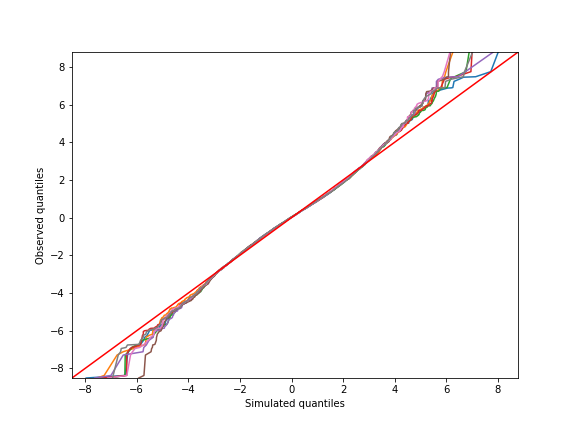}
    \end{minipage}
    \caption{Quantile quantile plots for observed and 9 simulated residuals of~\eqref{linear_reg_T_OU} for Stockholm, Paris and Rome for the Ornstein-Uhlenbeck model. }
    \label{fig:qqplot_simu_OU}
\end{figure}
\begin{figure}[h!]
    \begin{minipage}[b]{0.33\linewidth}
        \centering
        \includegraphics[width=1\textwidth]{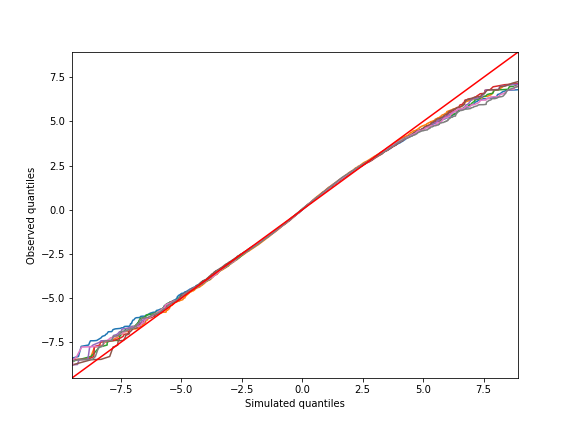}
    \end{minipage}
    \begin{minipage}[b]{0.33\linewidth}
        \centering
        \includegraphics[width=1\textwidth]{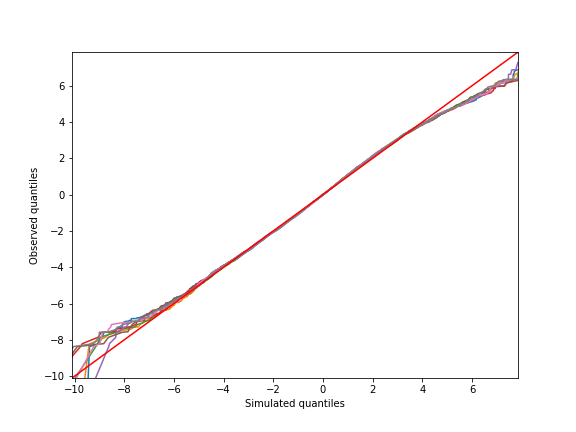}
    \end{minipage}
    \begin{minipage}[b]{0.33\linewidth}
        \centering
        \includegraphics[width=1\textwidth]{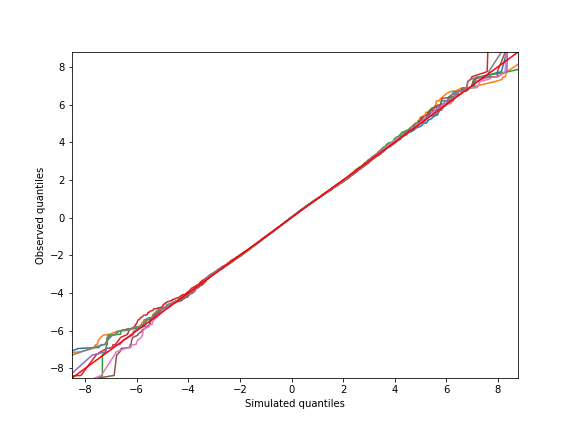}
    \end{minipage}
    \caption{Quantile quantile plots for observed and 9 simulated residuals of~\eqref{linear_reg_T_OU} for Stockholm, Paris and Rome for Model~\eqref{stoch}. }
    \label{fig:qqplot_simu}
\end{figure}

\subsubsection{GARCH models}

The Generalized Autoregressive Conditional Heteroskedasticity (GARCH) models~\cite{bollerslev1986generalized} represent a first response to these limits. They enable to integrate stylized features of temperature times series such as skewness, tail heaviness and volatility clustering.   Let us define the following temperature dynamics for $t\in \N$:

\begin{equation}
   \begin{cases}\label{garch} 
   T_t &= s(t) - \mu(T_{t-1} -s(t-1)) + Z_t\\
 %  s(t) &= \alpha_0 + \beta_1 t + \sum_{k=1}^{K_s} \alpha_{k} \sin(\xi_k t) + \sum_{k=1}^{K_s}\beta_{k} \cos(\xi_k t)  \\
   Z_t &= \sqrt{\zeta_t}\epsilon_t, (\epsilon_t)_{t \geq 0} \sim \mathcal{N}(0,1) \ \text{i.i.d.} \\
   \zeta_t& = \sigma^2(t)  + \sum_{k=1}^{p} \lambda_{k} (Z_{t-k}- \sigma^2(t-k))^2+ \sum_{l=1}^{q} \theta_{l} \zeta_{t-l}, %\\
 %  \sigma(t) &= \gamma_0 + \sum_{k=1}^{K_{\sigma^2}} \gamma_{k} \sin(\xi_k t)  + \sum_{k=1}^{K_{\sigma^2}} \delta_{k} \cos(\xi_k t)  
   \end{cases}
\end{equation}
with {$\lambda_k\ge 0$ for $1\le k\le p$, $\theta_l\ge 0$ for $1\le l\le q$,} and  $s$ and $\sigma^2$ are defined respectively by~\eqref{s} and~\eqref{sigma} { with parameters $\gamma$'s and $\delta$'s such that $\sigma^2(t)\ge 0$ for all $t\ge 0$. Thus, $(\zeta_t)$ is indeed a nonnegative process.} 

Different authors have worked on the application of these GARCH models particularly on GARCH(1,1) to temperature process fitting \cite{xiaochun} \cite{taylor2004comparison}, \cite{franses2001modeling} \cite{benth2012critical} \cite{campbell2005weather}. While these models enable to integrate the complexity of the series we will see that they also present considerable limits.

First, the analysis of autocorrelation and partial autocorrelation plots on squared corrected residuals show evidence of the necessity to integrate first order autocorrelation on the volatility model. This is coherent with the previously cited papers. However, it should be noted that the significance of the first order autocorrelation term remains small even for subsamples. We can therefore expect that it has a limited impact into the quality of multi-day forecasts and pricing.

Second, the GARCH volatility model and its continuous equivalents integrate an important hypotheses: they suppose that temperature and volatility are driven by the same noise (with a time shift). In Model~\eqref{stoch}, this would formally correspond to have $\rho \in \{-1,1\}$. Given the complexity of meteorological dynamics there is no reason explaining why temperature and volatility on temperature should be driven by the same noises. Model~\eqref{stoch} represents a generalization of average temperature dynamics that enables to integrate more flexibility on the temperature's volatility.

Finally, Model~\eqref{stoch} presents a clear advantage as it provides closed form formula for the pricing of temperature indices deriving from average temperature models. While these closed formulas can only be used for some derivatives depending on the payoff structure, they can be combined with Monte Carlo approaches to speed up derivative's pricing. In particular, this document will develop the use of control variates methods to reduce the variance of Monte Carlo approaches, see Subsection~\ref{subsec_control_variates}.
\bigbreak

To sum up, Model~\eqref{stoch} is aligned with the average temperature models that have been considered in the literature. First, from an Ornstein–Uhlenbeck process, it integrates a seasonal component corresponding to the natural climatology and climate change trend. The Ornstein–Uhlenbeck model also enables to include an autoregressive component upon which agree both literature and statistical tools. However, the residuals of this first model show deviation from normal hypotheses with skewness, tail heaviness and volatility clustering patterns. GARCH models partly answer to these limits and integrate an autoregressive component on the volatility process which is statistically observed. Our model is another natural extension to the Ornstein-Uhlenbeck dynamics that gives a larger flexibility on the volatility process. Model \eqref{stoch} presents two additional advantages. First, unlike ARMA and GARCH models it is a time-continuous model. While this might not be necessary when we only have daily data, it can be a significant advantage if this model is coupled with a model for energy or commodities that are traded continuously. We could foresee, for example, combining this model with other commodity models to identify hedging opportunities or price hybrid derivatives. Second, Model~\eqref{stoch} is an affine model for which efficient pricing methods based on Fourier techniques can be implemented.  In particular, we show in Section~\ref{Sec_pricing} how to use the Fast Fourier Transform for pricing, which translates into a significant competitive advantage.

\section{Fitting Model~\eqref{stoch} to historical data}\label{Sec_estim}

The previous section has motivated the interest of the stochastic volatility model~\eqref{stoch} for the temperature dynamics. The pertinence of this new model is however related to our capacity to well estimate its parameters. This section focuses on this challenge as well as on the robustness of the estimation. {Contrary to financial derivatives, we do not consider the calibration of our model to market prices. First, as pointed by Weagley~\cite{weagley2019financial} the market volume of weather derivatives is quite low and the transactions are mostly Over The Counter (OTC). Market prices are thus arguable. Besides, the underlying of these contracts is the real temperature, which is not a traded asset. There is therefore no justification of a risk-neutral pricing of weather derivatives. For these reasons, we prefer to estimate our model to historical data and then use it to determine the distributions of weather derivatives payoffs under the real probability. }

\subsection{Parameter estimation} \label{section:Parameter estimation}

Like Overbeck {and} Ryden~\cite{overbeck1997estimation} and Bolyog {and} Pap \cite{bolyog2019conditional}, we implement the conditional least squares estimation (CLSE) which consists in minimising the observed value against the predicted conditional expectation. Under general assumptions, Klimko {and} Nelson~\cite{klimko1978conditional} have shown that CLS estimators are {strongly consistent}, with a speed of convergence close to $O(N^{-1/2})$, where $N$ is the number of observations. 

It can be noted that Bolyog {and} Pap suggest in~\cite{bolyog2019conditional}, on a time-homogeneous model that is similar to~\eqref{stoch},  to simultaneously minimize the conditional temperature and volatility expectations, i.e. to minimize
$$ \sum_{i=0}^{N-1} \left( T_{(i+1)\Delta} - \mathbb{E}  [T_{(i +1) \Delta} | T_{i\Delta},\zeta_{i\Delta}  ] \right)^2 +\left( \zeta_{(i+1)\Delta} - \mathbb{E}  [\zeta_{(i +1) \Delta} | \zeta_{i\Delta} ] \right)^2, $$
with respect to $\kappa$, $K$, and the (parameterised) functions $s(\cdot)$ and $\sigma^2(\cdot)$.
As already remarked in~\cite{bolyog2019conditional}, the estimators of the parameters $\kappa,K, s(\cdot), \sigma^2(\cdot)$  do not  involve the values of $\eta$ and $\rho$. We can thus do the estimation  without knowing these values, and then estimate separately the parameters $\eta$ and $\rho$.
 Besides, Model \eqref{stoch} corresponds to a special case of the dynamics considered by Bolyog {and} Pap where the volatility is absent from the mean-reverting term of the temperature process (this corresponds to $\beta=0$ in~\cite{bolyog2019conditional}). In this case, $\mathbb{E}  [T_{(i +1) \Delta} | T_{i\Delta},\zeta_{i\Delta}  ] =\mathbb{E}  [T_{(i +1) \Delta} | T_{i\Delta}]$,  and the minimisation problem  is equivalent to minimise $\sum_{i=0}^{N-1} \left( T_{(i+1)\Delta} - \mathbb{E}  [T_{(i +1) \Delta} | T_{i\Delta} ] \right)^2 $ and $\sum_{i=0}^{N-1} \left( \zeta_{(i+1)\Delta} - \mathbb{E}  [\zeta_{(i +1) \Delta} | \zeta_{i\Delta} ] \right)^2$  separately, which we do here. %\todo{comprendre commentaire}
 
 Thus, Paragraph~\ref{par_estim_kps} deals with the estimation of $\kappa$ and $s(\cdot)$, Paragraph~\ref{par_estim_Ksg} brings on the estimation of $K$ and $\sigma^2(\cdot)$, while Paragraph~\ref{par_estim_etarho} focuses on the estimation of  $\eta$ and $\rho$.

\subsubsection{Parameter estimation for $\kappa$ and $s(\cdot)$}\label{par_estim_kps}

 To estimate $\kappa$ and $s(\cdot)$, we are thus interested by the following minimisation problem:

\begin{mini}|s|
{{(\kappa, \alpha_0,\beta_0,\alpha_1, \beta_1)\in \R^5}}{\sum_{i=0}^{N-1} \left( T_{(i+1)\Delta} - \mathbb{E}  [T_{(i +1) \Delta} | T_{i\Delta} ] \right)^2,}{}{}
\label{linear_reg_T}
\end{mini}

\noindent that brings on temperature process  $(T_t)_{t \geq 0}$. Proposition~\ref{prop-ap1} solves this problem and gives explicit formulas for $\hat{\kappa}$, $\hat{\alpha}_0$, $\hat{\beta}_0$, $\hat{\alpha}_1$ and $\hat{\beta}_1$.
Table~\ref{tab:parameter_row1} shows the estimated parameters for the 8 European cities. All the parameters are significant at 5\% confidence level (their 95\% confidence interval do not cross zero) . We can also add that there is a certain coherence between the different European cities for the mean-reverting parameter. In addition, given the range of values close to $0.25$ we can also derive that the temperature memory effect lasts for around 4 days. These results are aligned with the literature, see e.g.~\cite{benth2007putting}. We also notice that the estimated values of $\beta_0$ are around $0.00013$, which corresponds to a warming of about $0.5$°C every 10 years. 

\begin{table}[h!]
\centering
\begin{tabular}{l|lllll}
City & $\hat{\alpha}_0$ & $\hat{\beta}_0$  & $\hat{\alpha}_1$  & $\hat{\beta}_1$ & $\hat{\kappa}$ \\ \hline \hline
Stockholm & 6.678 & 0.00016 & -4.564 & -9.142 & 0.192 \\
Paris & 10.868 & 0.00013 & -3.540 & -6.993 & 0.230 \\
Amsterdam & 9.402 & 0.00013 & -3.509 & -6.426 & 0.228 \\
Berlin & 9.190 & 0.00013 & -3.863 & -8.834 & 0.203 \\
Brussels & 9.746 & 0.00012 & -3.467 & -6.761 & 0.195 \\
London & 10.670 & 0.00011 & -3.345 & -6.035 & 0.260 \\
Rome & 14.826 & 0.00013 & -4.733 & -7.522 & 0.228 \\
Madrid & 13.961 & 0.00010 & -4.572 & -8.608 & 0.221
\end{tabular}
\caption{Parameter estimations for the seasonal function $s$ and mean reverting $\kappa$.}
\label{tab:parameter_row1}
\end{table}

\subsubsection{Parameter estimation for $K$ and $\sigma(\cdot)$}
\label{par_estim_Ksg}

An important challenge we face when estimating the parameters of the process  $(\zeta_t)_{t \geq 0}$ is that the instantaneous volatility process is, per se, unobservable. There is a large literature on this issue, we mention here  \cite{ai2007maximum} and \cite{azencott2020realised} that deal with the Heston model. Following Azencott {et} \textit{al.} approach~\cite{azencott2020realised}, we approximate the unobservable volatility $\zeta$ by the series of realized volatilities $\hat{\zeta}$. These realized volatilities $\hat{\zeta}$ correspond to the observed volatility on a time window of $Q$-days such that we get:
\begin{equation}
   \begin{aligned}
  \hat{\zeta}_{iQ  \Delta} & := \frac{1}{Q } \sum_{j=1}^Q   \frac{2 \hat{\kappa}}{1-e^{-2 \hat{\kappa} \Delta}} \left( \tilde{T}_{(iQ+j)\Delta} - e^{-\hat{\kappa}  \Delta}\tilde{T}_{(iQ+j-1)\Delta} \right)^2, & i \in\{ 0,\dots, \lfloor N/Q\rfloor-1\}.
    \end{aligned}\label{vol} 
\end{equation}
Here, $\hat{\zeta}_{iQ  \Delta}$ corresponds to the realized volatility on $[iQ\Delta, (i+1)Q\Delta]$ and we have thus $I = \lfloor N/Q\rfloor $  different values. {The autoregressive factor comes from the integration of the temperature dynamics of Model~\eqref{stoch} which is explicit in Equation~\eqref{integral} in Appendix~\ref{appendix:T}.} The correction factor $\frac{2 \hat{\kappa}}{1-e^{-2 \hat{\kappa} \Delta}}$, which is not present in~\cite{azencott2020realised}, is related to the mean-reverting behaviour of the temperature and is justified by Remark~\ref{rk_correction}. Since $\frac{2 \hat{\kappa}}{1-e^{-2 \hat{\kappa} \Delta}}\approx_{\Delta \to  0} \frac 1 \Delta$, $\hat{\zeta}_{iQ  \Delta}$ is close to the usual quadratic variation, but the difference is not negligible as $\Delta$ cannot be smaller than one day on observed data. 

\begin{remark}\label{rk_correction}
Let us suppose that $\zeta_t=\zeta_{iQ\Delta}$ for $t\in [iQ\Delta, (i+1)Q\Delta]$. Then, for $j\in\{1,\dots, Q\}$, we have
$$\tilde{T}_{(iQ+j)\Delta} - e^{-\hat{\kappa}  \Delta}\tilde{T}_{(iQ+j-1)\Delta} =\sqrt{\zeta_{iQ\Delta}}\int_{(iQ+j-1)\Delta}^{(iQ+j)\Delta}e^{(iQ+j)\Delta-s}dW^\rho_s, $$

with $W^\rho=\rho W +\sqrt{1-\rho^2} Z$. Since $\int_{(iQ+j-1)\Delta}^{(iQ+j)\Delta}e^{(iQ+j)\Delta-s}dW^\rho_s \sim \mathcal{N}\left(0,\frac{1-e^{-2 \hat{\kappa} \Delta}}{2 \hat{\kappa}} \right)$ is independent of~$\zeta_{iQ\Delta}$, we get that $\frac{1}{ \zeta_{iQ\Delta}}\sum_{j=1}^Q   \frac{2 \hat{\kappa}}{1-e^{-2 \hat{\kappa} \Delta}} \left( \tilde{T}_{(iQ+j)\Delta} - e^{-\hat{\kappa}  \Delta}\tilde{T}_{(iQ+j-1)\Delta} \right)^2$ follows a chi-squared distribution with $Q$ degrees of freedom. Thus, if $\zeta_t$ were frozen for $t\in [iQ\Delta, (i+1)Q\Delta]$,  $\hat{\zeta}_{iQ  \Delta}$ would be an unbiased estimator, i.e. $\E[ \hat{\zeta}_{iQ  \Delta} | \mathcal{F}_{iQ\Delta}]=\zeta_{iQ\Delta}$.
\end{remark}

%First, we can observe that the realized volatility is only defined in the interval $[iQ\Delta, (i+1)Q\Delta]$. It corresponds to an average volatility in an interval of length $Q \Delta$. Therefore our volatility data is not observed daily but within $Q$ days. Second, the averaged components correspond to the rescaled residuals of the regression defined in Equation~\eqref{linear_reg_T}. The rescaling factor is determined following the reasoning in Appendix \ref{appendix:V} where residuals of Equation~\eqref{linear_reg_T} are shown to approximately follow a centered Normal distribution of variance $\frac{1-e^{-2 \kappa \Delta}}{2 \kappa}$. 

\noindent \paragraph{Estimation of $\sigma^2$ and $K$.} 
Once again, we use the conditional least squares method (CLSE) to simultaneously compute the parameters of the deterministic component of the volatility and the mean-reversion coefficient. For this, we would like to minimise \begin{equation}
    \sum_{i=0}^{N-1} \left( \zeta_{(i+1)\Delta} - \mathbb{E}  [\zeta_{(i+1)\Delta} | \zeta_{i\Delta} ] \right)^2 \label{mini:reg_vol}
\end{equation} and use Proposition~\ref{prop-zeta}. The convergence of such kind of estimators for a time inhomogeneous CIR is given by Theorem~\ref{thm_cv_gamma}.  However, since the volatility is not directly observed, we minimise
the difference of the realized volatility and its conditional expectation given by the previously realised volatilities. Namely, we apply Proposition~\ref{prop-zeta} to minimise \begin{equation}\label{mini:reg_vol_Q}
    \sum_{i=0}^{I-2} \left( {\zeta}_{(i+1)Q\Delta} - \mathbb{E}  [\zeta_{(i+1)Q\Delta} | \zeta_{iQ\Delta} ] \right)^2, 
\end{equation}replacing the volatility $\zeta_{iQ \Delta}$ by the realized volatility $\hat{\zeta}_{iQ \Delta}$ . This leads to the following estimators  of the volatility dynamics of Model~\eqref{stoch}:
\begin{equation} \label{simu_zeta}
\begin{dcases}
   \begin{aligned}
  \hat{\gamma}_0 &=  \frac{\hat{\theta}_0 }{1 - \hat{\phi}_0 } \\
  \hat{K}  &=-\frac{1}{Q\Delta}\ln(\hat{\phi}_0) \\
  \hat{\gamma}_k& = \frac{\hat{\theta}_k D_k- \hat{\phi}_k B_k}{A_k D_k- C_k B_k} \\
  \hat{\delta}_k& = \frac{\hat{\theta}_k C_k- \hat{\phi}_k A_k}{C_k B_k- A_k D_k}
    \end{aligned}
\end{dcases}
\end{equation}

\noindent where $k \in \{1,2\}$ and

\begin{equation*}
\begin{dcases}
   \begin{aligned}
  A_k& =  \hat{K} \frac{\hat{K} (\cos(\xi_k Q\Delta)-e^{-\hat{K} Q\Delta}) + \xi_k \sin(\xi_k Q\Delta)}{\hat{K}^2+\xi_k^2}\\
  B_k&=  - \hat{K} \frac{\hat{K} \sin(\xi_k Q\Delta)-\xi_k ( \cos(\xi_k Q\Delta) - e^{-\hat{K} Q\Delta})}{\hat{K}^2+\xi_k^2} \\
  C_k &=  \hat{K} \frac{\hat{K} \sin(\xi_k Q\Delta)-\xi_k ( \cos(\xi_k Q\Delta) - e^{-\hat{K}  Q\Delta})}{\hat{K} ^2+\xi_k^2} \\
  D_k &= \hat{K} \frac{\hat{K} (\cos(\xi_k Q\Delta)-e^{-\hat{K}Q\Delta}) +\xi_k  \sin(\xi_k Q\Delta)}{\hat{K}^2+\xi_k^2},
  \end{aligned}
\end{dcases}
\end{equation*}
and 
\begin{align}\label{def_hat_theta}
  \hat{\vartheta}&:= (\hat{\theta}_0,\hat{\phi}_0,\hat{\theta}_1,\hat{\theta}_2,\hat{\phi}_1,\hat{\phi}_{2})^T =  \left(\sum_{i=0}^{I-2} {\hat{X}'_{i Q\Delta} \hat{X}'^T_{iQ\Delta} } \right)^{-1} \left(\sum_{i=0}^{I-2} \hat{X}'_{iQ\Delta}  \hat{\zeta}_{(i + 1)Q\Delta}\right), \\ \text{ with }
\hat{X}'_{iQ\Delta} &= (1, \hat{\zeta}_{iQ\Delta} , \sin(\xi_1 iQ\Delta), \sin(\xi_2 iQ\Delta), \cos(\xi_1 i Q\Delta), \cos(\xi_2 i Q\Delta))^T.  \notag
\end{align}

Table \ref{tab:parameter_row2}
summarises the numerical implementation of the parameter estimation for our eight cities. We can again observe a coherence between the different cities. It can also be noted that  $\hat{\gamma}_0$ has more importance in the $\sigma^2$ than the trigonometric components. Finally the mean reverting parameter $K$ is more unstable than $\kappa$ through the cities. Its influence can differ from one to 7 days depending on the city.

\begin{table}[h!]
\centering
\begin{tabular}{l|llllll}
City & $\hat{\gamma}_0$ & $\hat{\gamma}_1$  & $\hat{\gamma}_2$  & $\hat{\delta}_1$ & $\hat{\delta}_2$ & $\hat{K}$\\ \hline \hline
Stockholm & 4.790 & 0.684 & -0.450 & 1.401 & 0.704 & 0.147 \\
Paris & 5.603 & 0.201 & -0.266 & 0.358 & 0.459 & 0.396 \\
Amsterdam & 4.690 & 0.503 & -0.524 & 0.500 & 0.604 & 0.335 \\
Berlin & 5.857 & 0.646 & -0.542 & 0.410 & 0.578 & 0.255 \\
Brussels & 4.490 & 0.266 & -0.337 & 0.298 & 0.431 & 0.255 \\
London & 4.387 & -0.014 & -0.314 & 0.479 & 0.174 & 0.774 \\
Rome & 3.086 & 0.212 & -0.391 & 1.335 & 0.373 & 0.332 \\
Madrid & 4.164 & 0.418 & -0.267 & 0.746 & 0.443 & 0.269
\end{tabular}
\caption{Parameter estimations for the seasonal function $\sigma^2$ and volatility mean reverting parameter $K$.}
\label{tab:parameter_row2}
\end{table}

We also tested the choice of setting $K_{\sigma^2}=1$. Table \ref{tab:parameter_row2_one_sin} summarizes the estimation of the different parameters for $K_{\sigma^2}=1$. We can see that the impact on the estimations of $\gamma_0$, $\gamma_1$, $\delta_1$ and $K$ is small or {nill}. Hence, while all the parameters $\gamma_2$ and $\delta_2$ are significant, we find that setting $K_{\sigma^2}=1$ or $2$ has a rather small impact in the reasoning that follows.

\begin{table}[h!]
\centering
\begin{tabular}{l|llll}
City      & $\hat{\gamma}_0$ & $\hat{\gamma}_1$ & $\hat{\delta}_1$ & $\hat{K}$     \\  \hline \hline
Stockholm & 4.790    & 0.674    & 1.408    & 0.137 \\
Paris     & 5.603    & 0.198    & 0.359    & 0.332 \\
Amsterdam & 4.691    & 0.496    & 0.509    & 0.256 \\
Berlin    & 5.857    & 0.644    & 0.416    & 0.226 \\
Brussels  & 4.491    & 0.265    & 0.300    & 0.233 \\
London    & 4.387    & -0.022   & 0.479    & 0.430 \\
Rome      & 3.086    & 0.198    & 1.338    & 0.273 \\
Madrid    & 4.164    & 0.414    & 0.749    & 0.243
\end{tabular}
\caption{Parameter estimations for the seasonal function $\sigma^2$ and mean reverting $K$.}
\label{tab:parameter_row2_one_sin}
\end{table}

In addition, Figure \ref{fig:plot_osbserved_vol} shows the plots corresponding to the observed volatility and estimated seasonality on volatility. We can see that while the seasonality does not seem negligible $\sigma^2$ is far from completely explaining the observed volatility. Indeed, we observe important fluctuations around $\sigma^2(t)$ on the dynamics of $\zeta$.

\begin{figure}[h!]
    \begin{minipage}[b]{0.45\linewidth}
        \centering
        \includegraphics[width=\textwidth]{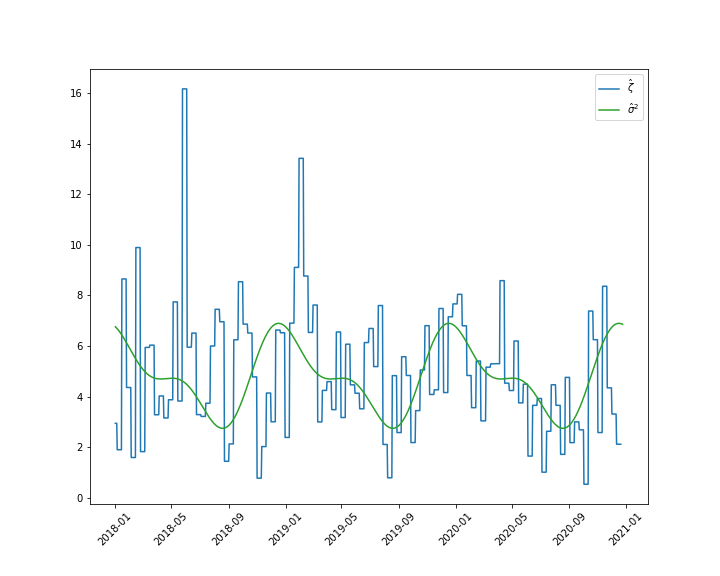}
    \end{minipage}
    \hspace{0.3cm}
    \begin{minipage}[b]{0.45\linewidth}
        \centering
        \includegraphics[width=\textwidth]{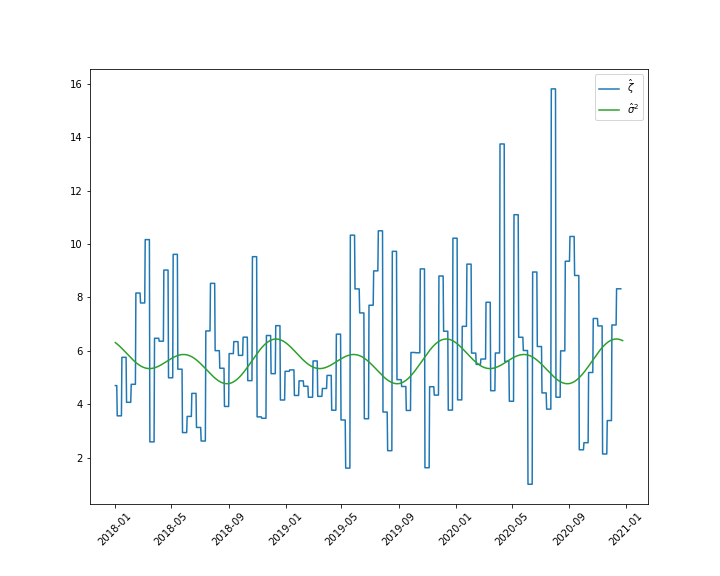}
    \end{minipage}
    \caption{Realized volatility and trend on observed volatility from 2018 to 2020 in Stockholm and Paris (from left to right respectively).} \label{fig:plot_osbserved_vol}
\end{figure}

\subsubsection{Estimation of the parameters $\eta^2$ and $\rho$}\label{par_estim_etarho}

\noindent \paragraph{Estimation of the volatility of volatility $\eta^2$} 
While in the previous section we use the conditional expectation to estimate $\sigma^2$ and $K$, here the idea is to implement a similar approach based on conditional variance.  Namely, we want to solve:

\begin{mini*}|s|  
{\eta^2}{\sum_{i=0}^{N-1} \left( (\zeta_{(i+1)\Delta}- \mathbb{E}  [\zeta_{(i +1) \Delta} | \zeta_{i \Delta} ])^2 - \mathbb{E}  [ \left(\zeta_{(i + 1)\Delta} - \mathbb{E}  [\zeta_{(i +1) \Delta} | \zeta_{i\Delta} ] \right)^2| \zeta_{i\Delta} ]  \right)^2}{}{}
\end{mini*}
Note that Li {and} Ma~\cite{LiMa} or Bolyog {and} Pap \cite{bolyog2019conditional} do not study the properties of such conditional second moment estimators. Here, we show in Theorem~\ref{thm_cv_eta2} 
the convergence of such estimator for a time-dependent CIR process, when the values $\zeta_{i\Delta}$ are directly observed.

\noindent Again, we have to work with the estimated volatility $\hat{\zeta}_{iQ\Delta}$ since we do not directly observe the process $\zeta$. We then apply Proposition~\ref{prop-eta2} with the previously estimated parameters $\hat{K},\hat{\gamma},\hat{\delta}$. This leads to the following estimator 
\begin{equation} \label{formula-eta2}
   \widehat{\eta^2} =  \frac{\sum_{i=0}^{I-2}  \hat{Y}_{i Q\Delta }(\hat{\zeta}_{(i+1)Q\Delta}- \hat{\vartheta}^T \hat{X}'_{iQ\Delta})^2}{\sum_{i=0}^{I-2} \hat{Y}_{iQ\Delta }^2},
\end{equation}

\noindent where $\hat{\vartheta}^T$ and $\hat{X}'_{iQ\Delta}$ are defined by~\eqref{def_hat_theta} and  \begin{equation}\label{def_Ytheta}
\hat{Y}_{iQ\Delta } =  \theta_0' + \phi_0' \hat{\zeta}_{iQ\Delta }  + {\sum_{k=1}^2} \theta_k' \sin(\xi_k iQ\Delta ) + {\sum_{k=1}^2} \phi_k' \cos(\xi_k iQ\Delta )\end{equation} with
\begin{equation*}
\begin{dcases}
   \begin{aligned}
  \theta_0' &= \hat{\gamma}_0 \frac{(1  - e^{- \hat{K} Q \Delta})^2 }{2\hat{K}}\\
  \phi_0' &=  \frac{e^{-\hat{K} Q \Delta}}{\hat{K}}  (1- e^{- \hat{K} Q \Delta}) \\
  \theta'_k &=\hat{\gamma}_k\hat{K} \frac{\hat{K}[A_k'-\phi_0'] +\xi_k B_k'}{\hat{K}^2+\xi_k^2}- \hat{\delta}_k \hat{K} \frac{\hat{K} B_k' -\xi_k [A_k'- \phi_0']}{\hat{K}^2+\xi_k^2} \\
  \phi'_k &= \hat{\gamma}_k \hat{K} \frac{\hat{K} B_k'-\xi_k [A_k'-\phi_0']}{\hat{K}^2+\xi_k^2} +\hat{\delta}_k \hat{K} \frac{\hat{K}[A_k'-\phi_0'] +\xi_k B_k'}{\hat{K}^2+\xi_k^2},   
    \end{aligned}
\end{dcases}
\end{equation*}

\noindent and

\begin{equation*}
\begin{dcases}
   \begin{aligned}
   A_k'&=\frac{2\hat{K}[\cos(\xi_k Q \Delta )-e^{- 2 \hat{K} Q \Delta}] +\xi_k \sin(\xi_k Q \Delta)}{4\hat{K}^2+\xi_k^2}\\
   B_k'&=\frac{2 \hat{K} \sin(\xi_k Q \Delta)-\xi_k( \cos(\xi_k Q \Delta)- e^{- 2 \hat{K} Q \Delta})}{4\hat{K}^2+\xi_k^2}
   %\\   C'&=e^{-\hat{K} Q \Delta} \frac{1-e^{-\hat{K} Q \Delta}}{\hat{K}} 
    \end{aligned}
\end{dcases}
\end{equation*}

\noindent and $(\hat{\gamma}_0, \hat{K}, \hat{\gamma}_k, \hat{\delta}_k )$  defined by Equation~\eqref{simu_zeta}.

%\noindent Finally, as (\ref{mini:min-eta22}) has a unique regressor and no intercept, the estimator of the $\eta^2$ corresponds to:

%\noindent where $\hat{\theta} X'_{(i+1)Q\Delta}$ is the estimate of $\mathbb{E}  [\hat{\zeta}_{(i +1) Q \Delta} | \hat{\zeta}_{i Q\Delta} ]$. 

\begin{table}[h!]
\centering
\begin{tabular}{l|llllllll}
city & Stockholm & Paris & Amsterdam & Berlin & Brussels & London & Rome & Madrid \\ \hline \hline
$\widehat{\eta^2}$ & 0.629 & 1.043 & 0.929 & 0.884 & 0.713 & 1.605 & 0.988 & 0.737
\end{tabular} 
\caption{Parameter estimations for the volatility of volatility $\eta^2$.}\label{tab:eta2}
\end{table} 

The numerical application of the above formula applied to our dataset gives the values collected in Table \ref{tab:eta2}. First, comparing Tables \ref{tab:parameter_row2} and \ref{tab:eta2} enables to observe the two components of the volatility dynamics, the mean reversion component and the pure noise. They both have comparable magnitude and therefore both contribute significantly to the volatility dynamics. The mean reverting coefficients are rather small, which is consistent with the fluctuations observed in Figure~\ref{fig:plot_osbserved_vol}. Second, Table \ref{tab:eta2} shows a certain coherence in terms of magnitude for all the European cities.

\noindent \paragraph{Estimation of the correlation $\rho$}
The last step consists in estimating the correlation $\rho$. The idea this time is to use conditional covariance to estimate this parameter, and Proposition~\ref{prop-rho} gives the minimiser of the following problem:
\begin{equation}
  \begin{aligned}
    \min_\rho \sum_{i=0}^{N-1}  &  \Bigl( (T_{(i+1)\Delta}- \mathbb{E}  [T_{(i +1) \Delta} |  \mathcal{F}_{i\Delta} ])(\zeta_{(i+1)\Delta}- \mathbb{E}  [\zeta_{(i +1) \Delta} |  \mathcal{F}_{i\Delta} ]) \\
    &- \mathbb{E}  \Big[ (T_{(i+1)\Delta}- \mathbb{E}  [T_{(i +1) \Delta} |  \mathcal{F}_{i\Delta} ])(\zeta_{(i+1)\Delta}- \mathbb{E}  [\zeta_{(i +1) \Delta} |  \mathcal{F}_{i\Delta} ]  \Big|  \mathcal{F}_{i\Delta} \Big]  \Bigr)^2.
  \end{aligned}  
\end{equation}

\noindent Again, since we do not observe the volatility, we use  Proposition~\ref{prop-rho} with the estimated volatility $\hat{\zeta}_{iQ\Delta}$ and the previously estimated parameters ${\hat{\kappa}},\hat{\alpha},\hat{\beta},{\hat{K}},\hat{\gamma},\hat{\delta},\widehat{\eta^2}$ given by Proposition~\ref{prop-ap1}, \eqref{simu_zeta} and~\eqref{formula-eta2}. This leads to
\begin{equation}
   \hat{\rho} =  \frac{\sum_{i=0}^{I-2} \hat{Y}'_{i Q \Delta} (T_{(i+1) Q\Delta}- \hat{\lambda}^T X_{i Q\Delta})(\hat{\zeta}_{(i+1) Q \Delta}- \hat{\vartheta}^T \hat{X}'_{iQ\Delta})}{\sum_{i=0}^{I-2} (\hat{Y}'_{i  Q\Delta})^2},
\end{equation}
\noindent where $\hat{\lambda}$ and $X_{i\Delta}$ are defined by~\eqref{def_hlambda} { with $N:= I-1$}, $\hat{\vartheta}$ and $\hat{X}'_{iQ\Delta}$ are defined by~\eqref{def_hat_theta}, and  $\hat{Y}'_{i  Q\Delta} =  \theta_0'' + \phi_0'' \hat{\zeta}_{iQ\Delta} + \sum_k \theta_k'' \sin(\xi_k i Q\Delta) + \sum_k \phi_k'' \cos(\xi_k i Q \Delta)$, with

\begin{equation*}
\begin{dcases}
   \begin{aligned}
  \theta''_0 &= \hat{\eta} \hat{\gamma}_0  \left( \frac{1 - e^{- (\hat{\kappa}+\hat{K}) Q \Delta} }{\hat{\kappa}+\hat{K}}+ \frac{e^{-(\hat{\kappa}+\hat{K})  Q \Delta}  - e^{-\hat{K}Q \Delta}  }{\hat{\kappa}} \right) \\
  \phi_0'' &= \hat{\eta}  e^{-\hat{K}  Q \Delta} \frac{1 - e^{-\hat{\kappa} Q \Delta} }{\hat{\kappa}} \\
  \theta''_k&=  \hat{\eta} \hat{\gamma}_k \hat{K} \frac{\hat{K}(A_k''-\phi_0'')  +  \xi_k B_k''}{\hat{K}^2+\xi_k^2}-  \hat{\eta} \hat{\delta}_k \hat{K} \frac{\hat{K} B_k'' -\xi_k (A_k''-\phi_0'')}{\hat{K}^2+\xi_k^2} \\
  \phi''_k&=   \hat{\eta} \hat{\gamma}_k \hat{K} \frac{\hat{K}B_k''-\xi_k (A_k''-\phi_0'')}{\hat{K}^2+\xi_k^2} +  \hat{\eta} \hat{\delta}_k \hat{K} \frac{\hat{K}(A_k''-\phi_0'') +\xi_k B_k''}{\hat{K}^2+\xi_k^2},
    \end{aligned}
\end{dcases}
\end{equation*}

and

\begin{equation*}
\begin{dcases}
   \begin{aligned}
  A_k''&=\frac{(\hat{K}+\hat{\kappa})(\cos(\xi_k Q \Delta)-e^{- (\hat{K}+\hat{\kappa})Q \Delta}) +\xi_k \sin(\xi_k Q \Delta)}{(\hat{K}+\hat{\kappa})^2+\xi_k^2}\\
  B_k''&=\frac{(\hat{K}+\hat{\kappa})\sin(\xi_k Q \Delta)-\xi_k (\cos(\xi_k Q \Delta)-e^{- (\hat{K}+\hat{\kappa})Q \Delta}) }{(\hat{K}+\hat{\kappa})^2+\xi_k^2}.
  %\\  C''&=e^{-\hat{K}Q \Delta} \frac{1-e^{-\hat{\kappa} Q \Delta}}{\hat{\kappa}},
  \end{aligned}
\end{dcases}
\end{equation*}

\begin{table}[h!]
\centering
\begin{tabular}{l|llllllll}
City & Stockholm & Paris & Amsterdam & Berlin & Brussels & London & Rome & Madrid \\ \hline \hline
$\hat{\rho}$ & -0.000 & -0.006 & -0.010 & -0.013 & -0.014 & -0.011 & 0.023 & -0.005
\end{tabular}
\caption{Parameter estimations for the correlation $\rho$.}
\label{tab:rho}
\end{table}

The numerical application of the above formula applied to our data sets gives the values collected in Table \ref{tab:rho}. We observe that the correlation is close to zero for all the cities. Therefore, for simplification purposes, we will consider on the following of this document that $\rho$ equals zero. This finding also questions the pertinence of GARCH model~\eqref{garch}  that corresponds to $\rho \in \{-1,1\}$ since the temperature and its volatility are driven by the same noise.

\subsection{Robustness of estimators}

In the previous subsection, we have obtained conditional least squares estimators for the different parameters of Model~\eqref{stoch}. All the above expressions have been computed by discretizing the processes $(\tilde{T}_t)_{t \geq 0}$ and $(\zeta_t)_{t \geq 0}$. 
Theoretically speaking, Overbeck {and} Ryden~\cite{overbeck1997estimation} and Bolyog {and} Pap ~\cite{bolyog2019conditional} have proven the convergence of the CLS estimators for the Cox-Ingersoll-Ross and a generalized Heston model. Their proof is mainly based in ergodic arguments. In Appendix~\ref{appendix:conv-CIR}, we prove the  {consistency} of CLSE for a time inhomogeneous Cox-Ingersoll-Ross process. 

However, to estimate the parameters of the volatility dynamics, we have approximated the unobservable volatility $\zeta$ by the realized volatility $\hat{\zeta}$. Azencott,  Ren and Timofeyev \cite{azencott2020realised} have deeply studied the convergence modes of the volatility process $\hat{\zeta}$  to the instantaneous process $\zeta$ as well as the estimated realized estimators $(K, \sigma, \eta)$ under a classic Heston framework. Under boundary and continuity hypothesis on $T$ and $\zeta$, uniform convergence of $\hat{\zeta}$  to $\zeta$ over $[0,T]$ in $L^2$ is proven. Probability convergence of estimators is also proven for moments based estimators. The extension of these convergence properties to our particular model is left as a further work. In this section, we test numerically the robustness of our estimators and check their accuracy on simulated data.

\subsubsection{Methodology}

{Robustness of the estimators is checked through simulated data. Essentially, we simulate data series with the model that have the same length as our data set (40 years) and we check that we find back the parameters by using the CLS estimators. The detailed methodology is presented in Appendix~\ref{appendix:method_simu}.}

\begin{figure}[h!]
    \begin{minipage}[b]{0.5\linewidth}
        \centering
        \includegraphics[width=\textwidth]{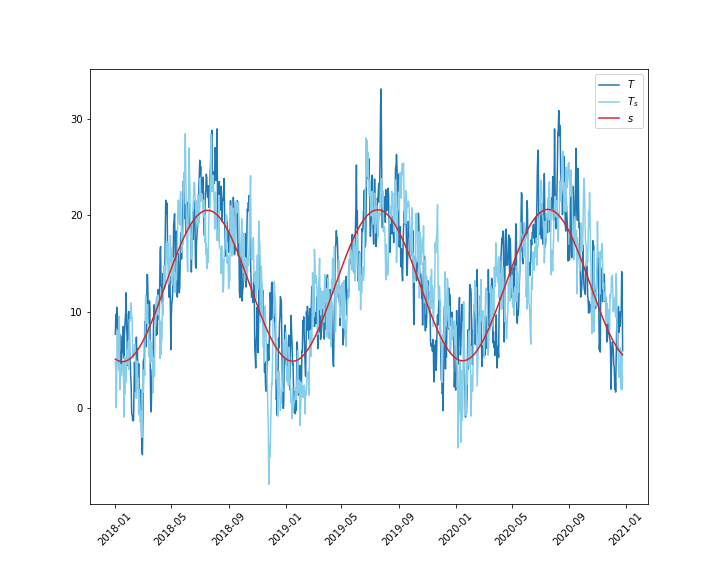}
    \end{minipage}
    \hspace{0.2cm}
    \begin{minipage}[b]{0.5\linewidth}
        \centering
        \includegraphics[width=\textwidth]{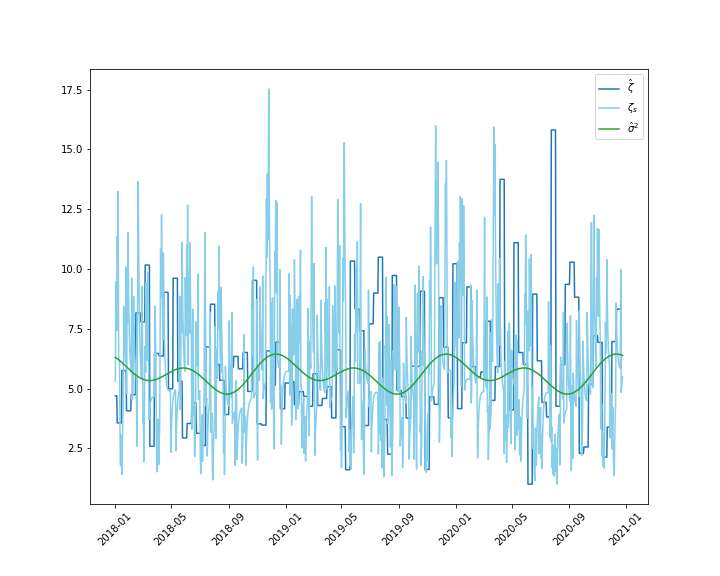}
    \end{minipage}
    \caption{Plots of simulated temperature and volatility processes for Paris and $(K,\eta^2)=(0.396,1.043)$. On the left, we plot the observed temperature $T$ (blue), the simulated temperature $T_s$ (light blue) and the trend and seasonal function $s$. On the right, we plot the observed volatility $\hat{\zeta}$ (blue), defined here as the 10-lag moving average of $\zeta$, the simulated volatility $\zeta_s$ (light blue) and the seasonal volatility function $\sigma^2$.}\label{fig:simu}
\end{figure}

As a first qualitative check, we have represented in Figure \ref{fig:simu} an example of simulated temperature for Paris. The simulated paths looks similar to the observed ones. 
In the following paragraphs, we will test our capacity to estimate the values of the parameters and discuss the choice of averaging-time windows $Q$.

\subsubsection{Estimation of parameters $\kappa,s(\cdot)$ related to the simulated temperature $T$ }

The estimation of $\kappa$ and $s(\cdot)$ is a priori easy since it relies on the temperature that is directly observable. Table \ref{tab:simuT} summarizes the estimators of the parameters related to the temperature. Figure \ref{fig:simu_T_estim} corresponds to the related temperature plot for Stockholm and Paris. We can see that all the estimated parameters remain very close to the original values. We can thus  conclude that the estimation of the parameters of the temperature is robust enough to be reliable.

\begin{table}[h!]
\centering
\begin{tabular}{r|rr|rr|rr|rr|rr}
City & $\alpha_0$ & $\hat{\alpha}_0$ & $\beta_0$ & $\hat{\beta}_0$ & $\alpha_1$ & $\hat{\alpha}_1$ & $\beta_1$ & $\hat{\beta}_1$  & $\kappa$ & $\hat{\kappa}$ \\ \hline \hline
Stockholm & 6.678 & 6.966 & 0.00016 & 0.00016 & -4.564 & -4.564 & -9.142 & -9.142 & 0.192 & 0.192 \\
Paris & 10.868 & 10.733 & 0.00013 & 0.00013 & -3.540 & -3.540 & -6.993 & -6.993 & 0.230 & 0.235 \\
Amsterdam & 9.402 & 9.353 & 0.00013 & 0.00013 & -3.509 & -3.509 & -6.426 & -6.426 & 0.228 & 0.220 \\
Berlin & 9.190 & 9.472 & 0.00013 & 0.00013 & -3.863 & -3.863 & -8.834 & -8.834 & 0.203 & 0.200 \\
Brussels & 9.746 & 9.350 & 0.00012 & 0.00012 & -3.467 & -3.467 & -6.761 & -6.761 & 0.195 & 0.192 \\
London & 10.670 & 10.659 & 0.00011 & 0.00011 & -3.345 & -3.345 & -6.035 & -6.035 & 0.260 & 0.270 \\
Rome & 14.826 & 14.751 & 0.00013 & 0.00013 & -4.733 & -4.733 & -7.522 & -7.522 & 0.228 & 0.224 \\
Madrid & 13.961 & 13.763 & 0.00010 & 0.00010 & -4.572 & -4.572 & -8.608 & -8.608 & 0.221 & 0.232
\end{tabular}
\caption{Estimation of temperature parameters from the simulated temperature path.}
\label{tab:simuT}
\end{table}

\begin{table}[h!]
\centering
\begin{tabular}{r|rr|rr|rr|rr|rr}
City & $\alpha_0$ & $\hat{\alpha}_0$ & $\beta_0$ & $\hat{\beta}_0$ & $\alpha_1$ & $\hat{\alpha}_1$ & $\beta_1$ & $\hat{\beta}_1$  & $\kappa$ & $\hat{\kappa}$ \\ \hline \hline
Stockholm & 6.678 & 6.939 &0.00016 &0.00015 &-4.564 &-5.492 &-9.142 &-8.355 &0.192 &0.185 \\
Paris & 10.868 &10.712 &0.00013 &0.00015 &-3.54 &-4.276 &-6.993 &-6.577 &0.230 &0.229 \\
Amsterdam & 9.402 &9.647 &0.00013 &0.00009 &-3.509 &-4.110 &-6.426 &-6.107 &0.228 &0.225 \\
Berlin & 9.190 &9.429 &0.00013 &0.00011 &-3.863 &-4.804 &-8.834 &-8.325 &0.203 &0.197 \\
Brussels & 9.746 &10.011 &0.00012 &0.0001 &-3.467 &-3.954 &-6.761 &-6.232 &0.195 &0.195 \\
London & 10.670 &10.830 &0.00011 &0.0001 &-3.345 &-4.059 &-6.035 &-5.439 &0.260 &0.259 \\
Rome &14.826 &14.781 &0.00013 &0.00013 &-4.733 &-5.522 &-7.522 &-6.933 &0.228 &0.235 \\
Madrid & 13.961 &14.081 &0.00010 &0.00008 &-4.572 &-5.377 &-8.608 &-7.95 &0.221 &0.224
\end{tabular}
\caption{{Estimation of temperature parameters from the simulated temperature path.}}
\label{tab:simuT}
\end{table}

\begin{figure}[h!]
    \begin{minipage}[b]{0.5\linewidth}
        \centering
        \includegraphics[width=\textwidth]{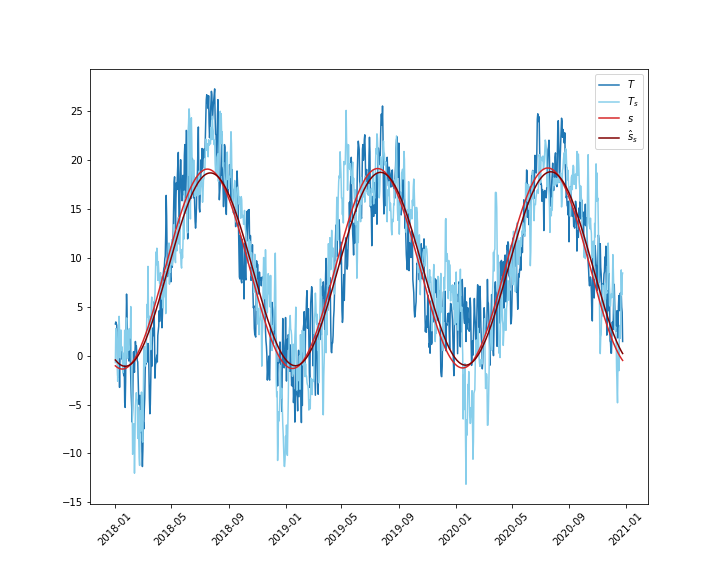}
    \end{minipage}
    \hspace{0.2cm}
    \begin{minipage}[b]{0.5\linewidth}
        \centering
        \includegraphics[width=\textwidth]{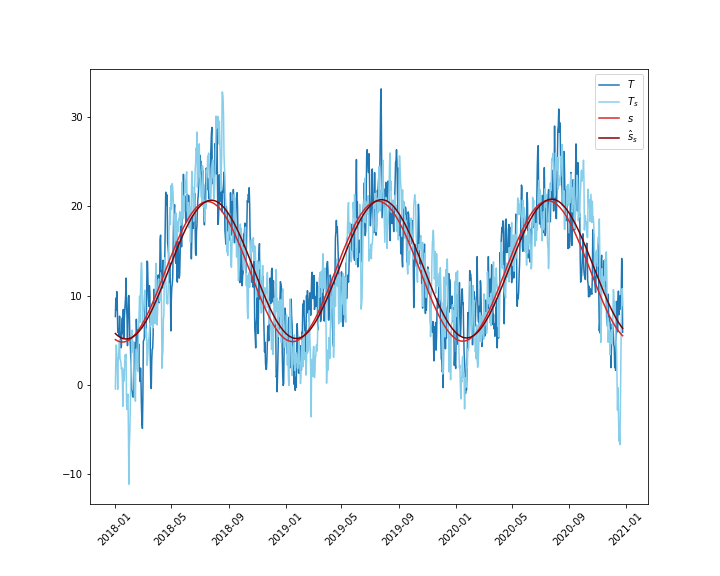}
    \end{minipage}
    \caption{Estimation of the trend and seasonal function $s$ on both real temperature $T$ (blue) and simulated temperature $T_s$ (light blue) for Stockholm (left) and Paris (right) and $Q=10$. This function is plotted in red (resp. brown) when estimated on observed (resp. simulated) data. }
    \label{fig:simu_T_estim}
\end{figure}

\subsubsection{Estimation of $K$, $\sigma^2(\cdot)$ and $\eta^2$ on the simulated realized volatility $\hat{\zeta}$ }

Let us recall that instantaneous volatility $\zeta$ is not observable and we estimate it by $\hat{\zeta}$  defined in Equation~\eqref{vol}.  In this paragraph, we apply the same process as in Section \ref{section:Parameter estimation} to simulated data.

We first focus on the effect of $Q$, i.e. the size of the averaging window~\eqref{vol}, on the estimation of parameters. For simplicity purposes, we start by setting the trigonometric coefficients of $\sigma^2$ to $0$ and study the impact on $K$ and $\eta^2$.  For each time window $Q$ and city, we perform $50,000$ simulations of daily volatility $\zeta$, by using~\eqref{schemes}. We average this series trough the time window $Q$ and then estimate the corresponding $(K, \eta^2)$. This exercise enables to analyse the influence of the time window $Q$ on the capacity to well estimate $(K, \eta^2)$.

\begin{table}[h!]
\centering
\begin{tabular}{r|rrrrrrrrrrrrrrr}
City & $\hat{K}$ & $\hat{K}_{Q=1}$ & $\hat{K}_{Q=2}$ & $\hat{K}_{Q=5}$ &  $\hat{K}_{Q=8}$ &  $\hat{K}_{Q=10}$  & $\hat{K}_{Q=12}$ \\ \hline \hline
Stockholm & 0.147 & 2.261 & 0.886  & 0.301  & 0.190& 0.157  & 0.140 \\
Paris & 0.396 & 2.853 & 1.336 & 0.552 & 0.403 & 0.286  & 0.265 \\
Amsterdam & 0.335 & 2.578 & 1.159  & 0.463  & 0.345 & 0.266  & 0.260  \\
Berlin & 0.255 & 2.590 & 1.083 & 0.408 & 0.278 & 0.243  & 0.208  \\
Brussels & 0.255 & 2.540 & 1.042  & 0.401  & 0.262  & 0.228 & 0.208  \\
London & 0.774 & 3.363 & 1.637  & 0.880 & 0.651  & 0.459  & 0.464  \\
Rome & 0.332 & 2.381 & 1.059  & 0.433 & 0.303  & 0.264 & 0.216  \\
Madrid & 0.269 & 2.495 & 1.067  & 0.407 & 0.278 & 0.260 & 0.219 
\end{tabular}
\caption{Estimation of $K$ for different averaging time windows $Q$.}
\label{tab:K_robust}
\end{table}

\begin{table}[h!]
\centering
\begin{tabular}{r|rrrrrrrrrrrrrrr}
City & $\widehat{\eta^2}$ & $\widehat{\eta^2}_{Q=1}$ & $\widehat{\eta^2}_{Q=2}$  & $\widehat{\eta^2}_{Q=5}$ &  $\widehat{\eta^2}_{Q=8}$ & $\widehat{\eta^2}_{Q=10}$ & $\widehat{\eta^2}_{Q=12}$ \\ \hline \hline
Stockholm & 0.629 & 56.229 & 12.288  & 2.123 & 0.896 & 0.644  & 0.499 \\
Paris & 1.043 & 56.429 & 13.385  & 2.506 & 1.156 & 0.690  & 0.531 \\
Amsterdam & 0.929 & 46.609 & 11.364 & 2.013& 1.010  & 0.625  & 0.528  \\
Berlin & 0.884 & 61.162 & 13.930 & 2.377 & 1.070 & 0.795  & 0.580  \\
Brussels & 0.713 & 47.390 & 10.652  & 1.921  & 0.834 & 0.599 & 0.474  \\
London & 1.605 & 45.348 & 11.535 & 2.608 & 1.236 & 0.723 & 0.588  \\
Rome & 0.988 & 35.879 & 8.962  & 1.686  & 0.787  & 0.561  & 0.407  \\
Madrid & 0.737 & 43.787 & 10.316  & 1.862  & 0.809  & 0.657  & 0.430 
\end{tabular}
\caption{Estimation of $\eta^2$ for different averaging time windows $Q$}
\label{tab:eta_robust}
\end{table}

Table \ref{tab:K_robust} and \ref{tab:eta_robust} represent different estimates of $(K, \eta^2)$ depending on window width $Q$. First, we can observe that for small $Q$,  $(K, \eta^2)$  is overestimated due to the preponderance of the noise related to volatility estimation. On the contrary,  for large values of $Q$,  $(K, \eta^2)$  is underestimated due to the averaging effect. Therefore there is a need to find a trade-off between the effect of the noise and the effect of the window averaging. Depending on the city the most efficient $Q$ can range from 5 to 10. 

We now focus on the impact of~$Q$ on the estimation of the function $\sigma^2$. We have plotted in Figure \ref{fig:simuvol} the estimated trend with the original one for $Q=5$ and $Q=12$. We see that $\sigma^2$ is correctly estimated, independently of $Q$.  %Heuristically, this can be understood by the fact that the CLSE of  $\sigma^2(\cdot)$ aims at minimising the noise..
From this numerical study, the choice of $Q=10$ appears to be quite reliable. Of course, as one may expect, the parameters are not as well estimated as for $\kappa$ and $s(\cdot)$. They still however give the correct magnitude of the parameters, which is acceptable for risk management. Since we are then interested in evaluating derivative products, we will then analyse the sensitivity to these parameters, see Section~\ref{subsec_sensi}, which can be done efficiently in    Model~\eqref{stoch}.

\begin{figure}[h!]
    \begin{minipage}[b]{0.5\linewidth}
        \centering
        \includegraphics[width=\textwidth]{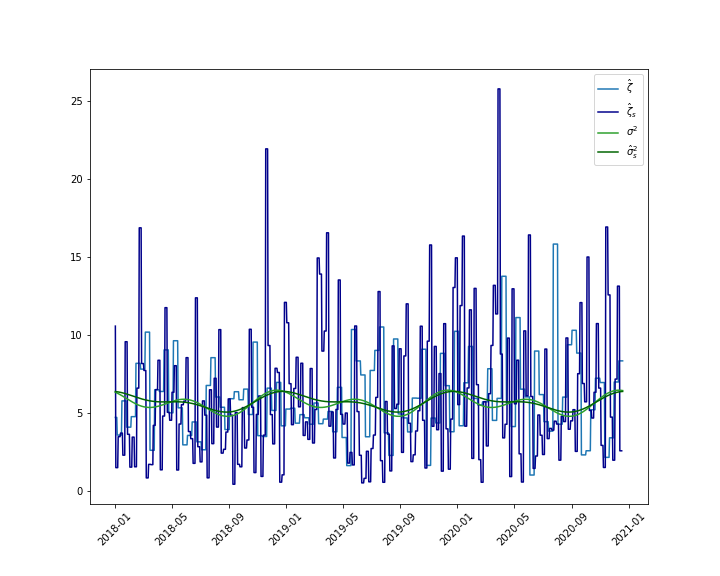}
    \end{minipage}
    \begin{minipage}[b]{0.5\linewidth}
        \centering
        \includegraphics[width=\textwidth]{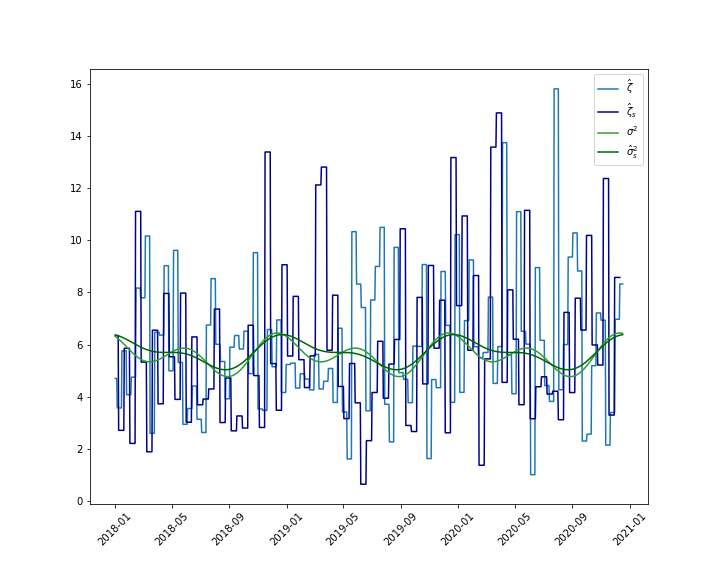}
    \end{minipage}
    \caption{Plots of observed volatility process $\hat{\zeta}$ (blue) and simulated volatility processes $\hat{\zeta}_s$ (dark blue) for Paris for averaging windows $Q$ equals 5 (left) and 12 (right). The function $\sigma^2$ is in green while the estimated functions for $Q=5$ and $Q=12$ are in dark green.}
    \label{fig:simuvol}
\end{figure}

\section{Application to pricing weather derivatives}\label{Sec_pricing}

The previous sections concentrate on modeling the daily average temperature and on the estimation of the parameters of the model. However, the final objective of our model is to better assess the risk related to weather temperature derivatives. This section will focus on how we evaluate the average payoff of these derivatives and how Model~\eqref{stoch} improves our capacity to understand their risk.

\subsection{Temperature derivatives}

\subsubsection{Average temperature indices}

Temperature derivatives are financial products used to hedge weather risk. The covers are often based on an index corresponding to a proxy of the buyer's financial risk. This index corresponds to an aggregate of a more granular meteorological parameter which in this document is the average daily temperature. There exist different possible indices, the main ones being HDD (Heating Degree Days), CDD (Cooling Degree Days) and CAT (Cumulative Average Temperature):

\begin{equation*}
   \begin{aligned}
  HDD & := \sum_{t=t_1}^{t_2} \max(0, T_b - T_t), \   CDD  := \sum_{t=t_1}^{t_2}  \max(0, T_t - T_b), \   CAT  := \sum_{t=t_1}^{t_2}  T_t,
    \end{aligned}\label{index} 
\end{equation*}

\noindent where $T_t$ corresponds to the average daily temperature on day $t$, $T_b$ to a base temperature, $t_1$ to the inception date and $t_2$ the exit date of the contract. We will call risk period the time period between $t_1$ and~$t_2$.

Physically speaking, the HDD corresponds to the cumulative degrees needed to heat a given building. Hence the base temperature $T_b$ corresponds to the temperature at which heating is probably switched on and the cumulative HDD measures the demand on energy of the building. The base temperature $T_b$ varies depending on the country. In the EU, $T_b$ is often taken as equal to $15.5^{\circ} C$ while in the US it corresponds to $65^{\circ} F$. Symmetrically, CDD corresponds to the energy demand for air conditioning. Finally, CAT corresponds to cumulative temperature degrees which is related to the energy demand between $t_1$ and $t_2$.

In the following of this document, we will focus on the HDD index however the methodology presented can be applied to all average temperature related indices. Particularly for options on the CAT, the Fast Fourier Transform approach would enable to get a very efficient pricing method.

\subsubsection{Payoff function}

Weather derivatives are used to hedge weather risk. They trigger a payment depending on an aggregate temperature index. The payment is defined given a payoff structure. Standard payoff structures correspond to capped put or call options applied to the aggregate index. For simplification purposes we will suggest the payoff structure:

\begin{equation}
   \begin{aligned}
  \min((HDD - HDD_{strike})^+ , L).
    \end{aligned}\label{payoff} 
\end{equation}

As we want to price this kind of instruments our objective is to understand the characteristics of the payment distribution (expectation, VaR and CVaR) under the real world probability. In particular, we consider the average payoff of the derivative

\begin{equation}
   \begin{aligned}
   \E [  D(t_0, t_2)  \min((HDD - HDD_{strike})^+ , L) ],
   \end{aligned}\label{payoff_disc} 
\end{equation}

\noindent where $D(t_0, t_2)$ is a discount rate that will be taken equal to $1$ in this paper. Note that this is not a fair price: there is no market dealing HDD continuously and therefore the classical pricing theory of Black and Scholes does not apply. The calculation of the average payoff~\eqref{payoff_disc}, as well as other indicators on the distribution of $\min((HDD - HDD_{strike})^+ , L)$ such as the variance and quantiles, is used in practice to propose a price over the counter. Thus, a very accurate evaluation of~\eqref{payoff_disc} for a given model is not really at stake: one is more interested in evaluating risk and how the average payoff may change under stressed parameters.

Finally, there is no consensus on how to choose $HDD_{strike}$. However, it is a market practice to use quantiles, and particularly historical quantiles of the index, to define this strike. In the present paper, we consider the  90\% quantile which is within market practices.

\subsection{Monte-Carlo Approach}

A first pricing approach to identify the distribution of payments is to simulate temperature paths based on the discretization schemes in~\eqref{schemes} for $\Delta=1$. We proceed as follows:

\begin{enumerate}
    \itemsep0em 
    \item Simulate temperature paths starting from the pricing date $t_0$, the day until which we can observe temperature data, to the expiration date $t_2$.
    \item Compute simulated HDD between $t_1$ and $t_2$ for each of the paths and obtain an HDD distribution.
    \item Either fix an arbitrary $HDD_{strike}$ or choose a quantile to select the moneyness of the structure.
    \item Deduce the payment distribution.
    \item Compute payment distribution characteristics: mean, VaR and CVaR.
\end{enumerate}

Figure \ref{fig:MC} shows the results of this method for Paris temperature in 2019. Contracts last one month and are computed 30 days in ahead i.e. $t_1-t_0=30$. We consider the payoff function~\eqref{payoff} with $L=+\infty$ and $HDD_{strike}$ set to the 90\% empirical quantile of the $HDD$ distribution obtained with Model~\eqref{stoch}. We perform $50,000$ Monte Carlo simulations for the Ornstein–Uhlenbeck model and for Model~\eqref{stoch}. All these choices are challenged in the following sections.

\begin{figure}[h!]
    \begin{minipage}[b]{0.5\linewidth}
        \centering
        \includegraphics[width=\textwidth]{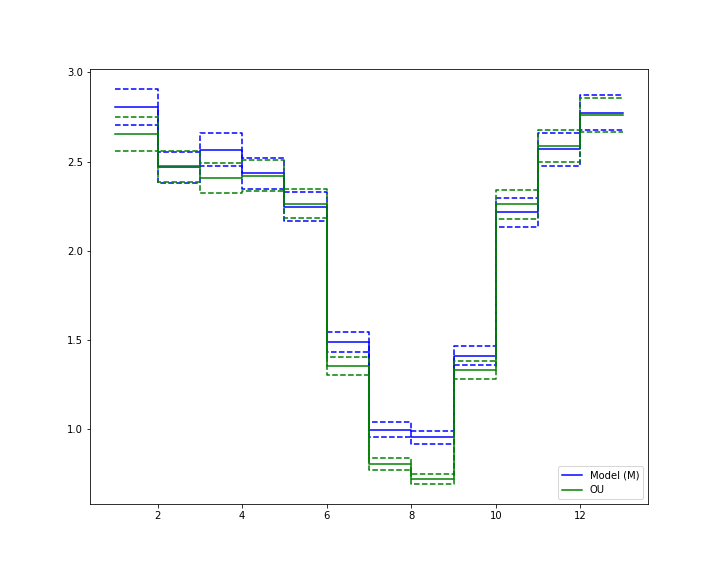}
        \caption*{Mean with 95\% confidence interval}
    \end{minipage}
    \begin{minipage}[b]{0.5\linewidth}
        \centering
        \includegraphics[width=\textwidth]{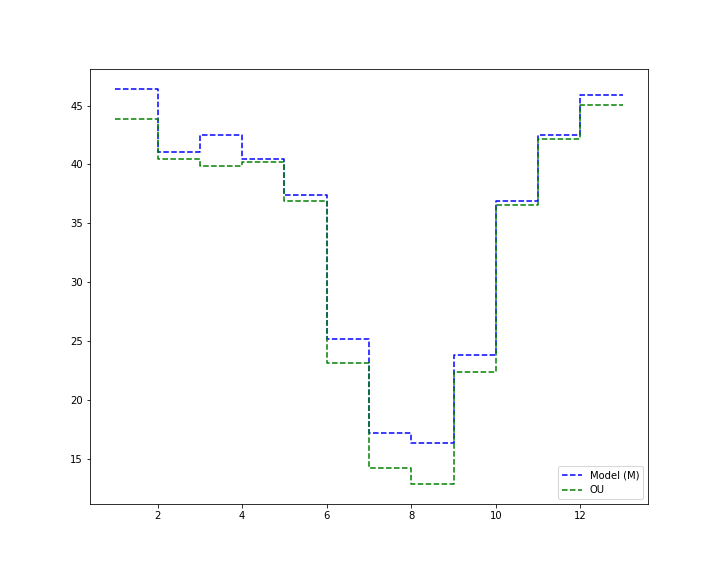}
        \caption*{Conditional Value at Risk at 95\%}
    \end{minipage}
    \caption{Different metrics of the payment distribution for $50,000$ Monte Carlo simulations, Paris, a cumulation period of a month, a forecast 30 days ahead and $HDD_{strike}$ corresponding to a 90\% quantile of the monthly HDD. Monte Carlo simulations are performed for both Model~\eqref{stoch} and the Ornstein–Uhlenbeck model~\eqref{T}.  }\label{fig:MC}
\end{figure}

From Figure \ref{fig:MC} we can see that both  Model~\eqref{stoch} and the Ornstein–Uhlenbeck model lead to similar expected payoffs for winter months while Model~\eqref{stoch} states higher expected payoffs for summer months. Model~\eqref{stoch} tends to have heavier left tails which are particularly visible in summer months. This explains slightly higher mean payoffs for Model~\eqref{stoch} during these months. Nevertheless, it should be noted that these derivatives are mainly sold for winter months to cover against cold waves. During these months both models give similar mean payoffs.

In terms of risk metrics, we compute Conditional Value at Risks at 95\% for both models. We can see  that Model~\eqref{stoch} presents again heavier tails.  Thus, as already noticed from Figures~\ref{fig:qqplot_simu_OU} and~\ref{fig:qqplot_simu}, Model~\eqref{stoch} is more conservative and reduces the problem of underestimating rare events related to the Ornstein–Uhlenbeck Gaussian framework. 

\subsection{Fast Fourier Transform Approach}

This section explores an alternative methodology for HDD pricing thanks to the Fast Fourier Transform (FFT) approach developed by Carr and Madan \cite{carr1999option}. {This method has been widely used in the literature for pricing, we just mention here the recent work of Benth et al.~\cite{BDK} for an application close to ours.}

\subsubsection{The characteristic function}

In order to apply FFT pricing, we first calculate the characteristic function of $(\tilde{T}_t,\zeta_t,\int_0^t\tilde{T}_s ds)$. A semi explicit formula is available because of the affine structure of Model~\eqref{stoch}.

\begin{prop}\label{prop_fct_car}
Let $0\le t \le t'$. Let $(\tilde{T},\zeta)$ be the solution of~\eqref{stoch} with $\rho=0$. The characteristic function of $(\tilde{T}_{t'},\zeta_{t'})$ given $\cF_t$ is, for $u_1,u_2,u_3 \in \R$,
\begin{equation}\label{fct_car}
\E\left[\exp\left( i[u_1 \tilde{T}_{t'}+u_2\zeta_{t'}+u_3 \int_t^{t'} \tilde{T}_s ds] \right) |\cF_t \right]=\exp(a_0(t,t')+ a_1(t'-t) \tilde{T}_t+a_2(t'-t)\zeta_t),
\end{equation}
where  $a_2$ is the unique solution on $\R_+$ of the time inhomogeneous autonomous Riccati equation
\begin{equation}\label{def_a2}
a_2'=-Ka_2-\frac 12 \left[u_1 \exp(-\kappa t)+u_3 \frac{1-\exp(-\kappa t)}{\kappa}\right]^2 %+ i \rho \eta u_1 a_2 \exp(2\kappa t)
+\frac 12 \eta^2 a_2^2, \  a_2(0)=iu_2,
\end{equation}
$a_1(t)=iu_1 \exp(-\kappa t)+iu_3 \frac{1-\exp(-\kappa t)}{\kappa}$ and $a_0(t,t')=K\int_t^{t'} \sigma^2(s) a_2({t'-}s)ds$. Besides, the real part of $a_2(t)$ remains nonpositive for all $t\ge 0$.
\end{prop}

\begin{proof} Let us first check that Equation~\eqref{def_a2} admits a unique solution, which is well defined for all $t\ge 0$. {When $u_1=u_2=u_3=0$, $a_2(t)=0$ is the unique solution and we get then $a_1(t)=0$, $a_0(t,t')$ so that~\eqref{fct_car} holds. We now exclude this case, and }  observe  that 
$$ \begin{cases}
\mathfrak{R}(a_2')=-K\mathfrak{R}(a_2)-\frac 12 \left[u_1 \exp(-\kappa t)+u_3 \frac{1-\exp(-\kappa t)}{\kappa}\right]^2 +\frac 12 \eta^2(\mathfrak{R}(a_2)^2-\mathfrak{I}(a_2)^2),\ \mathfrak{R}(a_2(0))=0,\\
\mathfrak{I}(a_2')=-K\mathfrak{I}(a_2)+ \eta^2\mathfrak{R}(a_2) \mathfrak{I}(a_2),\ \mathfrak{I}(a_2(0))=u_2,
\end{cases}$$
with $\mathfrak{R}(z)$ and $\mathfrak{I}(z)$ denoting the real and imaginary parts of   {a} complex number $z$. Let $\bar{t}=\inf \{t\ge  0: \mathfrak{R}(a_2(t))>0\}$. Since {$\mathfrak{R}(a_2(0))=0$, we have  $\mathfrak{R}(a'_2(0))=-\frac 12 (u_1^2 +\eta^2 u_2^2)<0 $ when $u_1\not=0$ or $u_2\not=0$ and thus we have $\bar{t}>0$. If $u_1=u_2=0$ and $u_3 \not=0$, we have $\mathfrak{I}(a_2(t))=0$,  $\mathfrak{R}(a'_2(0))=0$, $\mathfrak{R}(a''_2(0))=0$ and $\mathfrak{R}(a'''_2(0))=-u_3^2<0$ and thus again $\bar{t}>0$.} Then, we have $\mathfrak{I}(a_2(t))=u_2\exp \left(-Kt +\eta^2 \int_0^t\mathfrak{R}(a_2(s))ds\right)$ and thus $|\mathfrak{I}(a_2(t))| \le |u_2|$ for $t\in[0,\bar{t})$. We now observe that $\bar{t}$ cannot be finite. {If it were finite, we would have $\mathfrak{R}(a_2(\bar{t}))=0$ by continuity and then 
$$\mathfrak{R}(a'_2(\bar{t}))= -\frac 12 \left[u_1 \exp(-\kappa \bar{t})+u_3 \frac{1-\exp(-\kappa \bar{t})}{\kappa}\right]^2 -\frac 12 \eta^2\mathfrak{I}(a_2(\bar{t}))^2\le 0.$$
If $\mathfrak{R}(a'_2(\bar{t}))<0$, we get $\mathfrak{R}(a_2(t))>0$ in a left neighbourhood of $\bar{t}$ which is impossible. If $\mathfrak{R}(a'_2(\bar{t}))=0$, we then have $u_1 \exp(-\kappa \bar{t})+u_3 \frac{1-\exp(-\kappa \bar{t})}{\kappa}=0$ and $\mathfrak{I}(a_2(\bar{t}))=0$. The latter gives $u_2=0$. We check then that $\mathfrak{R}(a''_2(\bar{t}))=0$ and $\mathfrak{R}(a'''_2(\bar{t}))=-(u_3-\kappa u_1)^2e^{-2\kappa \bar{t}}<0$ since $u_1 \exp(-\kappa \bar{t})+u_3 \frac{1-\exp(-\kappa \bar{t})}{\kappa}=0$ and $(u_1,u_3)\not = (0,0)$. Again, this gives that $\mathfrak{R}(a_2(t))>0$ in a left neighbourhood of $\bar{t}$ which is impossible.
} Thus, $\bar{t}=+\infty$ and the ODE is then clearly well defined for all $t\ge 0$.

We now check that we indeed have~\eqref{fct_car}. Let $\mathcal{E}_t=\exp(a_0(t,t')+ a_1(t'-t) \tilde{T}(t)+a_2(t'-t)\zeta_t+iu_3 \int_0^t \tilde{T}_sds)$. By Itô's formula, we get for $t\in [0,t']$,
\begin{align*}
d\mathcal{E}_t=&\mathcal{E}_t\Bigg[\partial_t a_0(t,t') - a_1'(t'-t) \tilde{T}_t- a_2'(t'-t)\zeta_t-\kappa a_1(t'-t)\tilde{T}_t +a_2(t'-t)K({\sigma}^2(t)-\zeta_t) \\
&+\frac 12 a_1^2(t'-t) \zeta_t % + \rho \eta a_1a_2 \zeta_t \right]
+\frac 12 \eta^2 a_2^2(t'-t)\zeta_t  +iu_3 \tilde{T}_t \Bigg]dt +\mathcal{E}_t \sqrt{\zeta_t} [a_1(t'-t) dZ_t + \eta a_2(t'-t)dW_t].
\end{align*}
The first term vanishes, and we get 
$$\mathcal{E}_{t'}=\mathcal{E}_t+ \int_t^{t'}\mathcal{E}_s \sqrt{\zeta_s} [a_1(t'-s) dZ_s +\eta a_2(t'-s)dW_s]  $$
We note that $0\le { |\mathcal{E}_t|}\le \exp(a_0(t,t'))$ for $t\in [0,t']$ since $a_1 \in i \R$ and $\mathfrak{R}(a_2)\le 0$ { and that $\E[\zeta_t]=\zeta_0 e^{-Kt}+\int_0^t \sigma^2(s) e^{-K(t-s)}ds$ is integrable with respect to~$t$}. Thus, the integrand of the stochastic integral is square integrable, and we get  $$\mathcal{E}_t=\E[\mathcal{E}_{t'}|\cF_t]=\E\left[\exp\left( i[u_1 \tilde{T}_{t'}+u_2\zeta_{t'}+{u_3 \int_0^{t'} \tilde{T}_sds}] \right) \Bigg|\cF_t \right],$$ 
{which gives the claim. }
\end{proof} 

\begin{remark}\label{rk_fct_car}
Formula~\eqref{fct_car} can be extended easily to $u_2\in \R +i\R_+$. We then have $\mathfrak{R}(a_2(0))=-\mathfrak{I}(u_2)\le 0$, and the proof of Proposition~\ref{prop_fct_car} can be repeated step by step. 
\end{remark}

\subsubsection{Approximation of the characteristic function}\label{sub_sec_approx_riccati}

We now discuss the approximation of the characteristic function~\eqref{fct_car}. To do so, we consider a time step $\delta>0$, and we will assume that  $t = t_k = k \delta$ and $t' = t_l = l \delta$. Note that the function $a_1$ is fully explicit and does not need to be approximated. We use the trapezoidal rule to integrate the function $a_0$:
\begin{equation}\label{approx_a0}
a_0(t_k,t_l)  \approx K \sum_{j=k}^{l-1}  \frac 12 [{\sigma}^2(t_j)a_2(t_l-t_j)+\sigma^2(t_{j+1})a_2(t_l-t_{j+1})]\delta , \ k<l.
\end{equation}
The main issue may come from the discretization of the Riccati equation which may lead to instabilities if it is not well handled. Here, we take advantage of the fact that an explicit solution of~\eqref{def_a2} is known for $\kappa=0$ and $u_3=0$, see e.g.~\cite[p.~101]{alfonsi2015affine},
$$ a_2(t)=\Psi+\frac{2\sqrt{D} (\Psi-iu_2) }{\left( \eta^2(\Psi-iu_2)-2\sqrt{D}\right) \exp(-\sqrt{D} t)-\eta^2(\Psi-iu_2) },$$
with
\begin{align*}
  D&=K^2+\eta^2u_1^2,\quad    \Psi=\frac{K+\sqrt{D}}{\eta^2}.
\end{align*}
Thus, to solve~\eqref{def_a2}, we freeze on each interval $[t_k,t_{k+1}]$ the value of the time inhomogenous term to its value at $t=\frac{t_k+t_{k+1}}{2}$, and use the explicit formula. {This is the midpoint method that leads formally to a convergence of order $O(\delta^2)$.} This leads to:
\begin{equation} \label{char_func}
    a_2(t_{k+1}) =\Psi_k+\frac{2\sqrt{D_k} (\Psi_k-a_2(t_k)) }{\left( \eta^2(\Psi_k-a_2(t_k))-2\sqrt{D_k}\right) \exp(-\sqrt{D_k} \delta)-\eta^2(\Psi_k-a_2(t_k)) }
\end{equation}

\noindent where

\begin{equation*}
    D_k=K^2+\eta^2 \left(u_1 \exp\left(-\kappa \frac{t_k+t_{k+1}}2\right)+u_3 \frac{1-\exp\left(-\kappa \frac{t_k+t_{k+1}}2\right)}{\kappa}  \right)^2,\quad    \Psi_k =\frac{K +\sqrt{D_k}}{\eta^2}.
\end{equation*}

\begin{figure}[h!]
    \begin{minipage}[b]{0.45\linewidth}
        \centering
        \includegraphics[width=\textwidth]{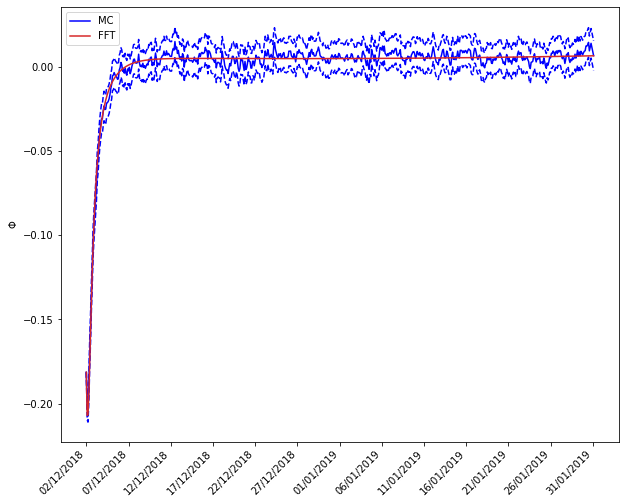}
    \end{minipage}
    \hspace{0.3cm}
    \begin{minipage}[b]{0.45\linewidth}
        \centering
        \includegraphics[width=\textwidth]{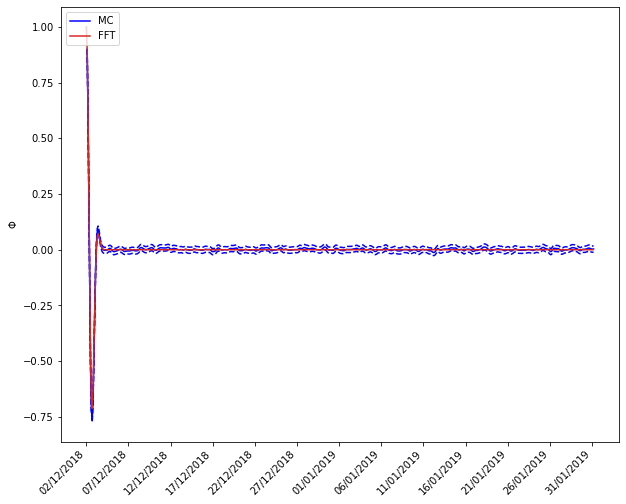}
    \end{minipage}
    \caption{Characteristic function $\E\left[\exp\left( iu_1 \tilde{T}_{t'}\right) \right]$ (left) and $\E\left[\exp\left( i u_3  \int_t^{t'} \tilde{T}_s ds   \right) |\cF_t  \right]$ (right) for Paris temperature during January 2019 for an observation time 30 days ahead and  $\delta=0.1$ day.}
    \label{fig:FFT_charact}
\end{figure}

We implement the three functions in~\eqref{char_func} which enable us to deduce the characteristic function of $(u_1 \tilde{T}(t)+u_2\zeta_{t})_{t\geq 0}$ where $(u_1,u_2) \in \mathbb{C}^2$. Figure \ref{fig:FFT_charact} shows the characteristic function of $(T_t)_{t\geq 0}$ calculated with the approximation~\eqref{char_func}. It is compared with the  Monte-Carlo estimator obtained with simulated path using~\eqref{schemes} (we have used here independent simulations for each values of $t'$). We can see that both methods give close results, which validates the relevance of the approximation. 

\subsubsection{Fast Fourier Transform for pricing HDD and related options}

Once we have the characteristic function of $(T_t)_{t\geq 0}$, we can price HDD by using Fourier inverse transform techniques. Here, we {adapt} the approach of Carr and Madan~\cite{carr1999option} that uses the Fast Fourier Transform {in order to calculate the cumulative distribution function of~$T_t$. This allows us to calculate then easily the average value different types of bespoke options. }

We first focus on the calculation of  $\E[(T_b- T_t)^+]$ {(resp. $\E[\min((T_b- T_t)^+,L)]$)}, that can be seen as the average price of a ``daily HDD" {(resp. capped daily HDD)}.  We will use this naming later on. The characteristic function $\Phi(u)=\E[e^{iu \tilde{T}_t}]$  is given by Proposition~\ref{prop_fct_car}. We have 
\begin{align}\label{price_dailyhdd}\E[(T_b- T_t  )^+]&= \int_{0}^{\infty} \mathbb{P}(T_b- T_t \geq x )dx = \int_{-\infty}^{T_b-s(t)} \mathbb{P}(\tilde{T}_t\leq x)dx  \\
{\E[\min((T_b- T_t  )^+,L)]}&{= \int_{0}^{L} \mathbb{P}(T_b- T_t \geq x )dx = \int_{T_b-s(t)-L}^{T_b-s(t)} \mathbb{P}(\tilde{T}_t\leq x)dx } \notag
\end{align}
We approximate this cumulative density function by using the Gil-Pelaez inversion formula~\cite[Theorem 4.2.3 p. 104]{alfonsi2015affine}:
$$   \mathbb{P}( \tilde{T}_t \leq x)=\frac 12 - \frac 1 \pi \lim_{m\to 0^+, M\to +\infty} \int_m^M \mathfrak{R}\left(\frac{e^{-ivx} \Phi(v)}{iv}\right)dv, \  \Phi(v)=\E[e^{iv\tilde{T}_t}].$$
Let $\delta_x,\delta_v>0$ be such that $\delta_x \delta_v=\frac{2\pi}{N}$ (one can take for example $\delta_x=\delta_v=\sqrt{\frac{2\pi}{N}}$ {or $\delta_x=\frac{L}{N-1}$ for the capped option so that $x_0$ defined below is equal to $T_b-s(t)-L$}). We define:
$$ v_{j+1/2}=(j+1/2)\delta_v, \ x_k=T_b-s(t) +(k-N+1)\delta_x, j \in \{0,{..., }N-1\}, k \in \{0,\dots, N-1\},$$
so that $x_{N-1}=T_b-s(t)$, and use the following approximation:
\begin{align}
   \notag \mathbb{P}( \tilde{T} \leq x_k) &\approx \frac 12 - \frac{\delta_v}{\pi} \mathfrak{R}\left( \sum_{j=0}^{N-1}\frac{e^{-iv_{j+1/2}x_k} \Phi(v_{j+1/2})}{iv_{j+1/2}}\right)\\ &= \frac 12 - \frac{\delta_v}{\pi} \mathfrak{R}\left(e^{-\frac 12 i\delta_v k \delta_x}\sum_{j=0}^{N-1}e^{-2i\pi  \frac{jk}{N}} \frac{ e^{-i (j+1/2)\delta_vx_0}\Phi(v_{j+1/2})}{iv_{j+1/2}}\right), \label{approx_cdf} 
\end{align}
since $ v_{j+1/2}x_k= 2\pi \frac{jk}{N}+(j+1/2)\delta_vx_0 + \frac 12 \delta_v k \delta_x$.
This amounts to use the {midpoint} rule and to truncate the integral at $M=N\delta_v$. {Other choices of quadrature are possible but have to be taken in compliance with the FFT.} Using~\eqref{approx_cdf}, we can obtain $(\mathbb{P}( \tilde{T} \leq x_k), 0\le k\le N-1)$ by applying the FFT to $\left(\frac{ e^{-i (j+1/2)\delta_vx_0}\Phi(v_{j+1/2})}{iv_{j+1/2}},0\le j\le N-1\right)$:  these $N$ values are obtain with a  time complexity of $O(N\log(N))$ (instead of $O(N^2)$ with the naive calculation of the sums).  We finally approximate the expectation of the daily HDD by:
\begin{equation}
    \E[(T_b- T_t )^+] \approx \delta_x\left( \sum_{k=0}^{N-2} \mathbb{P}(\tilde{T}_t\leq x_k)  + \frac12 \mathbb{P}(\tilde{T}_t\leq x_{N-1})\right).
\end{equation}

Figure \ref{fig:FFT} shows the characteristic function and the expected HDD in Paris during January 2019 comparing Monte Carlo simulation and FFT approach. Both graphs display a clear coherence between Monte Carlo simulations and the FFT approach. 
{In this case Monte Carlo simulations show precision given that we simulate $50,000$ scenarios. However, FFT pricing is more precise, smooth and faster. We perform FFT with $N=2^{17}$. }

\begin{figure}[h!]
    \begin{minipage}[b]{0.45\linewidth}
        \centering
        \includegraphics[width=\textwidth]{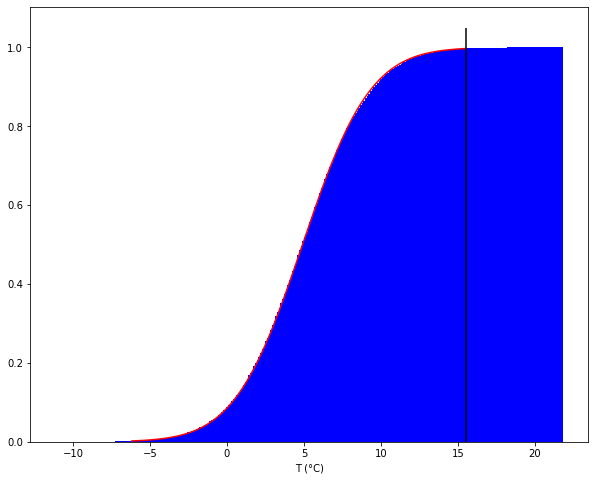}
    \end{minipage}
    \hspace{0.3cm}
    \begin{minipage}[b]{0.45\linewidth}
        \centering
        \includegraphics[width=\textwidth]{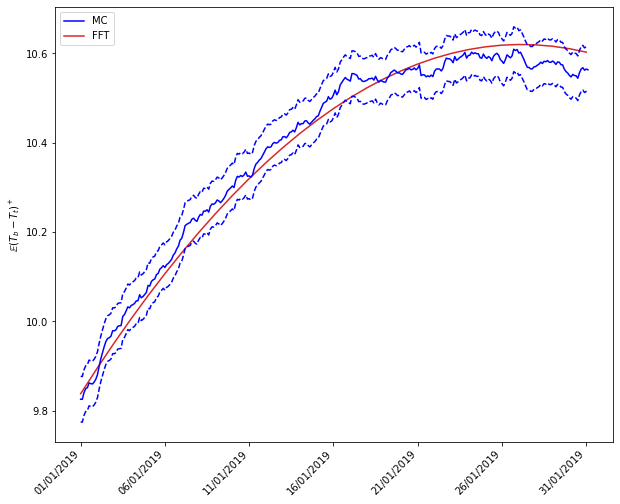}
    \end{minipage}
    \caption{Cumulative distribution function of 31st January 2019's daily temperature (left) and expected daily HDD during January 2019 days computed 30 days ahead of the month (right) by the FFT method (red) and Monte-Carlo (blue). The black vertical line  corresponds to $T_b$.}
    \label{fig:FFT}
\end{figure}

We now focus on the pricing on options on HDD. We observe from the distribution in Figure~\ref{fig:FFT} that we almost always have $T_t\le T_b$ in January, otherwise we would notice a Dirac mass at~$0$. In fact, with the standard strike $T_b=15.5$°C, we mostly have $T_t\le T_b$ during winter, and therefore $HDD\approx (t_2-t_1+1)T_b-CAT$, so that {the average value of the option~\eqref{payoff} can be approximated by }  
\begin{equation}\label{approx_HDDopt} {\E[\min((HDD-HDD_{strike})^+,L) ]\approx \E[\min(((t_2-t_1+1)T_b
-HDD_{strike}-CAT)^+,L)].}\end{equation}
The problem of computing the right hand side is then similar to the pricing of daily HDD~\eqref{price_dailyhdd}: the underlying is now $CAT$ instead of $T_t$. We can thus calculate the average payoff, provided that we know the characteristic function of the  CAT $\Phi(u)=e^{iu \sum_{t=t_1}^{t_2} s(t)} \E[e^{iu \sum_{t=t_1}^{t_2} \tilde{T}_t }]$. To do so, it is possible to use formula~\eqref{fct_car} inductively with $u_3=0$ (using Remark~\ref{rk_fct_car})  in order to calculate  $\E[e^{iu \sum_{t=t_1}^{t_2} \tilde{T}_t }|\cF_{t_2-\ell}]$ for $\ell=1,\dots,t_2-t_1$ and then $\Phi$. This is however cumbersome, and we prefer to make the following approximation
$$ \Phi(u)\approx e^{iu \sum_{t=t_1}^{t_2} s(t)}  \E[e^{iu \int_{t_1}^{t_2+1} \tilde{T}_t dt }]. $$
We apply Proposition~\ref{prop_fct_car} with $u_1=u_2=0$ and $u_3=u$, and get $\E[e^{iu \int_{t_1}^{t_2+1} \tilde{T}_t dt }|\cF_{t_1}]=\exp(a_0(t_1,t_2+1)+i u \frac{1-e^{-\kappa (t_2+1-t_1) }}\kappa  \tilde{T}_{t_1}+a_2(t_2+1-t_1)\zeta_{t_1})$.
Hence, for {$t_0 \leq t_1 \leq t_2$,}
\begin{align*}
    \E[e^{iu \int_{t_1}^{t_2+1} \tilde{T}_t dt }|\cF_{t_0}] &= \E[{\E}[e^{iu \int_{t_1}^{t_2+1} \tilde{T}_t dt }|\cF_{t_1}] |\cF_{t_0}]\\
    &=\E[\exp(a_0(t_1,t_2+1)+i u \frac{1-e^{-\kappa (t_2+1-t_1) }}\kappa  \tilde{T}_{t_1}+a_2(t_2+1-t_1)\zeta_{t_1}) |\cF_{t_0}]\\
    &=\exp(a_0(t_1,t_2+1))\exp(\check{a}_0(t_0,t_1)+\check{a}_1(t_1 -t_0) \tilde{T}_{t_0}+\check{a}_2(t_1-t_0)\zeta_{t_0}) 
\end{align*}
To obtain~$\check{a}_0, \check{a}_1, \check{a}_2$ and the above characteristic function,  we apply a second time  Proposition~\ref{prop_fct_car} with $u_1= u \frac{1-e^{-\kappa (t_2+1-t_1) }}\kappa$, $u_2=-ia_2(t_2+1-t_1)$ and $u_3=0$. Figure~\ref{fig:FFT_CAT} compares CAT distribution obtained with Monte Carlo and FFT inverse methods for the month of January 2019. We can observe a good fit between both methods.\\

This section has focused on the capacity of Fast Fourier Transform method to compute explicit formulas for options on $T$ and $CAT$. This methodology enables to get rid of the computation burden of Monte Carlo simulations. However, not all the indices or payoff functions can be explicited with FFT. In particular, derivatives that integrate double non-linearities, like put or call payoff functions applied to HDD, cannot be explicitly computed with the FFT method. The next section will focus on how to use FFT results to increase the performance of Monte Carlo simulation for such cases.

\begin{figure}[h!]
    \centering
        \includegraphics[width=0.6\textwidth]{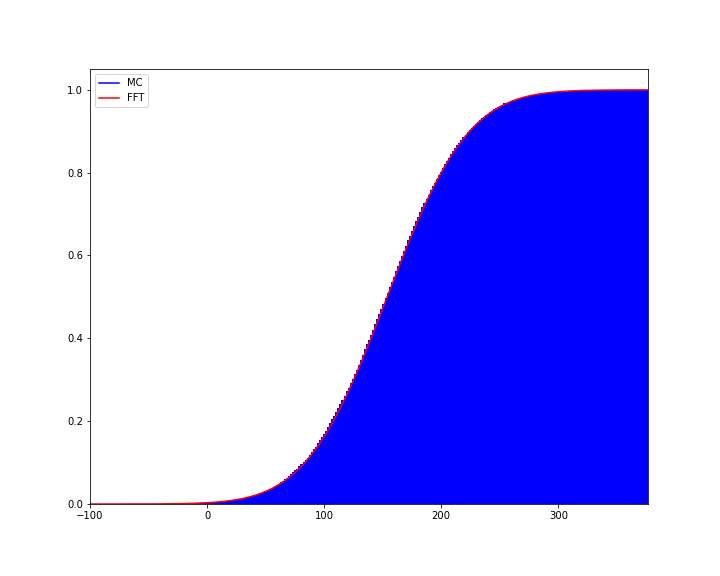}
    \caption{Cumulative distribution function of CAT for January 2019 and 30 days observation in advance computed by the FFT method (red) and Monte-Carlo  with $50,000$ simulations (blue). }
    \label{fig:FFT_CAT}
\end{figure}

\iffalse 
For $x_{N-1} = (t_2-t_1+1)T_b - HDD_{strike} + \sum_{t=t_1}^{t_2} s(t)$

\begin{align}
    \E[((t_2-t_1+1)T_b -HDD_{strike} - CAT )^+] &\approx \delta_x\left( \sum_{k=0}^{N-2} \mathbb{P}(\int_{t_1}^{t_2+1} \tilde{T}_t dt \leq x_k)  + \frac12 \mathbb{P}(\int_{t_1}^{t_2+1} \tilde{T}_t dt\leq x_{N-1})\right)\\
   \notag \mathbb{P}( \tilde{T} \leq x_k) & \approx \frac 12 - \frac{\delta_v}{\pi} \mathfrak{R}\left(e^{-\frac 12 i\delta_v k \delta_x}\sum_{j=0}^{N-1}e^{-2i\pi  \frac{jk}{N}} \frac{ e^{-i (j+1/2)\delta_vx_0}\varphi(v_{j+1/2})}{iv_{j+1/2}}\right)
\end{align}
\fi

\subsubsection{Control variates method for Monte-Carlo}\label{subsec_control_variates}

In practice the approximation~\eqref{approx_HDDopt} is precise when $\mathbb{P}(T_t>T_b)$ is close to zero. However, when this probability is small but not negligible, the approximation may not be enough precise. However, we can use the calculation above to run a Monte-Carlo method with the control variable {$\min((HDD-HDD_{strike})^+,L) - \lambda \min( ((t_2-t_1+1)T_b -HDD_{strike}-CAT)^+,L) $} in order to calculate the average payoff $\E[(HDD-HDD_{strike})^+ ]$. Namely, we write {(we take here $L=+\infty$ for simpler notation)} 
\begin{align}\label{control_variate} \E[(HDD-HDD_{strike})^+ ]=&\lambda \E[((t_2-t_1+1)T_b
-HDD_{strike}-CAT)^+] \\
&+\E\left[ (HDD-HDD_{strike})^+ - \lambda ((t_2-t_1+1)T_b
-HDD_{strike}-CAT)^+ \right],\notag
\end{align}
and we chose $\lambda$ that minimizes $Var\left[ (HDD-HDD_{strike})^+ - \lambda ((t_2-t_1+1)T_b -HDD_{strike}-CAT)^+ \right]$, i.e. 
$$ \lambda^*= \frac{Cov((HDD-HDD_{strike})^+,((t_2-t_1+1)T_b -HDD_{strike}-CAT)^+)}{Var(((t_2-t_1+1)T_b -HDD_{strike}-CAT)^+)}.$$
The first term of the right hand side of~\eqref{control_variate} is calculated by using the FFT while the second one is calculated by Monte-Carlo.

\begin{figure}[h!]
    \centering
    \includegraphics[width=0.6\textwidth]{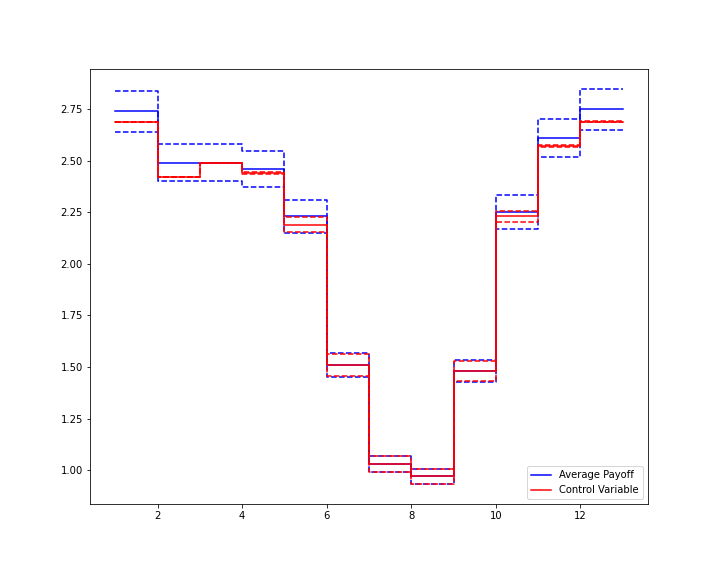}
    \caption{Expected payoffs forecasted 30 days ahead for the HDD derivative~\eqref{payoff} {(with $L=+\infty$)} on each month of 2019 (blue) and for the control variable (red). We performed $50,000$ Monte Carlo simulations, {dotted lines indicate the 95\% confidence interval}.}
    \label{fig:Control Variates}
\end{figure}

\begin{table}[h!]
\centering
\begin{tabular}{c|cccccccccccc}
Month & 1 & 2 & 3 & 4 & 5 & 6 & 7 & 8 & 9 & 10 & 11 & 12 \\ \hline \hline
Corr & 1.00 & 1.00 & 1.00 & 1.00 & 0.94 & 0.66 & 0.38 & 0.33 & 0.66 & 0.97 & 1.00 & 1.00 \\
VR & 2.41e5 & 5.24e4 & 4.73e3 & 2.22e2 & 5.08 & 1.19 & 1.01 & 1.01 & 1.20 & 9.84 & 3.92e2 & 1.40e4\end{tabular}
\caption{Correlation  and variance reduction (VR) brought by the control variates method for options computed during each month of 2019. Variance ratio corresponds to the variance of  $(\sum_{t=t_1}^{t_2}(T_b-T_t)^+ - HDD_{strike})^+$ divided by the variance of the control variable. }
\label{tab:Control Variates}
\end{table}

Figure~\ref{fig:Control Variates} and Table~\ref{tab:Control Variates} show the results of the implementation of the control variates method. First,   from Table~\ref{tab:Control Variates} we can see that the control variates method enables to decrease the variance up to $ 2.41 \times 10^5$ times, leading to a price computation $10^5$ time faster. Second, we can observe that the performance of the method depends on the correlation between $(\sum_{t=t_1}^{t_2}(T_b-T_t)^+ - HDD_{strike})^+$ and the control variable:  the more correlated they are, the more variance reduction we obtain (and the more approximation~\eqref{approx_HDDopt} is valid). Hence for the winter months the computation is more effective, which coincides with the months for which such options are sold. In Figure~\ref{fig:Control Variates}, we see that the confidence interval for winter months is considerably narrower.  

To sum up, this section has focused on exploring alternative pricing methodologies. The Fast Fourier Transform pricing method enables to bypass Monte Carlo simulations and to get analytical expressions for the expected payoffs of some derivatives. This enables direct expectation computations for some derivatives like $CAT$. For derivatives integrating non linear indices and payoffs, this method can be combined with the control variates method to decrease the computational cost of the Monte Carlo simulations. In our case, this method enables to considerably decrease the number of required simulations.
 
\subsection{Sensitivity study}\label{subsec_sensi}

This section studies the sensitivity of the pricing to the different parameters that were either imposed or estimated in the previous sections.

\noindent \paragraph{Sensitivity to $\kappa$}

We first analyse the sensitivity to the mean reverting parameter of the temperature dynamics. Figure \ref{fig:sensitivity_kappa} shows that increasing $\kappa$ has an important effect on the average payoff and the simulated HDD distribution and therefore on the pricing. Indeed when $\kappa$ increases the volatility loses its importance, the HDD distribution becomes more certain and therefore peaks around an expected average value. Similarly, the quantiles are less spread and therefore the strikes based on initially estimated $\kappa$ are less frequent and mean payoffs shrink.

\begin{figure}[h!]
    \begin{minipage}[b]{0.45\linewidth}
        \centering
        \includegraphics[width=1.1\textwidth]{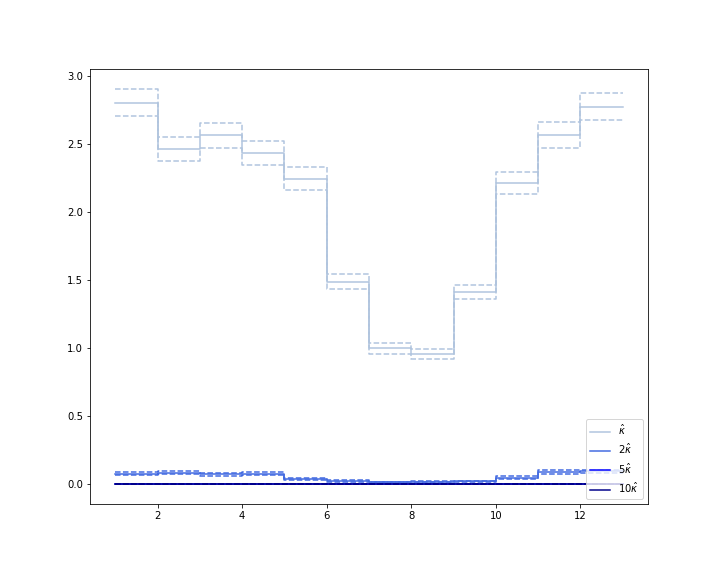}
    \end{minipage}
    \begin{minipage}[b]{0.45\linewidth}
        \centering
        \includegraphics[width=1.1\textwidth]{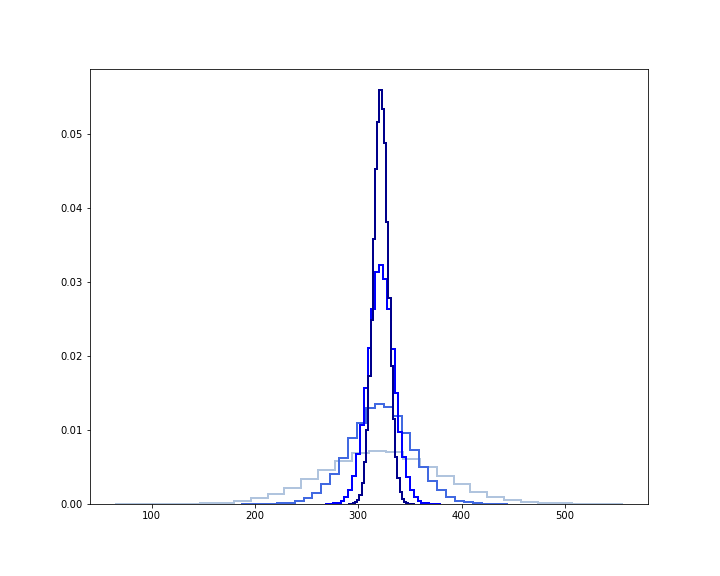}
    \end{minipage}
    \caption{Average payoffs for values of $\kappa \in \{\hat{\kappa}, 2\hat{\kappa}, 5\hat{\kappa}, 10\hat{\kappa}\}$ (left) and HDD distribution  for a derivative on the month of January 2019, forecasted 30 days ahead and based on $50,000$ Monte Carlo simulations (right). $HDD_{strike}$ is kept at 90\% empirical quantile of $\kappa = \hat{\kappa}$ for all simulations.}
    \label{fig:sensitivity_kappa}
\end{figure}

\noindent \paragraph{Sensitivity to $\eta^2$}

Figure \ref{fig:sensitivity_eta} shows different average payoffs and confidence intervals for $50,000$ simulations and different values of the volatility of the volatility  $\eta^2$. We can see that increasing $\eta^2$ creates more peaked distributions for HDD and heavier tails. This also leads to usually higher prices. Nevertheless, it should be noted that the impact of $\eta^2$ is marginal in winter months when this product is meant to be sold. In summer months, HDDs only capture extreme temperature left tails and this is when we can see a real impact of the volatility of Model \eqref{stoch}. Besides, let recall that the estimation of $\eta^2$ is sensitive to the choice of $Q$. Wrongly estimating this parameter would therefore mainly impact the pricing of derivatives on summer months where the demand of such derivatives is much lower.

\begin{figure}[h!]
    \begin{minipage}[b]{0.45\linewidth}
        \centering
        \includegraphics[width=1.1\textwidth]{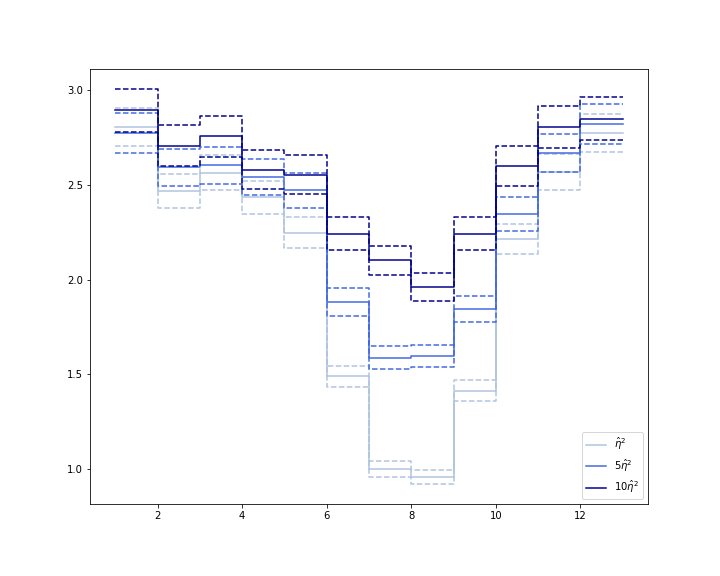}
    \end{minipage}
    \begin{minipage}[b]{0.45\linewidth}
        \centering
        \includegraphics[width=1.1\textwidth]{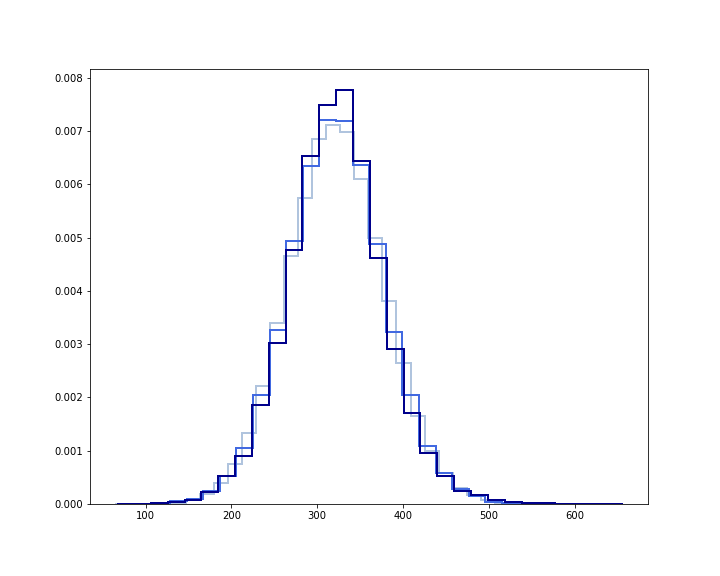}
    \end{minipage}
    \caption{Average payoffs for different values of  $\eta^2 \in \{\widehat{\eta^2}, 5\widehat{\eta^2}, 10\widehat{\eta^2}\}$ (left) and HDD distribution for a derivative on the month of January 2019, forecasted 30 days ahead and based on $50,000$ Monte Carlo simulations (right). $HDD_{strike}$ is kept at 90\% empirical quantile of $\eta^2  = \widehat{\eta^2}$ for all simulations. }
    \label{fig:sensitivity_eta}
\end{figure}

\noindent \paragraph{Sensitivity to $K$}

Figure \ref{fig:sensitivity_K} shows different average payoffs and confidence intervals for $50,000$ simulations. We can see that increasing $K$ creates less peaked distributions for HDD and lighter tails. This leads to usually lower mean payoffs when $K$ increases. Likewise, the impact is marginal in winter months when this product is meant to be sold. This phenomenon is intuitive as we increase the mean reverting term of the volatility, the volatility of the volatility losses weight in the dynamics and the extreme HDDs decrease.

\begin{figure}[h!]
    \begin{minipage}[b]{0.45\linewidth}
        \centering
        \includegraphics[width=1.1\textwidth]{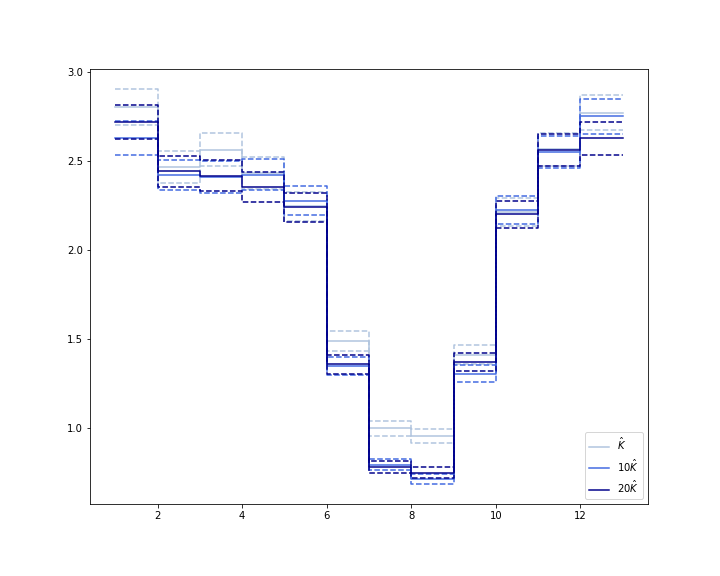}
    \end{minipage}
    \begin{minipage}[b]{0.45\linewidth}
        \centering
        \includegraphics[width=1.1\textwidth]{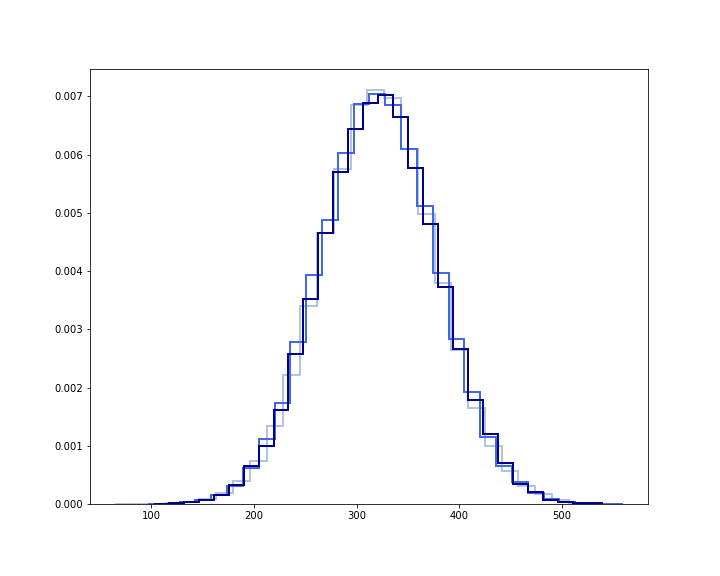}
    \end{minipage}
    \caption{Average payoffs for different values of  $K \in \{\hat{K}, 10 \hat{K}, 20\hat{K}\}$ (left), HDD distribution distribution starting from $HDD_{strike}$ (right) for a derivative on the month of January 2019, forecasted 30 days ahead and based on $50,000$ Monte Carlo simulations (right). $HDD_{strike}$ is kept at 90\% empirical quantile of $K  = \hat{K}$ for all simulations. }
    \label{fig:sensitivity_K}
\end{figure}

\noindent \paragraph{Sensitivity to $t_1-t_0$} Now we suppose we compute the price of the derivative different possible times ahead. Figure \ref{fig:sensitivity_time} shows temperature paths for different observation times $t_0$ but the same observed temperature and volatility at $t_0$. The derivative applies between the black vertical lines. We can observe that all the paths end up following the seasonality $s$. 

\begin{figure}[h!]
    \begin{minipage}[b]{0.30\linewidth}
        \centering
        \includegraphics[width=1.1\textwidth]{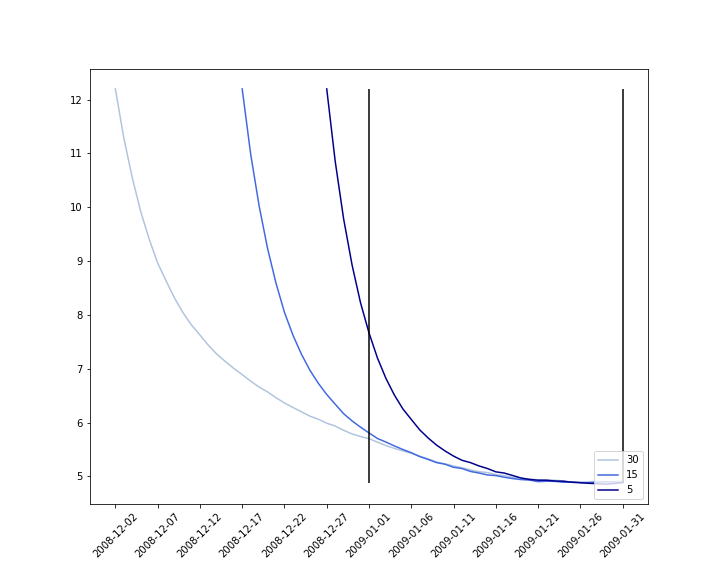}
    \end{minipage}
        \begin{minipage}[b]{0.30\linewidth}
        \centering
        \includegraphics[width=1.1\textwidth]{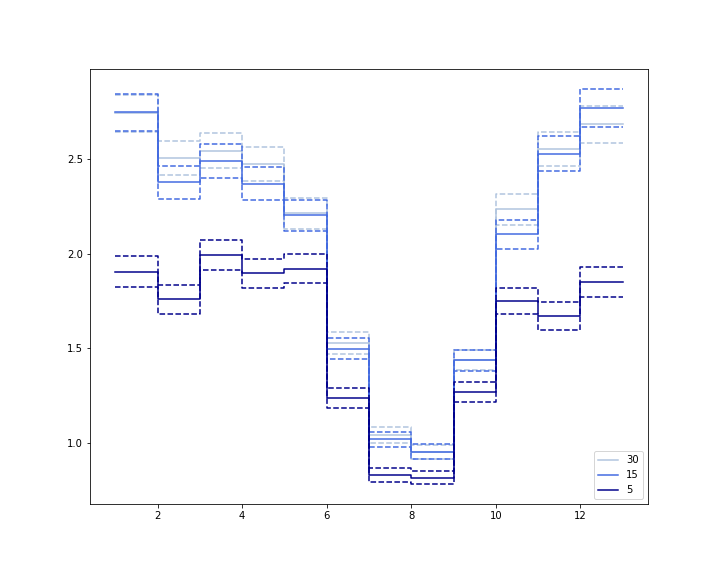}
    \end{minipage}
    \begin{minipage}[b]{0.30\linewidth}
        \centering
        \includegraphics[width=1.1\textwidth]{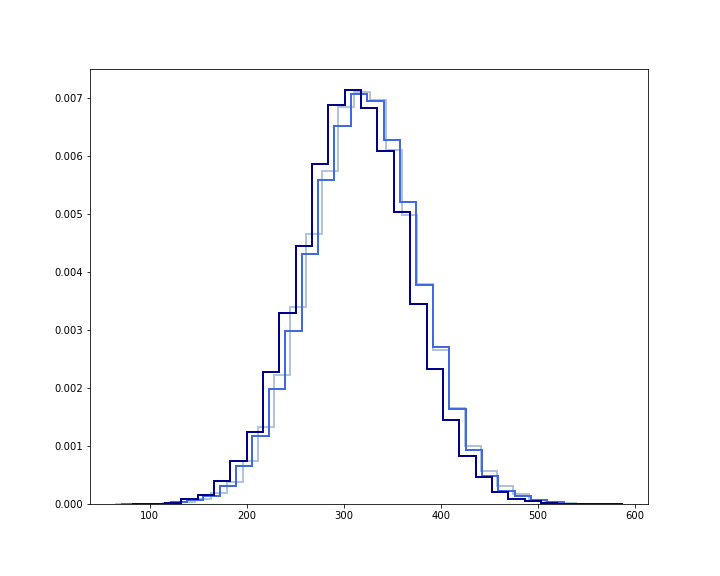}
    \end{minipage}
    \caption{Average temperature paths for different values of $t_1 - t_0$ (left), average payoffs for different values of $t_1 - t_0$ for a derivative on each month of 2019 (center) and HDD distribution for a derivative on the month of January 2019, forecasted 5, 15 and 30 days ahead and based on $50,000$ Monte Carlo simulations (right). Here, $T_{t_0} = s(t_0) + 2\sigma(t_0)$ and $\zeta_{t_0} = \sigma(t_0)$ for $t_0=30$. $HDD_{strike}$ is fixed at the 90\% quantile of the 30 days ahead simulation. In the left plot, the two black vertical lines represent times $t_1$ and $t_2$.}
    \label{fig:sensitivity_time}
\end{figure}

From simulated densities in Figure \ref{fig:sensitivity_time} we can first observe shifts depending on $t_1-t_0$. This significantly impacts the quantiles of these densites and therefore the $HDD_{strike}$. Second, we can note that the more ahead we forecast the less information we have. In this case, we can observe that while pricing 20 and 30 days ahead lead to similar average payouts during the risk period, forecasting 5 days ahead significantly impacts the average payoffs and therefore pricing. This element is key to answer the risk of antiselection.

\noindent \paragraph{Sensitivity to the moneyness of the product} 

The moneyness of the product has a direct impact on the payoffs distribution as can be observed in Figure \ref{fig:moneyness}. The lowest the $HDD_{strike}$, the more HDD we capture in the payoff and the higher the mean payoffs becomes. 

\begin{figure}[h!]
    \centering
    \includegraphics[width=0.6\textwidth]{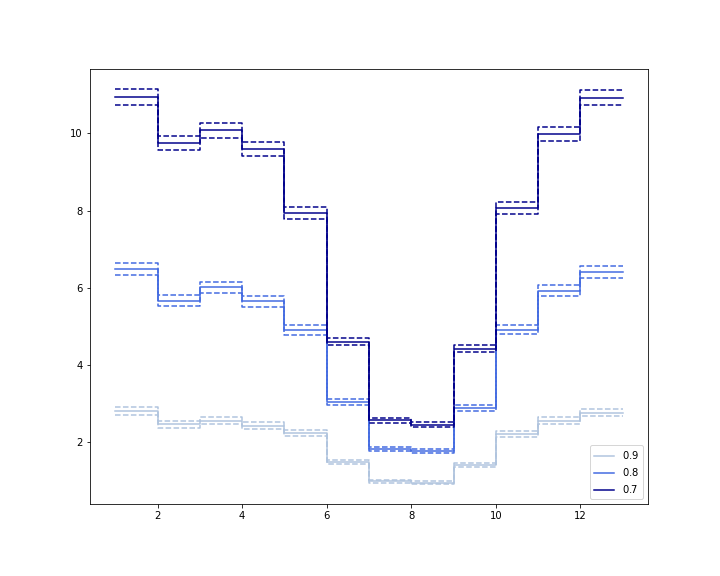}
    \caption{Average payoffs for different quantiles defining $HDD_{strike}$ in $\{0.7, 0.8, 0.9\}$ for each month of 2019 and based on the $50,000$ Monte Carlo simulation 30 days ahead.}
    \label{fig:moneyness}
\end{figure}

To sum up, this section has focused on the sensitivity of the pricing to the different parameters. We particularly show that the parameters related to the temperature dynamics like $\kappa$ as well as the moneyness of the payoff function are the ones affecting the most the mean payoffs. Parallely, the distribution of the payoffs and the strike based on quantiles show important sensitivity to the time interval $t_1-t_0$. Finally, the parameters related to the volatily have a relatively lower impact on the payoffs distribution and hence on the pricing.

\subsection{Comparison of our pricing methodology with business practices}

This section aims to make the bridge between the pricing methodology exposed in this document and the current market practices. In particular, we will compare a pricing based on modeling of the underlying meteorological parameter with a pricing based on historical index modeling. 

As described in Schiller {et} al. \cite{schiller2012temperature} and Jewson {and} Brix \cite{jewson2005weather}, the index based pricing methodology consists in modeling the independent yearly indices, in this case cumulative HDD.

For computation ease, the index pricing is automatised following the below algorithm:
\begin{enumerate}
    \itemsep0em 
    \item Compute historical index, cumulative $HDD$, from 1980 to 2018.
    \item Remove the linear trend in this time series. 
    \item Fit a gamma distribution to these observations through a maximum likelihood method.
    \item Compute the expected payoff of the fitted index distribution.
\end{enumerate}

While this approach can be simplistic, as we can consider other probability distributions for the index, it syntheses the common market practices.

Figure \ref{fig:index_model} represents the expected payoffs computed with the two methodologies. $HDD_{strike}$ corresponds to the 90\% simulated quantile with the Monte Carlo method and to the 90\% historical quantile for the index modeling method. First, we can observe there is a coherence between the approaches that give expected payoffs in the same ranges. However, the index modeling approach introduces important instabilities. These instabilities affect both the average payoffs as well as the strike $HDD_{strike}$, which is estimated from less than 40 observations.

\begin{figure}[h!]
    \centering
    \includegraphics[width=0.6\textwidth]{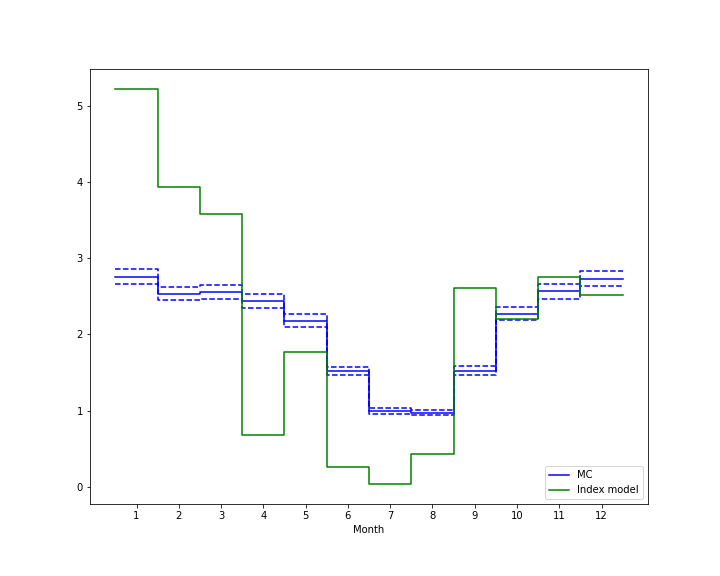}
    \caption{Expected payoffs forecasted 30 days ahead for a derivative on each month of 2019. $HDD_{strike}$ is defined differently with two methodologies. For the Monte Carlo approach it corresponds to the 90\% quantile of $50,000$ simulations while for the Index model approach the $HDD_{strike}$ corresponds to the historical quantile.}
    \label{fig:index_model}
\end{figure}

We also studied the possibility of using the same strikes with both methods. First, using the 90\% quantile of $50,000$ Monte Carlo simulations leads to slightly more volatile average payoffs for the index model method. However, we feel it is counter-intuitive to use a yearly index modeling for pricing and a more cumbersome daily index modeling just to get the strikes. 
Second, we can use historical quantiles for both methodologies as on the left of Figure \ref{fig:quantile_index_model}. In this case we can see that the winter months seem to be completely overpriced by the business practice while the summer months are underpriced. On the right of Figure \ref{fig:quantile_index_model}, we compare simulated and historical quantiles. We can see that the more difference we have between these quantiles the higher the risk of over or underpricing.

\begin{figure}[h!]
    \begin{minipage}[b]{0.45\linewidth}
        \centering
        \includegraphics[width=1.1\textwidth]{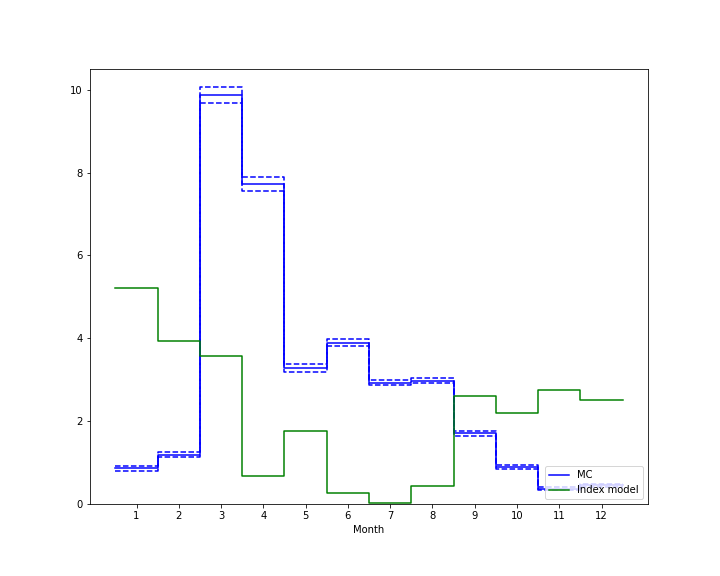}
    \end{minipage}
    \begin{minipage}[b]{0.45\linewidth}
        \centering
        \includegraphics[width=1.1\textwidth]{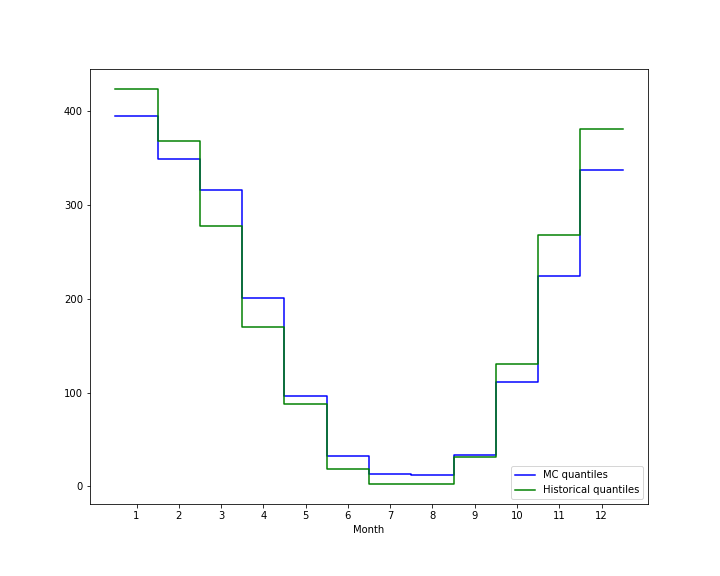}
    \end{minipage}
    \caption{On the right, expected payoffs forecasted 30 days ahead for a derivative on each month of 2019 and $HDD_{strike}$ defined as the 90\% historical quantile for both approaches. On the left, a comparison of simulated and historical 90\% quantile.}
    \label{fig:quantile_index_model}
\end{figure}

To sum up, there exist a clear coherence between the pricing methodology followed in this document and the current market practices as both give prices in the same ranges. However, our approach enables to better quantify the sources of risk, which is crucial  in  this growing market.

\bibliographystyle{abbrv} 

\bibliography{references.bib} 

%\newpage

\appendix

\section{{Weather station data description}} \label{appendix:weather_station}

\begin{table}[h!]
\centering
\begin{tabular}{l|lllll}
{ City} & { WMO} & { Latitude} & { Longitude} & { Elevation} & { Original Source} \\ \hline
{ Stockholm} & { 2485} & { 59.34} & { 18.05} & { 43 m} & { \begin{tabular}[c]{@{}l@{}}Swedish Meteorological and \\ Hydrological Institute, SMHI\end{tabular}} \\
{ \begin{tabular}[c]{@{}l@{}}Paris Charles \\ de Gaulle\end{tabular}} & { 7157} & { 49.02} & { 2.53} & { 109 m} & {ASOS-METAR } \\
{ \begin{tabular}[c]{@{}l@{}}Amsterdam \\ AP Schiphol\end{tabular}} & { 6240} & { 52.3} & { 4.78} & { -4 m} & { \begin{tabular}[c]{@{}l@{}}Royal Netherlands \\ Meteorological Institute, KNMI\end{tabular}} \\
{ Berlin Tempelhof} & { 10384} & { 52.47} & { 13.4} & { 50 m} & { Deutsche Wetterdienst, DWD} \\
{ Brussels National} & { 6451} & { 50.9} & { 4.53} & { 58 m} & { \begin{tabular}[c]{@{}l@{}}Royal Meteorological \\ Institute of Belgium\end{tabular}} \\
{ London Heathrow} & { 3772} & { 51.48} & { -0.45} & { 25 m} & { ASOS-METAR} \\
{ Rome Ciampino} & { 16239} & { 41.78} & { 12.58} & { 105 m} & { ASOS-METAR } \\
{ Madrid Barajas} & { 8221} & { 40.5} & { -3.58} & { 633 m} & { ASOS-METAR}
\end{tabular}
\caption{{Characteristics of the weather stations providing temperature data. Speedwell explicitly states the meteorological agency the data is extracted from except for ASOS-METAR data. ASOS-METAR is in charge of the monitoring of airport weather stations following the standards of the International Civil Aviation Organization (ICAO) and the World Meteorological Organization (WMO).}}
\label{tab:weather_stations}
\end{table}

\section{{Model~\eqref{stoch} simulation and estimator testing}} \label{appendix:method_simu}

{This section explains the algorithm to simulate Model~\eqref{stoch}. Simulation is first performed to test robustness of the estimation through}  the following steps:

\begin{enumerate}
    \itemsep0em 
    \item Estimate $\kappa$ and the seasonality $s$ from temperature data.
    \item Estimate the parameters $K$ and $\sigma^2$ and then $\eta^2$.
    \item Fix these parameters at the estimated values.
    \item Generate a simulated instantaneous volatility series $\zeta$ based on a generalized Ninomiya-Victoir scheme for Cox-Ingersoll-Ross (CIR) processes, 
    and the corresponding temperature series $T$.
    \item Estimate $(\alpha_0, \beta_0, \alpha_1, \beta_1, \kappa)$ on the simulated data and compare with the fixed values.
    \item Compute realized volatility $\hat{\zeta}$  for different time lags $Q$.
    \item Estimate $(\gamma_0, \gamma_1, \delta_1, \gamma_2, \delta_2, K)$ then $\widehat{\eta^2}$, and compare with the fixed values.
\end{enumerate}

\noindent The generation of the combined discrete series $(T_{i\Delta}, \zeta_{i\Delta})_{i \in \N}$ is performed thanks to the recurrence formula:
\begin{equation}\label{schemes} 
\left \{
   \begin{aligned}  
    T_{(i+1)\Delta} &= s((i+1)\Delta) + e^{-\kappa \Delta} (T_{i\Delta} - s(i\Delta)) +  \sqrt{\frac{1- e^{-2\kappa \Delta}}{2\kappa}\frac{\zeta_{i\Delta}+\zeta_{(i+1)\Delta}}2} Z_i\\
    \zeta_{(i+1)\Delta} &= \phi(\zeta_{i \Delta}, \Delta, \sqrt{\Delta} Y_{i}),
    \end{aligned}
    \right.
\end{equation}
{where $(Y_i,Z_i)_{i\ge 0}$ is an i.i.d. sequence of two independent standard normal variables.}
Remember that we assume here and in the sequel of the paper that $\rho=0$.
The first row of~\eqref{schemes} corresponds to a discretization of the integral following the temperature dynamics in Model \eqref{stoch} for a step $\Delta$. {The use of the trapezoidal rule comes from the operator splitting method and allow to get a second order scheme as in~\cite[Eq.~(4.31)]{AAbook}.}
The second row of \eqref{schemes} is the Ninomiya-Victoir scheme for CIR processes{~\cite{NV}} when freezing the time-dependent coefficients at time $(i+1/2)\Delta$, which preserves the convergence of order~2 of this scheme{, see~\cite[Paragraph 3.3.4]{AAbook}}. In general (i.e. when $K \sigma^2((i+1/2)\Delta) \geq \frac{\eta^2}{4}$), $\phi$ corresponds to: 

\begin{equation*}
\begin{aligned}
    \phi(\zeta_{i \Delta}, \Delta, \sqrt{\Delta} Y_{i}) &= e^{-\frac{K \Delta}{2}} \left( \sqrt{\left(K \sigma^2((i+1/2) \Delta)-\frac{\eta^2}{4}\right) \psi_K\left(\frac{\Delta}{2}\right)+ \zeta_{i \Delta}e^{-\frac{K \Delta}{2}}} + \frac{\eta}{2} \sqrt{\Delta} Y_{i} \right)^2  \\
    & + \left(K \sigma^2((i+1/2)\Delta)- \frac{\eta^2}{4} \right)\psi_K\left(\frac{\Delta}{2}\right)  
\end{aligned}
\end{equation*}

\noindent where $\psi_K(t) = \frac{1- e^{-K t}}{K}$. The case where $K \sigma^2((i+1/2)\Delta) < \frac{\eta^2}{4}$ is handled as in Alfonsi \cite{alfonsi2010high}, so that~\eqref{schemes} is a second order scheme for the weak error, and thus an accurate approximation of the exact law.

\begin{remark}
With this simulation method, it is possible to introduce additional granularity into our simulated processed having more than one point per day, that is $\Delta < 1$. We have analysed numerically if this extra-granularity has an incidence on the parameter estimation. Namely, we have calculated the estimators on simulated paths with different values of $\Delta$ but with the same number of points. We have noticed that $\Delta$ has a small influence on the estimators, and we do not reproduce these experiments in the paper.  
\end{remark}

\section{CLS estimators of the temperature process} \label{appendix:T}

Let us consider the following dynamics for the temperature $(T_t)_{t \geq 0}$:

\begin{equation}
  \begin{cases}
  T_t&=s(t)+\tilde{T}_t,\\
  d \tilde{T}_t&=- \kappa \tilde{T}_t dt +\sqrt{\zeta_t}(\rho dW_t +\sqrt{1-\rho^2} dZ_t)
  \end{cases} \label{dyn_T}
\end{equation}

\noindent where $\kappa>0$, $s(t) = \alpha_0 + \beta_0 t + \alpha_{1} \sin(\xi t) + \beta_{1} \cos(\xi t) $, $W$ and $Z$ are two independent Brownian motions and $\zeta_t$ is a nonnegative adapted process {such that $\E[\int_0^t \zeta_s de]<\infty$ for all $t>0$}. The goal of this appendix is to compute the conditional least squares estimators of $(\kappa, \alpha_0,\beta_0, \alpha_1, \beta_1 )$ {and to prove  the next proposition.}

\begin{prop}\label{lambda_def}
{Let $X_{i\Delta} = (1, i\Delta,  T_{i\Delta}, \sin(\xi i\Delta), \cos(\xi i \Delta))^T \in \R^5$ for $i \in \N$ with $(T_t)_{t \ge 0}$ following the dynamics of~\eqref{dyn_T} and $\Delta>0$. We assume that $\sum_{i=0}^{N-1} X_{i\Delta} X^T_{i\Delta}$ is invertible and define 
 \begin{equation}\label{def_hlambda}
\hat{\lambda} {= (\hat{\lambda}_0,\dots \hat{\lambda}_4)^T}=\left( \sum_{i=0}^{N-1} X_{i\Delta} X^T_{i\Delta}  \right)^{-1} \left(\sum_{i=0}^{N-1} X_{i\Delta}  T_{(i + 1)\Delta} \right). 
\end{equation} 
If $\hat{\lambda}_2 \in (0,1)\cup (1,+\infty)$}, the solution of the minimisation problem
\begin{mini}|s|
{{(\kappa,\alpha_0,\alpha_1,\beta_0,\beta_1})\in \R^5}{\sum_{i=0}^{N-1} \left( T_{(i + 1)\Delta} - \mathbb{E}  [T_{(i + 1)\Delta} | T_{i\Delta} ] \right)^2}{}{}\label{minim_T}
\end{mini}

\noindent is given by

\begin{equation}
\begin{cases}\label{def_estim_T}
  \hat{\kappa} &= - \frac{1}{\Delta}\ln{\hat{\lambda} _2} \\
  \hat{\alpha}_0&= \frac{\hat{\lambda}_0 }{1-\hat{\lambda}_2} -  \frac{\hat{\lambda}_1 \Delta}{(1-\hat{\lambda}_2)^2}\\
  \hat{\beta}_0&= \frac{\hat{\lambda}_1}{1-\hat{\lambda}_2} \\
  \hat{\alpha}_1 &= \frac{\hat{\lambda}_3 (\cos(\xi \Delta)-e^{-\hat{\kappa}\Delta}) + \hat{\lambda}_4 \sin(\xi \Delta)}{(\cos(\xi \Delta)-e^{-\hat{\kappa}\Delta})^2 + \sin^2(\xi \Delta)}\\
  \hat{\beta}_1 &= \frac{\hat{\lambda}_4 (\cos(\xi \Delta)-e^{-\hat{\kappa} \Delta}) - \hat{\lambda}_3 \sin(\xi \Delta)}{(\cos(\xi \Delta)-e^{-\hat{\kappa} \Delta})^2 + \sin^2(\xi \Delta)}.\\
\end{cases}
\end{equation}
\label{prop-ap1}
\end{prop}

\begin{proof} Recall that {$T_t=s(t)+\tilde{T}_t$}. {Applying Ito's formula for $(e^{\kappa t}\tilde{T}_t)_{t \geq 0}, we $} have 

\begin{equation}\label{integral}
\tilde{T}_{t+\Delta}=\tilde{T}_t e^{-\kappa \Delta}+ \rho \int_t^{t+\Delta}e^{-\kappa (t+\Delta-s)}\sqrt{\zeta_s} dW_s + \sqrt{1-\rho^2} \int_t^{t+\Delta}e^{-\kappa (t+\Delta-s)}\sqrt{\zeta_s} dZ_s 
\end{equation}

\noindent From the martingale property of the stochastic integral {(we have $\int_t^{t+\Delta} \E[\zeta_s]ds<\infty$  since $\E[\zeta_s]=\zeta_0e^{-Ks}+\int_0^s e^{-K(s_u)}\sigma^2(u)du$)}, we have $\E[\tilde{T}_{t+\Delta}|\cF_t]=\tilde{T}_t e^{-\kappa \Delta}$ and thus
$$\E[T_{t+\Delta}|\cF_t]=T_t e^{-\kappa\Delta} +s(t+\Delta)-s(t)e^{-\kappa\Delta}.$$

\noindent We now use trigonometric identities to get

\begin{align*}
s(t+\Delta)-e^{-\kappa\Delta}s(t)=& \alpha_0 + \beta_0 (t + \Delta) - \alpha_0 e^{-\kappa  \Delta} - \beta_0 e^{-\kappa  \Delta}t  + \alpha_1 \sin(\xi (t+\Delta)) - \alpha_1 e^{-\kappa \Delta} \sin(\xi t) \\
& + \beta_1 \cos(\xi (t+\Delta)) - \beta_1 e^{-\kappa \Delta} \cos(\xi t) \\
%=& \alpha_0 (1- e^{-\kappa  \Delta}) + \beta_0 \Delta + \beta_0 (1- e^{-\kappa  \Delta}) t  + (\alpha_1  \cos(\xi\Delta) - \alpha_1 e^{-\kappa \Delta} - \beta_1 \sin(\xi \Delta)  )\sin(\xi t) \\
%& + ( \alpha_1  \sin (\xi\Delta) + \beta_1 \cos(\xi \Delta) - \beta_1 e^{-\kappa \Delta} ) \cos(\xi t) \\
&= \lambda_0 + \lambda_1 t + \lambda_3 \sin(\xi t) + \lambda_4  \cos(\xi t),
\end{align*}

\noindent with

\begin{equation}
  \begin{cases}
  \lambda_0 &= \alpha_0(1-e^{-\kappa\Delta}) + \beta_0 \Delta\\
  \lambda_1 &= \beta_0(1-e^{-\kappa\Delta}) \\
  \lambda_2 &= e^{-\kappa\Delta} \\
  \lambda_3 &=\alpha_1 (\cos(\xi \Delta)-e^{-\kappa\Delta})-\beta_1\sin (\xi \Delta) \\
  \lambda_4 &=\alpha_1 \sin(\xi \Delta)+\beta_1(\cos (\xi \Delta)- e^{-\kappa\Delta}),
  \end{cases} \label{lambda-T}
\end{equation}
where $\lambda_2$ is set to have {$\E[{T}_{(i+1)\Delta}|\cF_t]=\lambda^T X_{i\Delta}$}. The minimization problem~\eqref{minim_T} is then equivalent to
\begin{mini*}|s|
{\lambda \in \R^5}{\sum_{i=0}^{N-1} \left( T_{(i + 1)\Delta} - \lambda^T X_{i\Delta}   \right)^2}{}{}.
\end{mini*}

%\noindent where $\lambda = (\lambda_0,\lambda_1, \lambda_2, \lambda_3, \lambda_4) $ and $X_{i\Delta} = (1, i\Delta,  T_{i\Delta}, \sin(\xi i\Delta), \cos(\xi i \Delta))^T$.

\noindent This corresponds to a linear regression, whose solution is given by~\eqref{def_hlambda}. When $\lambda_2 \in (0,1)$, the system~\eqref{lambda-T} can be inverted, and the claim follows easily. 

\end{proof}
\noindent Let us note here that $\hat{\lambda}^T X_{i\Delta}$ can then be seen as the estimation of $\mathbb{E}  [T_{(i+1) \Delta} | T_{i \Delta} ]$.

%\newpage

\section{CLS estimators of the volatility process} \label{appendix:zeta}

Let consider the volatility of the temperature $(\zeta_t)_{t \geq 0}$ follows the below dynamics:

\begin{equation} \label{zeta-dyn}
  d \zeta_t= - K (\zeta_t- \sigma^2(t))dt +  \eta \sqrt{\zeta_t} dW_t
\end{equation}

\noindent where $K>0$, $\sigma^2$ is {a nonnegative function with the parametric form} given by~\eqref{sigma} and $W$ is a Brownian motion. The goal of this appendix is to compute the conditional least squares estimators of {the parameters} $(\gamma_0, K  ,\gamma_1,\dots, \gamma_{K_{\sigma^2}}, \delta_1, \dots,\delta_{K_{\sigma^2}})$. 

\begin{prop} \label{prop-zeta}
{Let $X'_{i\Delta} = (1, \zeta_{i\Delta} , \sin(\xi_1 i\Delta),\dots, \sin( \xi_{K_{\sigma^2}} i\Delta), \cos(\xi_1 i \Delta), \dots,\cos(\xi_{K_{\sigma^2}} i \Delta))^T$ with {$(\zeta_t)_{t \geq 0}$} following the dynamics~\eqref{zeta-dyn} and $\Delta>0$. We assume that $\sum_{i=0}^{N-1} X'_{i\Delta} X'^T_{i\Delta}$ is invertible and define $$(\hat{\theta}_0,\hat{\phi}_0,\hat{\theta}_1,\dots,\hat{\theta}_{K_{\sigma^2}},\hat{\phi}_1,\dots,\hat{\phi}_{K_{\sigma^2}})^T =  \left(\sum_{i=0}^{N-1} X'_{i\Delta} X'^T_{i\Delta} \right)^{-1} \left(\sum_{i=0}^{N-1} X'_{i\Delta}  \zeta_{(i + 1)\Delta}\right). $$ If $\hat{\phi}_0\in (0,1)\cup (1,\infty)$,} the solution of the minimisation problem 

\begin{mini}|s|
{{K \in \R ,\gamma \in \R^{K_{\sigma^2}+1}, \delta \in \R^{K_{\sigma^2}}}}{\sum_{i=0}^{N-1} \left( \zeta_{(i+1)\Delta} - \mathbb{E}  [\zeta_{(i +1) \Delta} | \zeta_{i \Delta} ] \right)^2}{}{} \label{minim_zeta}
\end{mini}
is given by

\begin{equation}\label{estim_vol}
\begin{dcases}
   \begin{aligned}
  \hat{\gamma}_0 &=  \frac{\hat{\theta}_0 }{1 - \hat{\phi}_0 } \\
  \hat{K}  &=-\frac{1}{\Delta}\ln(\hat{\phi}_0) \\
  \hat{\gamma}_k& = \frac{\hat{\theta}_k D_k- \hat{\phi}_k B_k}{A_k D_k- C_k B_k} \\
  \hat{\delta}_k& = \frac{\hat{\theta}_k C_k- \hat{\phi}_k A_k}{C_k B_k- A_k D_k}
    \end{aligned}
\end{dcases}
\end{equation}

\noindent where, for $k\in \{1,\dots,  K_{\sigma^2}\}$,

\begin{equation}
\begin{dcases}
   \begin{aligned}
  A_k& =  \hat{K} \frac{\hat{K} (\cos(\xi_k \Delta)-e^{-\hat{K} \Delta}) + \xi_k \sin(\xi_k \Delta)}{\hat{K}^2+\xi_k^2}\\
  B_k&=  - \hat{K} \frac{\hat{K} \sin(\xi_k \Delta)-\xi_k ( \cos(\xi_k \Delta) - e^{-\hat{K} \Delta})}{\hat{K}^2+\xi_k^2} \\
  C_k &=  \hat{K} \frac{\hat{K} \sin(\xi_k \Delta)-\xi_k ( \cos(\xi_k \Delta) - e^{-\hat{K}  \Delta})}{\hat{K} ^2+\xi_k^2} \\
  D_k &= \hat{K} \frac{\hat{K} (\cos(\xi_k \Delta)-e^{-\hat{K}\Delta}) +\xi_k  \sin(\xi_k \Delta)}{\hat{K}^2+\xi_k^2}.
  \end{aligned}
\end{dcases}
\end{equation}
\end{prop}

\noindent For the proof of Proposition~\ref{prop-zeta} we will need Lemma \ref{lem_integ}, whose proof is straightforward. 

\begin{lemma}\label{lem_integ} 
  For $K^2+\xi^2>0$ and $k\in \N^*$, we have
  \begin{align*}
    \int_t^{t+\Delta}e^{-K(t+\Delta-s)}\cos(\xi_k s)ds&=\cos(\xi_k t)\frac{K[\cos(\xi_k \Delta)-e^{-K\Delta}] +\xi_k \sin(\xi_k \Delta)}{K^2+\xi_k^2}\\
    &-\sin(\xi_k t) \frac{K\sin(\xi_k \Delta) - \xi_k ( \cos(\xi_k \Delta)  -e^{-K \Delta})}{K^2+\xi_k^2}, \\
    \int_t^{t+\Delta}e^{-K(t+\Delta-s)}\sin(\xi_k s)ds&=\sin(\xi_k t) \frac{K[\cos(\xi_k \Delta)-e^{-K\Delta}] + \xi_k \sin(\xi_k \Delta)}{K^2+\xi_k^2} \\
    &+ \cos(\xi_k t)\frac{K\sin(\xi_k \Delta)-\xi_k (\cos(\xi_k \Delta) - e^{-K \Delta})}{K^2+\xi_k^2}.
  \end{align*}
\end{lemma}

\begin{proof}[Proof of Proposition~\ref{prop-zeta}]
{Applying Ito's formula for $(e^{K t}\zeta_t)_{t \geq 0}$, we have:} 

\begin{equation} \label{zeta_integral}
    \zeta_{t+\Delta}=\zeta_t e^{-K \Delta}+K\int_t^{t+\Delta}e^{-K(t+\Delta-s)}\sigma^2(s)ds + \eta \int_{t}^{t + \Delta} e^{-K (t+ \Delta -s)} \sqrt{\zeta_s} dW_s
\end{equation}

\noindent Hence, we get
$$\E[\zeta_{t+\Delta}|\cF_t]=\zeta_t e^{-K \Delta}+K\int_t^{t+\Delta}e^{-K(t+\Delta-s)}\sigma^2(s)ds.$$

\noindent From~\eqref{sigma} and Lemma~\ref{lem_integ}, we then obtain {for $t \geq 0$ and $\zeta_0 \geq 0$,}
\begin{equation}\label{formule_condexp_zeta}
  \E[\zeta_{t+\Delta}|\cF_t]=\theta_0 + \phi_0 \zeta_t + \sum_{k=1}^{K_{\sigma^2}} \theta_k \sin(\xi_k t) + \sum_{k=1}^{K_{\sigma^2}} \phi _k \cos( \xi_k t),
\end{equation}

\noindent with

$$\begin{dcases*}
  \begin{aligned}
  \theta_0&= \gamma_0(1-e^{-K\Delta}) \\
  \phi_0 &= e^{-K \Delta} \\
  \theta_k &= \gamma_k K \frac{K[\cos(\xi_k \Delta)-e^{-K\Delta}] + \xi_k \sin(\xi_k \Delta)}{K^2+\xi_k^2}- \delta_k K \frac{K\sin(\xi_k \Delta)-\xi_k( \cos(\xi_k \Delta) - e^{-K \Delta})}{K^2+\xi_k^2} \\
  \phi_k&= \gamma_k K \frac{K\sin(\xi_k \Delta)-\xi_k ( \cos(\xi_k \Delta) - e^{-K \Delta})}{K^2+\xi_k^2} +\delta_k K \frac{K[\cos(\xi_k \Delta)-e^{-K\Delta}] +\xi_k \sin(\xi_k \Delta)}{K^2+\xi_k^2}
  \end{aligned}
\end{dcases*}$$

\noindent We can invert the above system {when $\phi_0\in(0,1)\cup(1,+\infty)$} to get the formulas of $K, \gamma_0, \gamma_k, \delta_k${.}

\begin{equation} \label{invert-zeta}
\begin{dcases}
   \begin{aligned}
  \gamma_0&=  \frac{\theta_0 }{1 - \phi_0 } \\
  K  &=-\frac{1}{\Delta}\ln(\phi_0) \\
  \gamma_k& = \frac{\theta_k D_k- \phi_k B_k}{A_k D_k- C_k B_k} \\
  \delta_k& = \frac{\theta_k C_k- \phi_k A_k}{C_k B_k- A_k D_k}
    \end{aligned}
\end{dcases}
\end{equation}

\noindent  where

\begin{equation*}
\begin{dcases}
   \begin{aligned}
  A_k& =  K \frac{K(\cos(\xi_k \Delta)-e^{-K\Delta}) + \xi_k \sin(\xi_k \Delta)}{K^2+(\xi_k)^2}\\
  B_k&=  - K \frac{K\sin(\xi_k \Delta)-\xi_k ( \cos(\xi_k \Delta) - e^{-K \Delta})}{K^2+(\xi_k)^2} \\
  C_k &=  K \frac{K\sin(\xi_k \Delta)-\xi_k ( \cos(\xi_k \Delta) - e^{-K \Delta})}{K^2+(\xi_k)^2} \\
  D_k &= K \frac{K(\cos(\xi_k \Delta)-e^{-K\Delta}) +\xi_k  \sin(\xi_k \Delta)}{K^2+(\xi_k)^2}
  \end{aligned}
\end{dcases}
\end{equation*}

\noindent {Let us observe that $A_kD_k-C_kB_k\ge 0$ as a sum of squares. We even have $A_kD_k-C_kB_k> 0$ since $(K(\cos(\xi_k \Delta)-e^{-K\Delta}) + \xi_k \sin(\xi_k \Delta),K\sin(\xi_k \Delta)-\xi_k ( \cos(\xi_k \Delta) - e^{-K \Delta}))\not=(0,0)$ as $$\det\left[\begin{matrix}
    (\cos(\xi_k \Delta)-e^{-K\Delta}) & \sin(\xi_k \Delta) \\
    \sin(\xi_k \Delta) & (\cos(\xi_k \Delta)-e^{-K\Delta}) 
\end{matrix}\right]=1-2\cos(\xi_k \Delta)e^{-K\Delta}+e^{-2K\Delta}\ge (1-e^{-K\Delta})^2>0,$$ 
since $K\not = 0$  by the assumption $\phi_0 \not=1$.}
The minimization problem~\eqref{minim_zeta} is then equivalent to 
\begin{mini*}|s|
{{\vartheta}}{\sum_{i=0}^{N-1} \left( \zeta_{(i + 1)\Delta} - {\vartheta}^T X'_{i\Delta} \right)^2}{}{}
\end{mini*}
\noindent where ${\vartheta} = (\theta_0,\phi_0,\theta_1,\dots,\theta_{K_{\sigma^2}},\phi_1,\dots,\phi_{K_{\sigma^2}})^T$ and \begin{equation}\label{def_Xprime}
X'_{i\Delta} = (1, \zeta_{i\Delta} , \sin(\xi_1 i\Delta),\dots, \sin(\xi_{K_{\sigma^2}} i\Delta), \cos(\xi_1 i \Delta),\dots, \cos(\xi_{K_{\sigma^2}} i \Delta))^T, \ 0\le i\le N-1. \end{equation}

\noindent This problem corresponds to a simple multilinear regression problem. Its solution is given by:
$$\hat{\vartheta} = \left( \sum_{i=0}^{N-1} X'^T_{i\Delta} X'_{i\Delta} \right)^{-1}  \left(\sum_{i=0}^{N-1} X'^T_{i\Delta}  \zeta_{(i + 1)\Delta} \right) $$
Combined with Equation~\eqref{invert-zeta}, we get the estimators~\eqref{estim_vol} of the volatility parameters of of Model~\eqref{stoch}.
\end{proof}
\noindent Let us note here that $\hat{\vartheta}^T X'_{i\Delta}$ can be seen as the estimation of $\E[\zeta_{(i+1)\Delta}|\zeta_{i\Delta}]$.

%\newpage

\section{Computation of the CLS estimators of $\eta^2$ and $\rho$} \label{appendix:eta2rho}

Let consider the volatility of the temperature $(\zeta_t)_{t \geq 0}$ follows the below dynamics {for $t \geq 0$ and $\zeta_0 \geq 0$:}

\begin{equation} \label{zeta_dyn}
  d \zeta_t= - K (\zeta_t- \sigma^2(t))dt +  \eta \sqrt{\zeta_t} dW_t,
\end{equation}

\noindent as in Model~\eqref{stoch} with $ \sigma^2(t) = \gamma_0 + \sum_{k=1}^{K_{\sigma^2}} \gamma_{k} \sin(\xi_k t ) + \sum_{k=1}^{K_{\sigma^2}} \delta_{k} \cos(\xi_k t )$. We first focus on the conditional least squares estimator of the volatility of the volatility $\eta^2$, and assume that the coefficients $K$ and $\sigma^2(\cdot)$ are known. 

\begin{prop} \label{prop-eta2}
Let $(\zeta_t)_{t \ge 0}$ follow the dynamics~\eqref{zeta_dyn} {with $\sigma^2(t)$ being a nonnegative function of the form~\eqref{sigma} and $\Delta >0$. Then, we have for $i\in \N$ 
$$Y_{i\Delta } = \mathbb{E}  [ \left(\zeta_{(i + 1)\Delta} - \mathbb{E}  [\zeta_{(i +1) \Delta} | \zeta_{i\Delta} ] \right)^2| \zeta_{i\Delta} ]  = \theta_0' + \phi_0' \zeta_{i\Delta }  + \sum_{k=1}^{K_{\sigma^2}} \theta_k' \sin(\xi_k i\Delta ) + \sum_{k=1}^{K_{\sigma^2}} \phi_k' \cos(\xi_k i\Delta )>0,$$ with $\theta'$ and $\phi'$ defined by~\eqref{def_thetaphiprime}.} The solution of the minimisation problem 

\begin{mini}|s|  
{\eta^2 {\ge 0 }}{\sum_{i=0}^{N-1} \left( (\zeta_{(i+1)\Delta}- \mathbb{E}  [\zeta_{(i +1) \Delta} | \zeta_{i \Delta} ])^2 - \mathbb{E}  [ \left(\zeta_{(i + 1)\Delta} - \mathbb{E}  [\zeta_{(i +1) \Delta} | \zeta_{i\Delta} ] \right)^2| \zeta_{i\Delta} ]  \right)^2}{}{}
\label{mini:min-eta2}
\end{mini}

\noindent is given by 

\begin{equation}\label{eta_opt}
   \widehat{\eta^2} =  \frac{\sum_{i=0}^{N-1}  Y_{i\Delta }(\zeta_{(i+1)\Delta}- \vartheta^T X'_{i\Delta})^2}{\sum_{i=0}^{N-1} Y_{i\Delta }^2},
\end{equation}
where $\vartheta$ and $X'_{i\Delta}$ are defined by~\eqref{def_Xprime}.
%equivalent to \begin{mini}|s|{\eta^2}{\sum_{i=0}^{N-1} \left( (\zeta_{(i+1)\Delta}- \mathbb{E}  [\zeta_{(i + 1)\Delta} | \zeta_{i\Delta} ])^2 - \eta^2 Y_{i\Delta }  \right)^2}{}{ }\end{mini}\noindent where ,\noindent provides the following estimator for the volatility of the volatility parameter $\eta^2$ of Model~\eqref{stoch}\noindent where $\hat{\theta} X'_{(i+1)\Delta}$ corresponds to the estimate of $\mathbb{E}  [\zeta_{(i +1) \Delta} | \zeta_{i \Delta} ]$ as defined in Appendix \ref{appendix:zeta}. 
\end{prop}

\noindent For the proof of Proposition~\ref{prop-eta2}, we first state Lemma~\ref{lem_integ2}, which is a straightforward generalisation of~\eqref{formule_condexp_zeta}. 

\begin{lemma}  \label{lem_integ2}

For all $s\ge t$,

$$\E[\zeta_s|\cF_t]=\zeta_te^{-K(s-t)}+ \gamma_0(1-e^{-K(s-t)})+ \sum_{k=1}^{K_{\sigma^2}} \Theta_k(s-t)\sin(\xi_k t)+ \sum_{k=1}^{K_{\sigma^2}} \Phi_k(s-t) \cos(\xi_k t),$$

\noindent with
\begin{align*}
  \Theta_k(v)=\gamma_k K \frac{K[\cos(\xi_k v)-e^{-Kv}] +\xi_k \sin(\xi_k v)}{K^2+\xi_k^2}- \delta_k K \frac{K\sin(\xi_k v)-\xi_k ( \cos(\xi_k v) - e^{- K v})}{K^2+\xi_k^2}, \\
  \Phi_k(v)= \gamma_k K \frac{K\sin(\xi_k v)-\xi_k (\cos(\xi_k v)-e^{-K v})}{K^2+\xi_k^2} +\delta_k K \frac{K[\cos(\xi_k v)-e^{-K v}] +\xi_k \sin(\xi_k v)}{K^2+\xi_k^2}.
\end{align*}
\end{lemma}

\begin{proof}[Proof of Proposition~\ref{prop-eta2}]

\noindent {Let recall Equation~\eqref{zeta_integral} and compute the conditional variance of  $(\zeta_t)_{t \geq 0}$:}

$$\zeta_{t+\Delta}=\zeta_t e^{-K \Delta}+K\int_t^{t+\Delta}e^{-K(t+\Delta-s)}\sigma^2(s)ds + \eta \int_{t}^{t + \Delta} e^{-K (t+ \Delta -s)} \sqrt{\zeta_s} dW_s$$

\begin{equation*}
   \begin{aligned}
  \mathbb{E}  \left[ \left(\zeta_{t + \Delta} - \mathbb{E}  [\zeta_{t + \Delta} | \zeta_{t} ] \right)^2| \cF_{t} \right]  &= \mathbb{E}  \left[ \left(\eta \int_{t}^{t + \Delta} e^{-K (t+ \Delta -s)} \sqrt{\zeta_s} dW_s \right)^2| \cF_{t} \right] \\
  &=\eta^2 \int_{t}^{t + \Delta} e^{-2 K (t+ \Delta -s)} \mathbb{E}  [ \zeta_s | \zeta_{t} ]ds 
    \end{aligned}
\end{equation*}

\noindent by Ito's isometry, and using that $\E[\zeta_s|\cF_t]=\E[\zeta_s|\zeta_t]$.  From Lemmas \ref{lem_integ2} {and~\ref{lem_integ}}, we deduce that

\begin{align*}
  \int_t^{t+\Delta}e^{-2K(t+\Delta-u)} \E[\zeta_u|\zeta_t]du&= \theta_0' + \phi_0' \zeta_t+{\sum_{k=1}^{K_{\sigma^2}}} \theta'_k \sin(\xi_k t)+ {\sum_{k=1}^{K_{\sigma^2}}} \phi'_k \cos( \xi_k t),
\end{align*}

\noindent where

\begin{equation}\label{def_thetaphiprime}
\begin{dcases}
   \begin{aligned}
  \theta_0' &= \gamma_0 \frac{(1  - e^{- K \Delta})^2 }{2K}\\
  \phi_0' &=  \frac{e^{-K \Delta}}{K}  (1- e^{- K \Delta}) \\
  \theta'_k &=\gamma_k K \frac{K[A_k'-\phi_0'] +\xi_k B_k'}{K^2+\xi_k^2}- \delta_k K \frac{K B_k' -\xi_k [A_k'- \phi_0']}{K^2+\xi_k^2} \\
  \phi'_k &= \gamma_k K \frac{K B_k'-\xi_k [A_k'-\phi_0']}{K^2+\xi_k^2} +\delta_k K \frac{K[A_k'-\phi_0'] +\xi_k B_k'}{K^2+\xi_k^2},   
    \end{aligned}
\end{dcases}
\end{equation}

\noindent and

\begin{equation} 
\begin{dcases}
   \begin{aligned}
   A_k'&=\frac{2K[\cos(\xi_k  \Delta)-e^{- 2K  \Delta}] +\xi_k \sin(\xi_k  \Delta)}{4K^2+\xi_k^2}\\
   B_k'&=\frac{2K\sin(\xi_k  \Delta)-\xi_k( \cos(\xi_k  \Delta)- e^{- 2 K  \Delta})}{4K^2+\xi_k^2}.
   %\\   C'&=e^{-K  \Delta} \frac{1-e^{-K  \Delta}}{K} .
    \end{aligned}
\end{dcases}
\end{equation}
{For example, $\theta'_k=\int_t^{t+\Delta} e^{-2K(t+\Delta-u)} \Theta_k(u-t)du=\int_0^{\Delta} e^{-2K(\Delta-u)}\Theta_k(u)du$, and we use then Lemma~\ref{lem_integ} to get the formulas for $A'_k=\int_0^{\Delta} e^{-2K(\Delta-u)}\cos(\xi_k u)du$ and $B'_k=\int_0^{\Delta} e^{-2K(\Delta-u)}\sin(\xi_k u)du$. The calculation of $\psi'_k$ works in the same way.}

\noindent Hence, we get

\begin{equation*}
  \mathbb{E}  [ \left(\zeta_{t + \Delta} - \mathbb{E}  [\zeta_{t + \Delta} | \zeta_{t} ] \right)^2| \zeta_{t} ]  =  \eta^2 \left(\theta_0' + \phi_0' \zeta_{t}  + {\sum_{k=1}^{K_{\sigma^2}}} \theta_k' \sin(\xi_k t) + {\sum_{k=1}^{K_{\sigma^2}}} \phi_k' \cos(\xi_k t)\right),
\end{equation*} 
{and hence the value of $Y_{i\Delta}$. We also observe that $Y_{i\Delta}>0$ since  $e^{-2K((i+1)\Delta-u)}>0$ and $\E[\zeta_u|\zeta_{i\Delta}]=\zeta_{i\Delta}e^{-K(u-i\Delta)}+K \int_{i \Delta}^{u}e^{-K(u-s)}\sigma^2(s) ds\ge K \int_{i \Delta}^{u}e^{-K(u-s)}\sigma^2(s) ds >0$ for $u\in (i\Delta,(i+1)\Delta)$. }

\noindent Now, let come back to the considered minimisation problem~\eqref{mini:min-eta2}. We consider $\vartheta$ and $X'_{i\Delta}$ defined by~\eqref{def_Xprime} and define  {$Y_{i \Delta}$ as in Proposition~\ref{prop-eta2}} $=  \theta_0' + \phi_0' \zeta_{i \Delta}  + \sum_k \theta_k' \sin(\xi_k i \Delta) + \sum_k \phi_k' \cos(\xi_k i \Delta)$. Thus, {from Proposition~\ref{prop-zeta}, }we have $\E[\zeta_{(i+1)\Delta} |\zeta_{i\Delta}]=\vartheta^T X'_{i\Delta}$ and $\E[( \zeta_{(i+1)\Delta}-\E[\zeta_{(i+1)\Delta} |\zeta_{i\Delta}] )^2|\zeta_{i\Delta}]= {\eta^2} Y_{i \Delta}$. Problem~\eqref{mini:min-eta2} is then  equivalent to

\begin{mini*}|s|
{\eta^2}{\sum_{i=0}^{N-1} \left( (\zeta_{(i+1)\Delta}- \vartheta^T X'_{i\Delta})^2 - \eta^2 Y_{i \Delta} \right)^2}{}{ },
\end{mini*}
whose solution is given by~\eqref{eta_opt}. 
\end{proof}

%\newpage

%\section{Computation of the estimator of $\rho$} \label{appendix:rho}
We now focus on the conditional least squares estimator of the correlation $\rho$ for Model~\eqref{stoch} and assume that the coefficients $\kappa$, $s(\cdot)$, $K$, $\sigma^2(\cdot)$ and $\eta^2$  are known. 

\begin{prop} \label{prop-rho}
Let $(T_t,\zeta_t)_{t \ge 0}$  follow~\eqref{stoch} {with $\sigma^2(t)$ being a nonnegative function of the form~\eqref{sigma} and $\Delta >0$. Then, we have for $i\in \N$ 
$$Y'_{i  \Delta} =  \theta_0'' + \phi_0'' \zeta_{i\Delta} + \sum_k \theta_k'' \sin(\xi_k i \Delta) + \sum_k \phi_k'' \cos(\xi_k i \Delta)>0,$$ with $\theta''$ and $\phi''$ given by~\eqref{def_thetaphiscd}.} The solution of the minimisation problem

\begin{equation*}
  \begin{aligned}
    \min_{\rho \in \R} \sum_{i=0}^{N-1}  &  \Bigl( (T_{(i+1)\Delta}- \mathbb{E}  [T_{(i +1) \Delta} |  \mathcal{F}_{i\Delta} ])(\zeta_{(i+1)\Delta}- \mathbb{E}  [\zeta_{(i +1) \Delta} |  \mathcal{F}_{i\Delta} ]) \\
    &- \mathbb{E}  [ (T_{(i+1)\Delta}- \mathbb{E}  [T_{(i +1) \Delta} |  \mathcal{F}_{i\Delta} ])(\zeta_{(i+1)\Delta}- \mathbb{E}  [\zeta_{(i +1) \Delta} | \zeta_{i\Delta} ]| \mathcal{F}_{i\Delta} ] ) \Bigr)^2
  \end{aligned}  
\end{equation*}

\noindent is given by

\begin{equation}\label{rho_opt}
   \hat{\rho} =  \frac{\sum_{i=0}^{N-1} Y'_{i  \Delta} (T_{(i+1)\Delta}-  \lambda^T X_{i\Delta})(\zeta_{(i+1)\Delta}- \vartheta X'_{i\Delta})}{\sum_{i=0}^{N-1} (Y'_{i  \Delta})^2},
\end{equation}
where {$X_{i\Delta}$ and $\lambda$ are defined in Proposition~\ref{lambda_def} (resp. $X'_{i\Delta}$ and $\vartheta$ in Proposition~\ref{prop-zeta})}.
\end{prop}
\noindent {Let us note that we do not know a priori that $\hat{\rho} \in [-1,1]$. }
\begin{proof}
  
We first calculate the covariance between the temperature $(T_t)_{t\geq0}$ and the volatility $(\zeta_t)_{t\geq0}$:

\begin{equation*}
   \begin{aligned}
  &\mathbb{E}  [ \left(T_{t + \Delta} - \mathbb{E}  [T_{t + \Delta} | T_{t} ] \right) \left(\zeta_{t + \Delta} - \mathbb{E}  [\zeta_{t + \Delta} | \zeta_{t} ] \right)| \mathcal{F}_{t} ] \\
  &=\mathbb{E}  \left[ \left( \int_t^{t+ \Delta} e^{-\kappa (t+ \Delta-s)}\sqrt{\zeta_s}(\rho dW_s +\sqrt{1-\rho^2} dZ_s) \right) \left(\eta \int_{t}^{t + \Delta} e^{-K (t+ \Delta -s)} \sqrt{\zeta_s} dW_s \right)| \mathcal{F}_{t} \right] \\
  %&= \rho \eta \mathbb{E}  [ \left( \int_t^{t+ \Delta} e^{-\kappa (t+ \Delta-s)}\sqrt{\zeta_s} dW_s  \right) \left( \int_{t}^{t + \Delta} e^{-K (t+ \Delta -s)} \sqrt{\zeta_s} dW_s \right)| \mathcal{F}_{t} ] \\
  &= \rho \eta \mathbb{E}  \left[ \int_t^{t+ \Delta} e^{-(\kappa+K) (t+ \Delta-s)}  \zeta_s ds| \zeta_{t} \right]= \rho \eta  \int_t^{t+ \Delta} e^{-(\kappa+K) (t+ \Delta-s)}  \mathbb{E}  [\zeta_s | \zeta_{t} ] ds,
    \end{aligned} 
\end{equation*}
by using the Ito isometry, the independence between $W$ and $Z$ and $\E[\zeta_s|\cF_t]=\E[\zeta_s|\zeta_t]$. 

\noindent From Lemma~\ref{lem_integ2}, we get by standard calculations
\begin{align*}
  \eta \int_t^{t+\Delta}e^{-(K+\kappa)(t+\Delta-s)} \E[\zeta_s|\cF_t]ds&=\theta''_0  +  \phi''_0 \zeta_t + \sum_{k=1}^{K_{\sigma^2}} \theta_k'' \sin(\xi_k s)+ \sum_{k=1}^{K_{\sigma^2}} \phi_k'' \cos( \xi_k s),
\end{align*}

\noindent with

\begin{equation}\label{def_thetaphiscd}
\begin{dcases}
   \begin{aligned}
  \theta''_0 &= \eta \gamma_0  {\left( \frac{1 - e^{- (\kappa+K) \Delta} }{\kappa+K}+ \frac{e^{-(\kappa+K)  \Delta }  - e^{-K\Delta }  }{\kappa} \right)} \\
  \phi''_0 &= \eta  e^{-K  \Delta} \frac{1 - e^{-\kappa \Delta} }{\kappa} \\
  \theta''_k&= \eta \gamma_k K \frac{K(A_k''-\phi_0'')  +  \xi_k B_k''}{K^2+\xi_k^2}- \eta \delta_k K \frac{K B_k'' -\xi_k (A_k''-\phi_0'')}{K^2+\xi_k^2} \\
  \phi''_k&=  \eta \gamma_k K \frac{KB_k''-\xi_k (A_k''-\phi_0'')}{K^2+\xi_k^2} + \eta \delta_k K \frac{K(A_k''-\phi_0'') +\xi_k B_k''}{K^2+\xi_k^2},
    \end{aligned}
\end{dcases}
\end{equation}

\noindent and

\begin{equation*}
\begin{dcases}
   \begin{aligned}
  A_k''&=\frac{(K+\kappa)(\cos(\xi_k \Delta)-e^{- (K+\kappa)\Delta}) +\xi_k \sin(\xi_k\Delta)}{(K+\kappa)^2+\xi_k^2}\\
  B_k''&=\frac{(K+\kappa)\sin(\xi_k \Delta)-\xi_k (\cos(\xi_k \Delta)-e^{- (K+\kappa)\Delta}) }{(K+\kappa)^2+\xi_k^2}.
  %\\C''&=e^{-K\Delta} \frac{1-e^{-\kappa \Delta}}{\kappa}
  \end{aligned}
\end{dcases}
\end{equation*}
\noindent {These calculations are similar to the ones of Proposition~\ref{prop-eta2}, and we get that $Y'_{i\Delta}>0$ exactly as we have obtained $Y_{i\Delta}>0$ in this proposition.} Now, let come back to the considered minimisation problem and define  {$Y_{i \Delta}$ as in Proposition~\ref{prop-eta2}}. We also consider $X_{i\Delta}=(1,i\Delta,T_{i\Delta},\sin(\xi i \Delta), \cos( \xi i \Delta))^T$ and $\lambda$ defined by~\eqref{lambda-T} (resp. $X'_{i\Delta}$ and $\vartheta$ defined by~\eqref{def_Xprime}), so that $\E[T_{(i+1)\Delta}|T_{i\Delta}]=\lambda^T X_{i\Delta}$ (resp.  $\E[\zeta_{(i+1)\Delta}|\zeta_{i\Delta}]=\vartheta^T \zeta_{i\Delta}$). 
The minimisation problem can be rewritten as follows,
\begin{mini*}|s|
{\rho}{\sum_{i=0}^{N-1} \left( (T_{(i+1)\Delta}- \lambda^T X_{i\Delta})(\zeta_{(i+1)\Delta}- \vartheta^T \zeta_{i\Delta} | \mathcal{F}_{i  \Delta} ]) - \rho Y'_{i \Delta} \right)^2}{}{ },
\end{mini*}
and the minimum is clearly given by~\eqref{rho_opt}. 
\end{proof}

%\newpage

\section{Strong {consistency} of CLS estimators for the time-dependent CIR processes} \label{appendix:conv-CIR}

We study in this appendix the strong  {consistency} of  CLS estimators of a time-dependent CIR process. This process is implemented in this paper to represent the temperature volatility dynamics. Let us consider the following process
    \begin{equation}\label{CIR_timedep}
      d\zeta_t= K ( \gamma \theta(t) -\zeta_t) dt + \eta \sqrt{\zeta_t} dW_t, {\ \zeta_0\ge 0,}
    \end{equation}
with $K,\gamma,\eta>0$  and  $\theta:\R_+\to \R_+$.  We assume that process is observed at discrete times $(\zeta_{k\Delta})_{k \in \N}$. 

The goal of this appendix is twofold. First, we prove in Theorem~\ref{thm_cv_gamma} the {consistency} of the CLS estimator of $\gamma$ when other parameters are known and give the rate of convergence. This result complements the one of Overbeck and Ryden~\cite{overbeck1997estimation} in a time inhomogeneous case. This is a simplification with respect to the estimation of~$K$, $\gamma$'s and $\delta$'s in model~\eqref{stoch} given {by Proposition~\ref{prop-zeta}: we only estimate one drift parameter instead of $2(K_{\sigma^2}+1)$ drift parameters.}  This avoids cumbersome calculations, but the same behaviour is expected for the CLS estimators of these {$2(K_{\sigma^2}+1)$} parameters. Second, we prove in Theorem~\ref{thm_cv_eta2} the  {consistency} of the CLS estimator of~$\eta^2$ when other parameters are known. This result complements the results of Bolyog and Pap~\cite{bolyog2019conditional} that only focus on the CLS estimation of the drift part.

By straightforward calculations, we have for $0\le s \le t$,
    
\begin{align}
      \zeta_t&=\zeta_s e^{-K (t-s)} + \int_s^t K \gamma \theta(u) e^{-K(t-u)}du+ \eta \int_s^t e^{-K(t-u)} \sqrt{\zeta_u}dW_u, \label{zeta_integre} \\
      \E[\zeta_t|\zeta_s]&=\zeta_s e^{-K (t-s)} + \int_s^t K \gamma \theta(u) e^{-K(t-u)}du. \notag
\end{align}

The CLS estimator of $\gamma$ consists in minimizing $\sum_{i=0}^{N-1}  \left(\zeta_{i\Delta}- \E[\zeta_{(i+1)\Delta}|\zeta_{i\Delta}]\right)^2$, i.e.
$$\sum_{i=0}^{N-1}  \left(\zeta_{(i+1)\Delta}- \zeta_{i\Delta}e^{-K\Delta}  -  K \gamma \int_{i\Delta}^{(i+1)\Delta} \theta(u) e^{-K((i+1)\Delta -u)}du  \right)^2, $$
which leads to

\begin{equation}\label{def_estim_gamma}
      \hat{\gamma}_{N,\Delta}= \frac{\sum_{i=0}^{N-1} (\zeta_{(i+1)\Delta}- \zeta_{i\Delta}e^{-K\Delta})\int_{i\Delta}^{(i+1) \Delta} \theta(u) e^{-K((i+1)\Delta -u)}du    }{K \sum_{i=0}^{N-1}  \left(\int_{i\Delta}^{(i+1)\Delta} \theta(u) e^{-K((i+1)\Delta -u)}du \right)^2}. 
\end{equation}
{In this appendix, we note $\hat{\gamma}_{N,\Delta}$ instead of $\hat{\gamma}$ to remind the dependence on $N$ and $\Delta$. This makes clearer the statements of Theorems~\ref{thm_cv_gamma} and~\ref{thm_cv_eta2} that involve these two quantities.} 

In the particular case $\theta \equiv 1$, we have
$$\hat{\gamma}_{N,\Delta}=\frac 1{N (1-e^{-K\Delta})} \sum_{i=0}^{N-1}  (\zeta_{(i+1)\Delta}- \zeta_{i\Delta}e^{-K\Delta})=\frac 1N \sum_{i=1}^{N-1} \zeta_{i\Delta} +\frac {\zeta_{N\Delta}-\zeta_0 e^{-K\Delta}}{N (1-e^{-K\Delta})}.$$
The second term is negligible and, following Overbeck and Ryden~\cite{overbeck1997estimation}, we get that the estimator $\hat{\gamma}_{N,\Delta}$ is strongly  {consistent} (i.e. $\hat{\gamma}_{N,\Delta}\to \gamma$ a.s.)  and asymptotically normal (i.e. $\sqrt{N}(\hat{\gamma}_{N,\Delta}-\gamma)$ converges in law to a normal random variable) by using the ergodic theorem.

When $\theta$ is not constant, we can no longer use the ergodic theorem. We will make the proof of  {consistency} under the assumption that $\theta$ is a bounded function. We lose the asymptotic normality but still have a convergence rate of $\sqrt{N}$. We will use the following lemma.

\begin{lemma}\label{lem_moments}
Let $\theta:\R_+\to \R_+$ be a bounded measurable function and $K>0$. Then, the process~\eqref{CIR_timedep} is well defined, nonnegative, and we have
$$\forall p>0, \  \sup_{t\ge 0} \E[\zeta_t^p]<\infty .$$ 
\end{lemma}

\begin{proof}
By using the well-known result of Yamada and Watanabe (see e.g. Karatzas and Shreve~\cite[Proposition 2.13 p. 291]{MR1121940}), there exists a {pathwise unique strong} solution to~\eqref{CIR_timedep}. From the comparison result~\cite[Proposition 2.18 p. 293]{MR1121940}, $\zeta_t$ is greater than $\tilde{\zeta}_t=\zeta_0-\int_0^t K\tilde{\zeta}_s ds +\eta \int_0^t\sqrt{\tilde{\zeta}_s}dW_s$, { since the initial values are the same and the drift of $\tilde{\zeta}$ is below the one of $\zeta$. Since $\tilde{\zeta}$} is a Cox-Ingersoll-Ross process, {it is}  nonnegative. We thus have $\zeta_t\ge \tilde{\zeta}_t\ge 0$.

Now let us turn to the moments. It is sufficient to check the result for $p\in \N^*$. For $p=1$, we have
$$\E[\zeta_t]=\zeta_0e^{-K t} + \int_0^t K \gamma \theta(u) e^{-K(t-u)}du \le \zeta_0+\gamma \bar{\theta},$$
{with $\bar{\theta}=\sup_{u\ge 0}\theta(u) <\infty$.} We then prove $\sup_{t\ge 0} \E[\zeta_t^p]<\infty$ by induction on~$p$. 

By Itô's formula, we have $d\zeta_t^p=p\zeta_t^{p-1}K(\theta(t)-\zeta_t)dt+p\eta \zeta_t^{p-1/2}dW_t +p(p-1)\frac{\eta^2}2 \zeta_t^{p-1} dt$ and thus

$$\E[\zeta_t^p]=\zeta_0^p e^{-Kp t}+ \int_0^t e^{-Kp(t-u)} pK \left(\theta(u) +\frac{p-1}{2K} \eta^2\right) \E[\zeta_t^{p-1}] du,$$
{since the stochastic integral has a zero expectation (note that the SDE~\eqref{CIR_timedep} has finite moments of any order by~\cite[Problem 3.15 p.~306]{MR1121940}).}
This leads to $\E[\zeta_t^p]\le \zeta_0^p+ \left(\bar{\theta} +\frac{p-1}{2K} \eta^2\right) \sup_{t\ge 0} \E[\zeta_t^{p-1}]$, and to the claim by induction on~$p$.  
\end{proof}

\begin{theorem} \label{thm_cv_gamma}
Let us assume that $\theta:\R_+\to \R_+$ is a bounded measurable function such that $0<\underline{\theta}\le \theta(u)<\bar{\theta}$ for some $\underline{\theta},\bar{\theta}\in \R_+^*$. Then, {for all $\Delta>0$,} the estimator $\hat{\gamma}_{N,\Delta}$ is strongly {consistent} (i.e. converges {to $\gamma$ a.s. as $N\to \infty$}) and such that $N^\alpha (\hat{\gamma}_{N,\Delta}-\gamma) \to 0$ {as $N\to \infty$} a.s. for any $\alpha \in (0,1/2)$.
\end{theorem}

\begin{proof}

From~\eqref{zeta_integre}, we get $$\zeta_{(i+1)\Delta}- \zeta_{i\Delta}e^{-K\Delta}= \gamma \int_{i\Delta}^{(i+1)\Delta} K \theta(u) e^{-K((i+1)\Delta-u)}du+ \eta \int_{i\Delta}^{(i+1)\Delta} e^{-K((i+1)\Delta-u)} \sqrt{\zeta_u}dW_u.$$
Using this in~\eqref{def_estim_gamma}, we obtain
\begin{align*}
\hat{\gamma}_{N,\Delta}&= \gamma+ \eta \frac{\sum_{i=0}^{N-1} \int_{i\Delta}^{(i+1)\Delta} e^{-K((i+1)\Delta-u)} \sqrt{\zeta_u}dW_u \int_{i\Delta}^{(i+1) \Delta} \theta(u) e^{-K({(i+1)} \Delta -u)}du    }{K \sum_{i=1}^{N{-1}} \left(\int_{i\Delta}^{(i+1)\Delta} \theta(u) e^{-K((i+1) \Delta -u)}du \right)^2}\\
&         = \gamma+ \eta \frac{\sum_{i=0}^{N-1} \Theta_i(M_{i+1}-M_{i})  }{ \sum_{i=0}^{N-1} \Theta_i^2}, 
\end{align*}
with $\Theta_i= {K} \int_{i\Delta}^{(i+1) \Delta} \theta(u) e^{-K((i+1) \Delta -u)}du$, $M_{i+1}-M_{i}=\int_{i\Delta}^{(i+1)\Delta} e^{-K((i+1)\Delta-u)} \sqrt{\zeta_u}dW_u$. %% We calculate
      %% \begin{align*}
      %%   \E[(M_{i}-M_{i-1})^2|\mathcal{F}_{(i-1)\Delta}]&=\int_{(i-1)\Delta}^{i\Delta} e^{-2K(i\Delta-u)} \E[\zeta_u|\zeta_{(i-1)\Delta}] du \\
      %%   &=\int_{(i-1)\Delta}^{i\Delta} e^{-2K(i\Delta-u)} \left(\zeta_{(i-1)\Delta}e^{-K(u-(i-1)\Delta)}+ \int_{(i-1)\Delta}^u K\gamma \theta(v) e^{-K(u-v)}dv \right) du 
      %% \end{align*}

We have $\sum_{i=0}^{N-1} \Theta_i(M_{i+1}-M_{i})= \int_0^{N\Delta}\Theta_{i(u)}e^{-K(i(u)\Delta-u)}\sqrt{\zeta_u}dW_u$. {By Burkholder-Davis-Gundy inequality and then Jensen inequality}, we get for $p\ge 2$,

\begin{align*}
    \E\left[\left|\sum_{i=0}^{N-1} \Theta_i(M_{i+1}-M_{i})\right|^p\right] &\le C_p \E\left[\left|\int_0^{N\Delta}\Theta_{i(u)}^2e^{-2K(i(u)\Delta-u)} \zeta_udu\right|^{p/2}\right] \\
    & \le C_p (N\Delta)^{p/2-1}\E\left[\int_0^{N\Delta}\Theta_{i(u)}^pe^{-p K(i(u)\Delta-u)} \zeta^{p/2}_udu\right] \\
    & \le C_ p(N\Delta)^{p/2} \bar{\theta}^p \sup_{t\ge 0} \E[\zeta_t^{p/2}],
\end{align*}
where $i(u)=i$ for $u\in [i\Delta,(i+1)\Delta]$. {Here, we have used $\Theta_i \le (1-e^{-K \Delta})\overline{\theta}\le\overline{\theta}$.}

On the other hand, we have {$\Theta_i \ge (1-e^{-K \Delta})\underline{\theta} $}, and therefore $$\E[|\epsilon_N|^p]\le {\left(\frac{1}{N(1-e^{-K \Delta})^2\underline{\theta}^2} \right)^p \times C_ p(N\Delta)^{p/2} \bar{\theta}^p \sup_{t\ge 0} \E[\zeta_t^{p/2}]  ,}$$ with $\epsilon_N=\frac{\sum_{i=0}^{N-1} \Theta_i(M_{i+1}-M_{i})  }{ \sum_{i=0}^{N-1} \Theta_i^2}$. This gives $\E[|\epsilon_N|^p]=O(N^{-p/2})$ by Lemma~\ref{lem_moments}.

Therefore, for any $\alpha \in (0,1/2)$, we can take $p>2$ such that $p(1/2-\alpha)>1$ and thus $\E[\sum_{N=1}^\infty|N^\alpha \epsilon_N |^p]<\infty$, which gives that $N^\alpha \epsilon_N \to 0$, a.s.
\end{proof}

We now turn to the  Conditional Least Squares estimation of $\eta^2$ for the process~\eqref{CIR_timedep}. We now assume that $K,\gamma>0$ and $\theta(\cdot)$ are known. Without loss of generality, we assume that $\gamma=1$ and consider the minimization problem of  
$$\sum_{i=0}^{N-1}\left[\left(\zeta_{(i+1)\Delta}-\E[\zeta_{(i+1)\Delta}|\cF_{i\Delta}] \right)^2- \E\left[\left(\zeta_{(i+1)\Delta}-\E[\zeta_{(i+1)\Delta}|\cF_{i\Delta}] \right)^2|\cF_{i\Delta}\right]\right]^2,$$
with respect to $\eta^2$. By using Equation~\eqref{zeta_integre} and {Fubini theorem, we get}
\begin{align}
    &\E\left[\left(\zeta_{(i+1)\Delta}-\E[\zeta_{(i+1)\Delta}|\cF_{i\Delta}]\right)^2|\cF_{i\Delta}\right]=\eta^2\int_{i\Delta}^{(i+1)\Delta} e^{-2K((i+1)\Delta-u)} \E[\zeta_u{|}\cF_{i\Delta}] du \notag \\
    & {=\eta^2 \left(\int_{i\Delta}^{(i+1)\Delta} e^{-2K((i+1)\Delta-u)} \zeta_{i\Delta}e^{-K(u-i\Delta)} du + \int_{i\Delta<v<u<(i+1)\Delta} e^{-2K((i+1)\Delta-u)} \theta(v)e^{-K(u-v)} dudv \right)}\notag\\
   & =\eta^2 \left(  e^{-K\Delta}\frac{1-e^{-K\Delta}}{K}\zeta_{i\Delta}+\int_{i\Delta}^{(i+1)\Delta} \theta(v) e^{-K((i+1)\Delta-v)}(1-e^{-K((i+1)\Delta-v)}) dv\right).\label{calc_esp_cond_sq}
\end{align}
{The minimization of {\footnotesize
$$ \sum_{i=0}^{N-1}\left[\left(\zeta_{(i+1)\Delta}-\E[\zeta_{(i+1)\Delta}|\cF_{i\Delta}] \right)^2- \eta^2 \left(  e^{-K\Delta}\frac{1-e^{-K\Delta}}{K}\zeta_{i\Delta}+\int_{i\Delta}^{(i+1)\Delta} \theta(v) e^{-K((i+1)\Delta-v)}(1-e^{-K((i+1)\Delta-v)}) dv\right)\right]^2 $$} then} leads to the following estimator
\begin{equation}
    \widehat{\eta^2}_{\Delta,N}=\frac{\sum_{i=0}^{N-1}  \left(\zeta_{(i+1)\Delta}-(\zeta_{i\Delta}e^{-K\Delta}+\Theta^1_i)\right)^2 \left(  e^{-K\Delta}\frac{1-e^{-K\Delta}}{K}\zeta_{i\Delta} +\Theta^2_i\right) }{\sum_{i=0}^{N-1}\left(  e^{-K\Delta}\frac{1-e^{-K\Delta}}{K}\zeta_{i\Delta} +\Theta^2_i\right)^2},
\end{equation}
with $\Theta^1_i={K \gamma} \int_{i\Delta}^{(i+1)\Delta} \theta(v) e^{-K((i+1)\Delta-v)}dv$ and $\Theta^2_i=\int_{i\Delta}^{(i+1)\Delta} \theta(v) e^{-K((i+1)\Delta-v)}(1-e^{-K((i+1)\Delta-v)}) dv$. 

\begin{theorem} \label{thm_cv_eta2}
Let us assume $\gamma=1$ and that $\theta:\R_+\to \R_+$ is a bounded measurable function such that $0<\underline{\theta}\le \theta(u)<\bar{\theta}$ for some $\underline{\theta},\bar{\theta}\in \R_+^*$. Then, {for all $\Delta>0$,} the estimator $\widehat{\eta^2}_{N,\Delta}$ is strongly  {consistent} (i.e. converges to $\eta^2$ a.s.  as $N\to + \infty$) and such that $N^\alpha (\widehat{\eta^2}_{N,\Delta}-\eta^2) \to 0$ a.s. for any $\alpha \in (0,1/2)$.
\end{theorem}
\begin{proof}
The proof follows the same arguments as the one of Theorem~\ref{thm_cv_gamma}, and we give the main lines. We have $\zeta_{(i+1)\Delta}-(\zeta_{i\Delta}e^{-K\Delta}+\Theta^1_i)=\eta \int_{i\Delta}^{(i+1)\Delta} e^{-K((i+1)\Delta-u)}\sqrt{\zeta_u}dW_u $ by~\eqref{zeta_integre} and
define $a_i=e^{-K\Delta}\frac{1-e^{-K\Delta}}{K}\zeta_{i\Delta} +\Theta^2_i$ which is nonnegative. We can rewrite
$$ \widehat{\eta^2}_{\Delta,N}=\eta^2 +\eta^2 \frac{\sum_{i=0}^{N-1}   a_i \left[\left( \int_{i\Delta}^{(i+1)\Delta} e^{-K((i+1)\Delta-u)}\sqrt{\zeta_u}dW_u  \right)^2- a_i \right] }{\sum_{i=0}^{N-1} a_i ^2}$$
We set $M_0=0$ and $M_{i+1}-M_i=\left( \int_{i\Delta}^{(i+1)\Delta} e^{-K((i+1)\Delta-u)}\sqrt{\zeta_u}dW_u  \right)^2- a_i $ for $i \in \N$. The process $M$ is a $\cF_{i\Delta}$-martingale by~\eqref{calc_esp_cond_sq}, and since $a_i$ is $\cF_{i\Delta}$-adapted, $\sum_{i=0}^{N-1}a_i(M_{i+1}-M_i)$ is also a martingale. Applying  Burkholder-Davis-Gundy inequality and then Jensen inequality, we get for $p\ge 2$
\begin{align*}
\E\left[ \left|\sum_{i=0}^{N-1}a_i(M_{i+1}-M_i)\right|^p\right ]  \le C_p \E\left[ \left(\sum_{i=0}^{N-1}a^2_i(M_{i+1}-M_i)^2\right)^{p/2}\right ]\le C_p N^{p/2-1}  \E\left[ \sum_{i=0}^{N-1}a^p_i|M_{i+1}-M_i|^p\right ].  
\end{align*}
Now, we check easily from Lemma~\ref{lem_moments} that $\E[a_i^p|M_{i+1}-M_i|^p] \le C'_p<\infty$ for all $i$, and thus $\E\left[ \left|\sum_{i=0}^{N-1}a_i(M_{i+1}-M_i)\right|^p\right ] =O(N^{p/2})$.

On the other hand, we have $a_i\ge \Theta^2_i \ge \underline{\theta} {\frac{(1-e^{-K\Delta})^2}{2K}}$ and thus $\sum_{i=0}^{N-1}a_i^2\ge N {\underline{\theta}^2\left(\frac{(1-e^{-K\Delta})^2}{2K}\right)^2 }$. Setting $\epsilon_N=\frac{\sum_{i=0}^{N-1} a_i(M_{i+1}-M_i)}{\sum_{i=0}^{N-1} a_i^2}$, we get $\E[|\epsilon_N|^p]=O(N^{-p/2})$, and we conclude as in the proof of Theorem~\ref{thm_cv_gamma}.
\end{proof}

\end{document}